\documentclass[hyperfootnotes=false,letterpaper]{article}
\usepackage[english]{babel}
\usepackage{cite}

\usepackage[utf8]{inputenc}
\usepackage{fullpage}
\usepackage{amsmath, amsfonts, amssymb, amsthm}
\usepackage[linesnumbered,noresetcount,vlined,ruled,procnumbered]{algorithm2e}

\usepackage{color}
\usepackage{xcolor}
\definecolor{darkgreen}{rgb}{0,0.5,0}
\definecolor{darkblue}{rgb}{0,0,0.6}
\definecolor{DarkRed}{rgb}{0.65,0,0}
\definecolor{Red}{rgb}{1,0,0}

\usepackage[linktocpage=true,backref=page,
pagebackref=true,colorlinks, linkcolor=darkblue,citecolor=darkgreen,urlcolor=blue,filecolor=magenta,bookmarksnumbered,hypertexnames=false]{hyperref}

\usepackage[capitalize, nameinlink]{cleveref}
\usepackage{graphicx}
\usepackage{comment}
\usepackage{thm-restate}
\usepackage{bbm}
\usepackage{tikz}
\usetikzlibrary{shapes.geometric, arrows.meta, positioning, bending,calc}
\usetikzlibrary{fit}
\usetikzlibrary{patterns, patterns.meta}
\usepackage{pgfplots}
\usepgfplotslibrary{fillbetween} 
\pgfplotsset{compat=1.16} 
\usetikzlibrary{patterns}

\usepackage[linewidth=1pt]{mdframed}
\usepackage[most]{tcolorbox}
\definecolor{block-gray}{gray}{0.85}
\newtcolorbox{shadequote}{colback=block-gray,grow to right by=-2mm,grow to left by=-2mm,
boxrule=0pt,boxsep=0pt,breakable}
\usepackage{color}
\usepackage[inline]{enumitem}

\usepackage{caption}
\usepackage{cancel}
\usepackage{subcaption}
\usepackage{mathrsfs}
\usepackage{xspace}

\usepackage{mathtools} 
\usepackage{booktabs} 
\usepackage{makecell}
\usepackage{tabularx}
\usepackage{lmodern}

\SetKwComment{Comment}{$\quad\triangleright$\ }{}
\newcommand{\nosemic}{\renewcommand{\@endalgocfline}{\relax}}

\usepackage[normalem]{ulem} 

\crefalias{AlgoLine}{line}%

\makeatletter
\let\cref@old@stepcounter\stepcounter
\def\stepcounter#1{%
  \cref@old@stepcounter{#1}%
  \cref@constructprefix{#1}{\cref@result}%
  \@ifundefined{cref@#1@alias}%
    {\def\@tempa{#1}}%
    {\def\@tempa{\csname cref@#1@alias\endcsname}}%
  \protected@edef\cref@currentlabel{%
    [\@tempa][\arabic{#1}][\cref@result]%
    \csname p@#1\endcsname\csname the#1\endcsname}}
\makeatother

\usepackage{nicefrac}
\usepackage{pgf}
\usepackage{environ} 

\usepackage{float}
\usepackage{tocloft}
\usepackage{tocvsec2}

\newtheorem{theorem}{Theorem}[section]
\newtheorem{lemma}[theorem]{Lemma}
\newtheorem{remark}{Remark}
\newtheorem{definition}{Definition}
\newtheorem{observation}[theorem]{Observation}
\newtheorem{proposition}[theorem]{Proposition}
\newtheorem{claim}[theorem]{Claim}
\newtheorem{corollary}[theorem]{Corollary}

\makeatletter
\renewenvironment{proof}[1][\proofname]{\par
    \pushQED{\qed}%
    \normalfont \topsep6\p@\@plus6\p@\relax
    \trivlist
    \item\relax
    {\bfseries\boldmath
        #1\@addpunct{.}}\hspace\labelsep\ignorespaces
}{%
    \popQED\endtrivlist\@endpefalse
}
\makeatother

\DeclareMathSymbol{\mhyphen}{\mathord}{AMSa}{"39}

\newcommand{\wc}{w_{\mathrm{heavy}}}
\newcommand{\apm}{(1\pm\eps)}

\newcommand{\MM}[1]{\mathsf{MM}(#1)}

\newcommand{\brak}[1]{\left(#1\right)}
\newcommand{\Exp}[1]{\mathbb{E}\left[ #1 \right]}
\renewcommand{\Pr}[1]{{\mathrm{Pr}}\left[ #1 \right]}
\newcommand{\set}[1]{\left\{ #1 \right\}}
\newcommand{\ch}[2]{{#1 \choose #2}}

\newcommand{\mo}{O}
\newcommand{\TO}[1]{\tilde{\mo}(#1)}
\newcommand{\To}{\tilde{O}}
\newcommand{\tg}{\tilde{\Omega}}
\newcommand{\tf}{\tilde{O}}
\newcommand*{\boldone}{\text{\usefont{U}{bbold}{m}{n}1}}

\newcommand{\eps}{\varepsilon}
\newcommand{\EE}{\mathcal{E}}
\newcommand{\UU}{\mathcal{U}}
\newcommand{\HH}{\mathcal{H}}
\newcommand{\PP}{\mathcal{P}}
\newcommand{\XX}{\mathcal{X}}
\newcommand{\AC}{\mathcal{A}}
\newcommand{\FF}{\mathcal{F}}

\makeatletter
\newcommand{\whp}{%
  w.h.p.\@ifnextchar.{\@gobble}{\xspace}%
}
\makeatother

\newcommand{\yes}{\texttt{Yes}\xspace}
\newcommand{\no}{\texttt{No}\xspace}

\DeclarePairedDelimiter\abs{\lvert}{\rvert}%

\makeatletter
\let\oldabs\abs
\def\abs{\@ifstar{\oldabs}{\oldabs*}}

\crefname{claim}{Claim}{Claims} 
\crefname{ineq}{inequality}{inequalities} 
\crefname{proof}{Proof}{Proofs} 
\crefname{line}{Line}{Lines}
\crefname{algorithm}{Algorithm}{Algorithms}
\crefname{observation}{Observation}{Observation}
\crefname{black box}{Procedure}{Procedure}
\crefname{figure}{Figure}{Figures}
\crefname{hope}{Hope}{Hopes}

\newcommand{\DD}{\mathcal{D}}
\newcommand{\Q}{\log^2(\abs{V(G)})\cdot 2^{2h+8}}
\newcommand{\cg}{\brak{2\log \Lambda}^h}
\newcommand{\eo}{1-1/e}
\newcommand{\nemp}{\neq\emptyset}

\newcounter{blackbox}[section] 

\crefname{blackbox}{Procedure}{Procedure}

\newcommand{\omegaval}{2.371552}
\newcommand{\CC}{\mathcal{C}}
\newcommand{\poly}{\mathsf{poly}}
\newcommand{\QQ}{\mathcal{Q}}

\newcommand{\GG}{\mathcal{G}}
\newcommand{\Gp}{G^\prime}
\newcommand{\cclog}{(k\log (4n))^{k^2}}
\newcommand{\W}{L}
\newcommand{\Vp}{S}

\newcommand{\algA}{\textsc{GetZ}}
\newcommand{\algCD}{{\normalfont\textsc{DetectCID}}\xspace}
\newcommand{\algSensitive}{{\normalfont\textsc{SensitiveDetection}}\xspace}
\newcommand{\algXlight}{{\normalfont\textsc{$x$-LightDetection}}\xspace}
\newcommand{\algVertexOne}{{\normalfont\textsc{IsvInDiamond}}\xspace}
\newcommand{\algVertex}{{\normalfont\textsc{FindVertexInDiamond}}\xspace}
\newcommand{\algVertexC}{{\normalfont\textsc{FindDeg3Vertex}}\xspace}
\newcommand{\algPart}{{\normalfont\textsc{FindPartition}}\xspace}
\newcommand{\algCluster}{{\normalfont\textsc{FindClustering}}\xspace}
\newcommand{\algClusterVeri}{{\normalfont\textsc{VerifyClustering}}\xspace}

\newcommand{\degb}{{\normalfont\textsc{Deg2}}\xspace}
\newcommand{\degc}{{\normalfont\textsc{Deg3}}\xspace}

\renewcommand{\phi}{\varphi}

\newcommand{\rmax}{r_{\max}}

\newcommand{\dsum}{{\normalfont4-\textsf{SUM}}\xspace}
\newcommand{\ksum}{{\normalfont\textsf{k-SUM}}\xspace}
\newcommand{\runtimex}{\rmax\cdot(n/\rmax)^\omega +(n/\rmax)^3+n\cdot \rmax}
\newcommand{\runtime}{r\cdot(n/r)^\omega +(n/r)^3+n\cdot r}
\newcommand{\runtimeT}{n^{2.425}/t^{0.25}}
\newcommand{\pat}{F}

\author{Keren Censor-Hillel\thanks{Department of Computer Science, Technion. The research is supported in part by the Israel Science Foundation (grant 529/23)} 
\and Tomer Even\thanks{Department of Computer Science, Technion.}
\and Virginia Vasillevska Williams\thanks{Massachusetts Institute of Technology, Cambridge, MA, USA. Supported by NSF Grant CCF-2330048, BSF Grant 2024233, and a Simons Investigator Award.}
\and Nathan Wallheimer\thanks{Weizmann Institute of Science.}
}

\title{Witness-Sensitive Detection of Induced Diamonds}
\date{}
\begin{document}
\maketitle
\begin{abstract}
	We provide a fast \emph{witness-sensitive} algorithm for detecting an induced diamond (a $K_4$ minus an edge) in an $n$-vertex graph containing $t$ induced diamonds. Our algorithm runs in time $\tilde{O}(\min(n^{2.425}/t^{0.25}+n^2, n^\omega))$ with high probability, improving upon the prior state of the art (witness-oblivious) algorithm that runs in time $O(n^\omega\log{n})$ [Vassilevska Williams, Wang, Williams, Yu, SODA 2014] whenever $t \geq n^{(3-\omega)/3}$, where $\omega < 2.372$ is the matrix multiplication exponent.

	Our key insight is that the size of a clique containing one of the triangles of an induced diamond plays a crucial role in detecting such a diamond. We say that a diamond is $r$-heavy if this size is at least $r$, and we provide a fast detection algorithm for $r$-heavy diamonds in $\tilde{O}(r \cdot (n/r)^\omega + (n/r)^3+ nr)$ time.
	When there are no $r$-heavy diamonds, we provide a different fast detection algorithm in $\tilde{O}(\mathsf{MM}(n,n,n\sqrt{r/t}))$ time, where $\mathsf{MM}(a,b,c)$ denotes the time to multiply an $a \times b$ matrix by a $b \times c$ matrix, which is conditionally optimal for $r=\tilde{O}(1)$.

	Our main technical contribution is in designing a refinement framework for sampling vectors, which allows sampling vertices for detecting diamonds in a manner that is adaptive to the structure of graphs with no $r$-heavy diamonds.
	We establish that our technique is of a wide applicability, by showing how it also allows for faster witness-sensitive algorithms for $4$-SUM and for a special case of $4$-cycles.
\end{abstract}

\setlength{\cftbeforesecskip}{3pt}
\setlength{\cftbeforesubsecskip}{1pt}
\tableofcontents

\section{Introduction}
The problem of detecting a fixed subgraph $\pat$ within a host graph $G$ is a cornerstone task of theoretical computer science, with broad applications ranging from social network analysis~\cite{Kuramochi2001FrequentSD,Liu2019NeuralSI,tsourakakis2015k} and computational biology~\cite{Alon2008BiomolecularNM,dao2011optimally,rahman2009small} to machine learning~\cite{sun2021sugar,Bouritsas2020ImprovingGN,BordesCW14}.

While detecting a pattern $\pat$ on $h$ vertices is easily solvable in $O(h^2 \cdot n^h)$ time, extensive work has been devoted to obtaining faster algorithms.
In some cases, when a barrier prevents improving the state-of-the-art worst-case complexity, \emph{witness-sensitive} algorithms become a significant paradigm: these are algorithms that do not improve upon the general case, e.g., if an input instance has only a single copy of $\pat$ hidden in it, but they do run much faster on instances that have more copies of $\pat$.
Distinguishing between graphs that contain $t$ copies of $\pat$ and $\pat$-free graphs has been studied in various models, such as property testing \cite{goldreich1998property,alon2000efficient,Goldreich2002-rm,Alon2015EasilyTG,Gishboliner2013DeterministicVN}, sublinear algorithms \cite{parnas2002testing,kaufman2004tight,aliakbarpour2018sublinear,eden2017approximately,assadi2018simple,Fichtenberger2020SamplingAS,Biswas2021TowardsAD,AssadiKK19}, streaming \cite{ahn2012graph,mcgregor2016better,bera2017towards,kallaugher2019complexity}, and more.

In the standard word-RAM model, the baseline approach for exploiting multiple copies is uniform sampling. This is because if we
sample $h$ vertices and check if they induce $\pat$, then in expectation $O(n^h/t)$ samples are sufficient.
For triangles, this approach was pushed further by
~\cite{tvetek2022approximate}, who
improved this running time to $\tf(t^2 \cdot (n/t)^\omega)$ by reducing the problem to multiplying $\TO{t^2}$ matrices of size $n/t \times n/t$, also providing a $(1\pm \eps)$ approximation for $t$. Here, $\omega<2.372$ is the matrix multiplication exponent~\cite{moreasym}.
This was further improved by~\cite{CEW24} to $\MM{n,n,n/t}$ time, which is the time to multiply an $n \times n$ matrix by an $n \times (n/t)$ matrix. Improving upon this complexity is shown in~\cite{CEW24} to hit the barrier of a conditional lower bound.
Witness-sensitive algorithms go beyond triangles: In~\cite{CEW24}, the approach for triangles is shown to generalize, providing a $(1\pm \eps)$ approximation for the number of $k$-cycles in $\MM{n,n,n/t^{1/(k-2)}}$ time. In~\cite{CEW25}, witness-sensitive algorithms are given for $k$-clique detection, $k$-sum, and more.

In this work, we address the complexity of detecting an induced diamond, which is a $4$-clique minus an edge and is one of the simplest non-trivial patterns beyond cycles and cliques
\cite{kloks2000finding,eisenbrand2004complexity,williams2014finding}.
The state of the art for diamond detection runs in $O(n^\omega \cdot \log n)$ time~\cite{williams2014finding}, and a faster algorithm would imply faster triangle detection, which is a long-standing open problem~\cite{itai1977finding,nesetril,focsy}.

One can directly use the known approaches to obtain witness-sensitive induced diamond detection: Na\"{\i}ve sampling would give $\tf(n^4/t)$ time, and a reduction to $4$-clique detection would give $\tf(\MM{n^2,n,n}/\sqrt{t})$ time~\cite{CEW25}.
We ask:
\begin{center}
    \emph{Is there a faster witness-sensitive algorithm for induced diamond detection?}
\end{center}
We answer this question affirmatively by presenting an algorithm that runs faster as the number of induced diamonds increases. Our algorithm improves upon the state of the art already for $t\geq n^{(3-\omega)/3}$.

Our key insight is to categorize induced diamonds based on their \emph{$r$-heaviness}, a new notion that captures whether three of the four vertices of the induced diamond are part of a clique of size $r$.
By designing different algorithms for detecting \emph{$r$-heavy} induced diamonds and for detecting \emph{$r$-light} ones, we are able to obtain our improvement.

Prior witness-sensitive algorithms define sampling vectors that are used to sample vertices from a $k$-partite graph derived from the input. The main technical novelty in our approach lies in a \emph{refinement} framework for these vectors, allowing us to sample in a way that leverages the structure of a graph with only $r$-light diamonds. 

Our refinement technique is not limited to induced diamonds, but rather applies to other patterns as well: we obtain faster witness-sensitive algorithms for \dsum and for a special case of induced $4$-cycles.
\begin{figure}[h!]
    \center

\begin{tikzpicture}
    \tikzfading[name=fadeUp,   bottom color=transparent!100, top color=transparent!20]
    \tikzfading[name=fadeDown, bottom color=transparent!20,   top color=transparent!100]
    \def\omegaval{2.371552}

    \begin{axis}[
            axis lines=middle,
            xlabel={$\log_n(t)=\tau$},
            ylabel={$\log_n(\mathsf{Time})$},
            xlabel style={at={(ticklabel* cs:1.3)}, anchor=east},
            ylabel style={at={(ticklabel* cs:1.15)}, anchor=north},
            xmin=0, xmax=2,
            ymin=1.55, ymax=3.15,
            grid=both,
            grid style={line width=.1pt, draw=gray!10},
            major grid style={line width=.2pt, draw=gray!50},
            minor x tick num=5,
            minor y tick num=4,
            ytick={2.0,2.25,2.5,2.75,3.0},
            legend pos=outer north east,
            legend style={anchor=north west, at={(1.02,1)}, cells={anchor=west}},
            extra y ticks={\omegaval},
            extra y tick labels={$\omega$},
            extra y tick style={
                    major tick length=6pt, tick align=outside, tick pos=left, thick,
                    /pgfplots/major grid style={draw=none}
                }
        ]

        \def\colA{cyan}
        \def\colB{red!70}
        \def\colC{green!50!black}
        \def\colOpt{violet}
        \def\colS{red!70!black}
        \def\markRep{8}
        \def\tableDir{figures/omega_data_tx9.dat}
        \def\aup{\uparrow}
        \def\adown{\downarrow}

        \pgfmathsetmacro{\tauzero}{(3-\omegaval)/3}


        \addplot[name path=a_low,  \colA, forget plot] table[x=tau, y=f1_rho0] {\tableDir};
        \addplot[name path=a_high, \colA, forget plot] table[x=tau, y=f1_rho_tau] {\tableDir};

        \addplot[name path=b_low,  \colB, forget plot] table[x=tau, y=f2_rho_tau] {\tableDir};
        \addplot[name path=b_high, \colB, forget plot] table[x=tau, y=f2_rho0] {\tableDir};

        \addplot[\colOpt, thick, mark=o, mark repeat=2\markRep, forget plot]
            table[x=tau, y expr={\thisrow{g_opt}<\omegaval ? \thisrow{g_opt} : \omegaval}] {\tableDir};

        \addplot[name path=cond_lb, draw=\colC, line width=0.9pt,
                 dash pattern=on 1pt off 2pt, line cap=round, forget plot]
            table[x=tau, y expr=\thisrow{f1_rho0}-0.005] {\tableDir};

        \addplot[name path=simpler, draw=\colS, line width=0.9pt, dash pattern=on 3pt off 2pt, line cap=round, forget plot,
                 domain=0:2, samples=200]
            {min(max(2.424-0.248*x,2),\omegaval)};


        \addplot[draw=none, fill=\colB, fill opacity=0.06, forget plot] fill between[of=b_low and b_high];
        \addplot[draw=none, fill=\colB, fill opacity=0.8, path fading=fadeDown, forget plot] fill between[of=b_low and b_high];

        \addplot[draw=none, fill=\colA, fill opacity=0.06, forget plot] fill between[of=a_low and a_high];
        \addplot[draw=none, fill=\colA, fill opacity=0.8, path fading=fadeUp, forget plot] fill between[of=a_low and a_high];

        \draw[red, thin, dashed]
            (axis cs:\tauzero,\pgfkeysvalueof{/pgfplots/ymin}) --
            (axis cs:\tauzero,\pgfkeysvalueof{/pgfplots/ymax})
            node[pos=0.91, anchor=south west, font=\small]
            {$\tau_0=\tfrac{3-\omega}{3}$};


        \addlegendimage{area legend, draw=\colB, fill=\colB, fill opacity=0.25}
        \addlegendentry{\Cref{thm:r-heavy diamonds}: $\runtimex$ }

        \addlegendimage{area legend, draw=\colA, fill=\colA, fill opacity=0.25}
        \addlegendentry{\Cref{thm:r-light diamonds}: $\MM{n,n,n\cdot\sqrt{\rmax/t}}$ }

        \addlegendimage{\colOpt, thick, mark=o}
        \addlegendentry{Best of both}

        \addlegendimage{draw=\colS, line width=0.9pt, dash pattern=on 3pt off 2pt, line cap=round}
        \addlegendentry{\Cref{thm:simpler}: $\min(\runtimeT + n^2,n^\omega)$}

        \addlegendimage{draw=\colC, line width=0.9pt, dash pattern=on 1pt off 2pt, line cap=round}
        \addlegendentry{\Cref{thm:sens-lb}: $\Omega(\MM{n,n,n/{\sqrt{t}}})$ (Conditional)}

    \end{axis}

\end{tikzpicture}
    \caption{Running time of induced diamond detection as a function of $t=n^\tau$ and $\rmax=n^{\rho}$.
    The cyan band shows the runtime of \Cref{thm:r-light diamonds}: $\tf(\MM{n,n,n\cdot \sqrt{\rmax/t}})$ as $\rmax$ varies from $3$ (lower boundary) to $t$ (upper boundary, equal to $\omega$); darker blue indicates larger $\rmax$.
    The red band shows the runtime of \Cref{thm:r-heavy diamonds}: $\tf(\rmax\cdot (\frac{n}{\rmax})^\omega +(\frac{n}{\rmax})^3 + n^2)$; darker red indicates larger $\rmax$.
    The violet curve shows our algorithm's running time as a function of $t$ in the worst case over all possible values of $\rmax$. For each fixed $t$, we consider all possible graphs with different values of $\rmax$, apply whichever of the two algorithms is faster for that $\rmax$, and the violet curve represents the maximum such runtime over all choices of $\rmax$.
    The vertical dashed line marks $\tau_0=(3-\omega)/3$, where our running time first drops below $n^\omega$, the state-of-the-art non-sensitive bound \cite{williams2014finding}.
    The red dashed curve shows a looser upper bound given by \Cref{thm:simpler}.   
    Finally, the lower boundary of the cyan band (highlighted in dashed green) gives a conditional lower bound under the Unbalanced Triangle Detection Hypothesis (\Cref{thm:sens-lb}). 
    }
    \label[figure]{fig5:comp}
\end{figure}
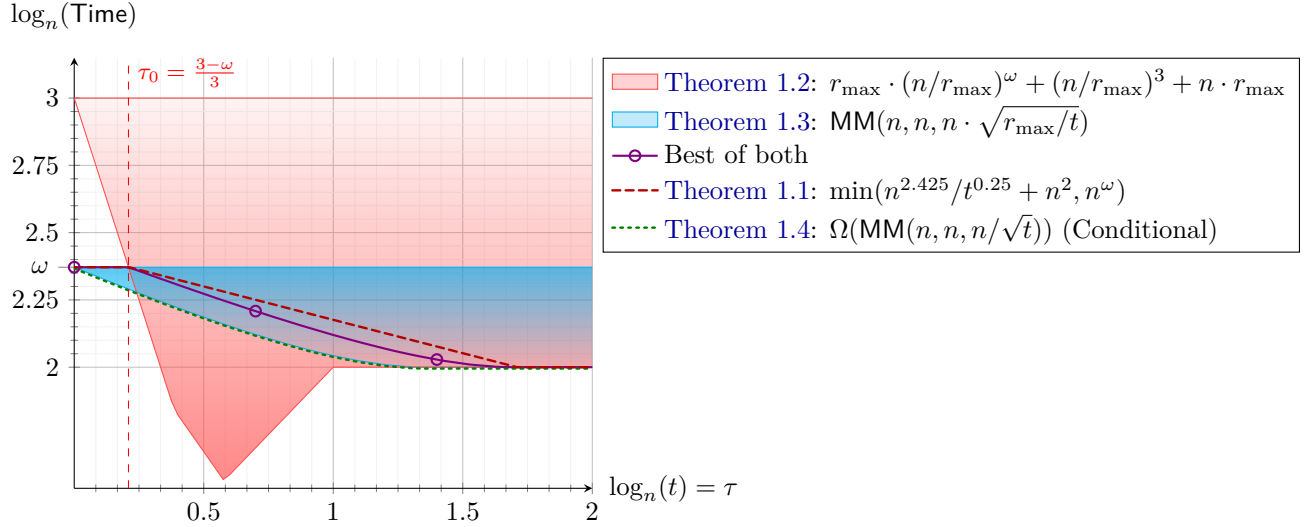

\subsection{Our Contribution}

Our main result is a witness-sensitive algorithm for detecting an induced diamond in $G$ that runs faster as the number of induced diamonds increases. Here, $G$ is a graph with $n$ vertices and $t$ induced diamonds.
\begin{theorem}[Simplified]
    \label{thm:simpler}
    There is a randomized algorithm that, given a graph $G$ with $n$ vertices and $t$ induced diamonds, finds an induced diamond in time
    $\To(\min(\runtimeT + n^2,n^\omega))$.
\end{theorem}
For $t\geq n^{(3-\omega)/3}$, our algorithm improves upon the prior running time $O(n^\omega\cdot \log n)$~\cite{williams2014finding}.
To achieve this, we design two different algorithms to detect $r$-heavy and $r$-light induced diamonds, and then combine them.

We say that an induced diamond is $r$-heavy if three of its vertices (that form a triangle) are contained in a clique of size at least $r$. Otherwise, we say that the induced diamond is $r$-light.
We use $\rmax$ to denote the maximum integer $r$ such that there exists at least one $r$-heavy induced diamond in $G$.
Our algorithm for $r$-heavy induced diamonds is as follows:
\begin{restatable}[$r$-Heavy Diamonds]{theorem}{ThmHeavyD}\label{thm:r-heavy diamonds}
    There is a randomized algorithm that finds an induced diamond in time $\tf(\runtimex)$ \whp.
    
\end{restatable}
Notably, this is subquadratic in $n$ for $n^{1/3}\ll \rmax \ll n$.

Our algorithm for $r$-light induced diamonds is as follows:
\begin{restatable}[$r$-Light Diamonds]{theorem}{ThmLightD}\label{thm:r-light diamonds}
    There exists an algorithm that given an $n$-vertex graph $G$, detects an induced diamond in $G$ \whp, running in time $\To(\MM{n,n,n \cdot \sqrt{\rmax/t}})$.
\end{restatable}

To obtain \Cref{thm:simpler}, we combine \Cref{thm:r-light diamonds,thm:r-heavy diamonds}.
Applying the standard approximation $\MM{n,n,np}\leq n^\omega\cdot p^\beta + n^2$ for $p\in[1/n,1]$, where $\beta\approx 0.5475$~\cite{huang1998fast,zwick2002all}, to \Cref{thm:r-light diamonds} yields the simpler expression $\tf(\runtimeT + n^2)$.
For values of $t$ where this bound exceeds $n^\omega$, we instead use the $O(n^\omega\cdot \log n)$ time algorithm of \cite{williams2014finding}.
\Cref{fig5:comp} illustrates the running times of our algorithms for various values of $t$ and $\rmax$.

We also show that for $\rmax =\tf(1)$, our algorithm from \cref{thm:r-light diamonds} is conditionally optimal, following the same construction as in~\cite{CEW24} for triangles.
\begin{restatable}[Sensitive Lower Bound]{theorem}{SensLB}
    \label{thm:sens-lb}
    For every $0 \leq t \leq n^2/5$, every randomized algorithm that finds an induced diamond in an $n$-vertex graph that has at least $t$ diamonds requires $\MM{n,n/\sqrt{t},n}/n^{o(1)}$ time under the Unbalanced Triangle Detection Hypothesis. 
    The lower bound holds even for graphs with no $4$-cliques. 
\end{restatable}

A consequence of this theorem, which we find interesting, is that the parameter $t$ separates the complexities of triangle detection and diamond detection even for small values of $t$, assuming $\omega > 2$. 
Previously, it was known that diamond detection and triangle detection belong to the same complexity classes when measured in terms of $n$ or the number of edges $m$~\cite{williams2014finding}. 
However, when considering the number of witnesses $t$,~\cite{CEW24} showed that witness-sensitive triangle detection can be done in $\To( \MM{n,n,n/t} )$ time for every $t \leq n$. In contrast, our theorem implies a higher lower bound of $\tilde\Omega( \MM{n,n,n/\sqrt{t}} )$ for diamond detection when $t \leq n$. 

\subsection{Technical Overview}
\label{sub:technical-overview}

There are essentially two existing approaches for diamond detection. We review them to determine their witness sensitivity:
\begin{enumerate}[label=(\alph*)]
	\item \textbf{Neighborhood structural analysis.} The approaches of~\cite{kloks2000finding,eisenbrand2004complexity} and the algorithm of\\~\cite[Theorem 5.1]{williams2014finding} exploit a structural property that is specific to induced diamonds: a graph is diamond-free if and only if every vertex neighborhood is a $P_3$-free graph (i.e., a disjoint union of cliques). This includes an $O(m^{3/2})$-time algorithm~\cite{eisenbrand2004complexity} and an $\tilde{O}(\min\{n^{\omega},m^{2\omega/(\omega+1)}\})$-time algorithm~\cite{williams2014finding}. Notably, both algorithms operate by searching for a \degc vertex, namely a vertex whose neighborhood contains an induced $P_3$. Moreover, unlike the second approach, these algorithms are deterministic.

    \item \textbf{Algebraic substructure counting.} The algorithm of~\cite[Theorem 1.1]{williams2014finding} operates by computing counts of smaller structures and relating them to the diamond count. It utilizes polynomial identity testing to distinguish between zero and non-zero counts. This approach is versatile, applying to all induced subgraphs on 4 vertices, and results in a randomized running time of $\tilde{O}(\min\{n^{\omega},m^{2\omega/(\omega+1)}\})$.
\end{enumerate}

\paragraph{Na\"{\i}ve attempt: Making known detection algorithms witness-sensitive.} To construct a witness-sensitive algorithm, a straightforward strategy would be to sample a small number of vertices uniformly at random and check if any of them participates in a diamond.
However, the existing algorithms lack an efficient procedure for checking if a small subset of vertices is incident to a diamond.
Instead, they rely on identifying a \degc vertex. Therefore, our first attempt modifies the sampling strategy: we sample a small set of vertices and specifically check if any of them participates in a diamond \emph{as a \degc vertex}.

This strategy is feasible because the algorithm of~\cite[Theorem 5.1]{williams2014finding} can be adapted to efficiently check if a subset of vertices $S \subseteq V(G)$ contains a \degc vertex.
This check runs in time $\MM{n,n,|S|}$. By combining this with a hitting set argument, we derive an algorithm sensitive to the number of \degc vertices, denoted as $x_3$.

\paragraph{Challenge 1:} The first issue is that the above algorithm is not truly witness-sensitive with respect to the total number of diamonds $t$. It is possible for $t$ to be very large while $x_3$ remains constant, resulting in no speedup when $t$ is large.
This phenomenon occurs when diamonds cluster heavily around the same diagonal edge. Consider a diagonal edge $(v_1,v_2)$ that forms triangles with a set of independent vertices $W = \{w_1, w_2, \ldots, w_d\}$.
Every pair of non-adjacent vertices in $W$ creates a diamond with $(v_1,v_2)$ as the diagonal edge. The number of diamonds is $\binom{d}{2}$, yet $x_3 = 2$ as only $v_1$ and $v_2$ are \degc vertices. Thus, an algorithm that is sensitive to $x_3$ may perform poorly on this example even though $t$ is large.

\paragraph{Challenge 2:}
The above challenge is further complicated by the fact that graphs with large $t$ but small $x_3$ may be structured very differently (see \Cref{fig:diamond_extremes}). In contrast to the case where $W$ is an independent set of vertices, consider a clique-like extreme: Aside from a single vertex (say, $w_1$), the remaining vertices $w_2, \ldots, w_d$ form a clique. Here, the density of edges in $W$ prevents the formation of diamonds among the clique members. The number of diamonds $d-1$ is still large, yet $x_3 = 2$.

\begin{figure}[H]
    \centering
    \tikzset{
        v/.style={circle, draw, fill=gray!10, minimum size=20pt, inner sep=0pt, font=\small},
        u_node/.style={v, fill=blue!20},
        w_node/.style={v, fill=red!20},
        clique_boundary/.style={circle, draw, thick, blue!60, dashed, minimum size=3.2cm}
    }
    \begin{subfigure}[b]{0.48\textwidth}
        \centering
        \begin{tikzpicture}
            \node[v] (v1) at (-2.5, 1) {$v_1$};
            \node[v] (v2) at (-2.5, -1) {$v_2$};
            \node[v] (v3) at (-4, 0) {$w_1$};

            \draw[] (v1) -- (v2);
            \draw[] (v1) -- (v3);
            \draw[] (v2) -- (v3);

            \node[clique_boundary] (clique) at (2.5,0) {};
            \node[blue!60, align=center] at (2.5,0) {Clique of size $d-1$};

            \foreach \angle in {120,130,...,250} {
                    \draw (v1) -- (clique.\angle);
                    \draw (v2) -- (clique.\angle);
                }
        \end{tikzpicture}
        \caption{The clique-like extreme. Vertices $v_1$ and $v_2$ are adjacent to every vertex in $W$,
            so each of them forms a triangle with any clique vertex.
            The vertex $w_1$
            completes the induced diamonds: every diamond is $d$-heavy.
            The only \degc vertices are $v_1$ and $v_2$, hence $x_3=2$.
        }
        \label{fig:clique_extreme}
    \end{subfigure}
    \hfill
    \begin{subfigure}[b]{0.48\textwidth}
        \centering
        \begin{tikzpicture}
            \node[v] (v1) at (-2, 1.5) {$v_1$};
            \node[v] (v2) at (-2, -1.5) {$v_2$};

            \draw (v1) -- (v2);

            \node[w_node] (w1) at (2, 2) {$w_1$};
            \node[w_node] (w2) at (2, 0.7) {$w_2$};
            \node[font=\Large] (dots) at (2, -0.2) {$\vdots$};
            \node[w_node] (wd) at (2, -1.5) {$w_d$};

            \draw[dashed, red!60, rounded corners] (1.4, 2.5) rectangle (2.6, -2) node[midway, right=1cm, red!60, align=center, font=\footnotesize] {I.S. of size $d$};

            \foreach \node in {w1, w2, wd} {
                    \draw (v1) -- (\node);
                    \draw (v2) -- (\node);
                }
        \end{tikzpicture}
        \caption{The independent-set extreme. The vertices $v_1,v_2$ are adjacent to all vertices of an independent set $W$ of size $d$. Every pair of vertices in $W$ together with $v_1,v_2$ forms an induced diamond, yielding $\binom{d}{2}$ diamonds, while $x_3=2$.}
        \label{fig:indep_extreme}
    \end{subfigure}
    \caption{Two extreme configurations illustrating the ``hard case'' where diamonds cluster around a single diagonal edge $(v_1,v_2)$. In both scenarios, the number of diamonds $t$ can be arbitrarily large, while the number $x_3$ of \degc vertices remains a constant $2$.
    }
    \label{fig:diamond_extremes}
\end{figure}
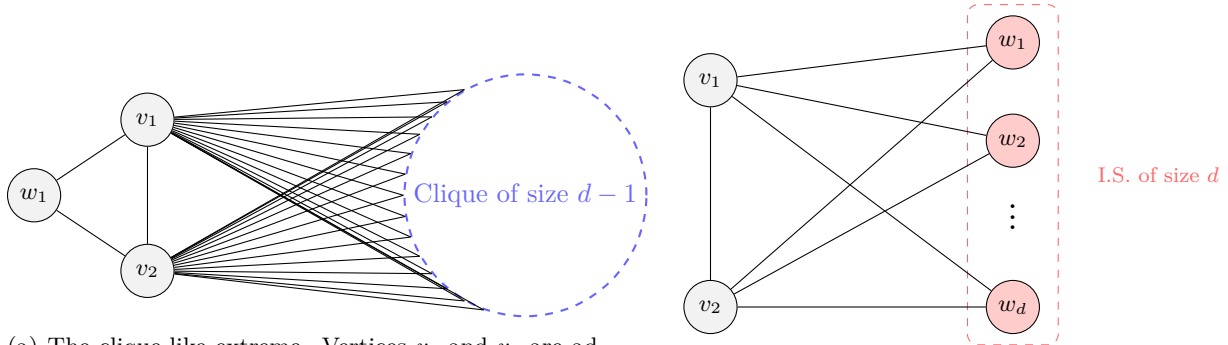

\paragraph{Our approach: pinpointing heaviness.} The spectrum between the above two extremes for large $t$ and small $x_3$ actually serves as a hint for us for how to look at the structure of graphs in order to quickly detect induced diamonds. As mentioned earlier, we define the notion of $r$-heaviness of a diamond, which says that three of its vertices are contained in a clique of size $r$. We define
$\rmax$ to be the size of the largest clique in the graph containing a diagonal edge of some diamond. Note that in the independent-set extreme example $\rmax=3$ and in the clique-like extreme example $\rmax=d+1$.

This notion is our first step toward overcoming Challenge 2: it enables a ``win-win'' approach that quickly detects an induced diamond either when there are $r$-heavy diamonds or when there are $r$-light diamonds. For Challenge 1, we note that (i) our algorithm for detecting $r$-heavy diamonds is independent of $x_3$, and (ii) our algorithm for detecting $r$-light diamonds uses the above algorithm for finding a \degc vertex when $x_3$ is large, while its main technical novelty of sampling vector refinement lies in the case where $x_3$ is small. This way, our final algorithm is faster as $t$ increases as desired, regardless of $x_3$.

\subsubsection{Detecting \texorpdfstring{$r$-}~Heavy Diamonds}
Our first algorithm detects an $r$-heavy induced diamond with a running time that improves as $r$ increases.
At a high level, we are looking for an edge that is contained in two different maximal cliques, as such an edge implies an induced diamond.
Intuitively, one might expect that the presence of an $r$-heavy diamond would make detection easier via sampling.
However, we must be careful: there may be only one such special edge in the graph, making it hard to sample directly.

Concretely, we use the aforementioned edge-based characterization of diamond-free graphs:
$G$ is diamond-free if and only if every edge of $G$ is contained
in exactly one maximal clique \cite[Lemma~7]{chiarelli2021strong}.
We call an edge that lies in two or more maximal cliques a \emph{violating edge}, since its existence implies an induced diamond.
For an edge $e=(u,v)$, we can test whether it is violating in $O(n^2)$ time by computing
$W(e)\triangleq (N(u)\cap N(v))\cup\set{u,v}$ and checking whether $W(e)$ is a clique: $e$ is violating if and only if $W(e)$ is not a clique.
This yields a simple $O(mn^2)$ time algorithm for induced diamond detection.
While this can be improved upon as mentioned earlier~\cite{eisenbrand2004complexity, williams2014finding}, this approach does not benefit from the fact that there is an $r$-heavy diamond.

We strengthen this characterization as follows to leverage the existence of an $r$-heavy diamond.
We know that there exists a violating edge that lies in a clique of size at least $r$, called an \emph{$r$-heavy violating edge}.
Is finding such an edge any easier?
We cannot sample $o(m)$ edges and expect to hit an $r$-heavy violating edge with good probability, since there may be only one such edge (see \Cref{fig:clique_extreme}).
Luckily, an $r$-heavy violating edge implies the existence of many other edges that are easier to find, and we can use them to find the violating edge itself, as follows.
Given an $r$-heavy violating edge $f$ that lies in a clique $C$ of size at least $r$, we refer to every edge in $C$ as an $r$-heavy \emph{revealing} edge.
While there might be only one $r$-heavy violating edge, it implies at least $\Omega(r^2)$ distinct $r$-heavy revealing edges (the edges in $C$).
Therefore, we can sample $o(m)$ edges and expect to hit an $r$-heavy revealing edge with good probability. We then use it to find the violating edge $f$.

Given this observation, our goal is to (1) sample a small set of vertices $S$ that induces a small set of edges $L$ containing an $r$-heavy revealing edge \whp, and (2) efficiently test whether an edge $e\in L$ is $r$-heavy revealing.
We provide a simple procedure to test whether an edge $e$ is revealing in $O(|W(e)|^2 + n)$ time: first check whether $W(e)$ is a clique (if not, $e$ is already violating); if it is, search for a neighbor $z\not\in W(e)$ connected to two vertices in $W(e)$.
To avoid processing large cliques, we only test edges $e$ such that $|W(e)|\in[r,2r)$.
That is, we run the same algorithm in a round-robin fashion with different values of $r$ until one of them detects an induced diamond.
Specifically, we run the algorithm for values of
$r=2^i$ for $i=0,1,2,\ldots,\log n$, where the $i$-th iteration samples $\tf((n/2^i)^2)$ edges, and tests only edges $e$ with $|W(e)|\in[2^i,2^{i+1})$.
In what follows, we analyze the $i$-th iteration for which $\rmax/2\leq 2^i \leq \rmax$.
This iteration is guaranteed to find an $r$-heavy revealing edge \whp
and its running time is the fastest among all iterations that find such an edge.

To obtain $L$, we select a random subset $S$ of $\Theta(n\log n/r)$ vertices.
Then, using fast matrix multiplication, we compute $|W(e)|$ for every $e$ with both endpoints in $S$ in $\MM{|S|,n,|S|}=\To(r\cdot (n/r)^\omega)$ time, where we add an edge $e$ to $L$ if $|W(e)|\in[r,2r)$. Clearly, the size of $L$ is at most the square of the size of $S$.
So far, we explained how to find $L$ in $\To(r\cdot (n/r)^\omega)$ time,
and how to process it in time $O(\abs{L}\cdot (r^2+n))=\To(n^2 + n^3/r^2)$.

We improve upon the above processing time via two modifications.
First, we modify the algorithm so that after processing an edge $e\in L$, it removes from $L$ all edges with both endpoints in $W(e)$.
Note that this never removes an edge $e'$ that is a revealing edge, for the following reason. Let $e'$ be such a removed edge and consider its set $W(e')$. If $W(e')=W(e)$, then $e'$ is revealing if and only if $e$ is revealing, so testing $e'$ in addition to $e$ is redundant. If $W(e')\neq W(e)$, then $e'$ is itself a violating edge, which means that $e$ is a revealing edge and the algorithm would terminate when processing $e$.

Second, we provide an improved analysis that shows that the algorithm stops after processing $\To(n/r+n^2/r^3)$ edges, which is better than the trivial bound of $|L|=\tf((n/r)^2)$.
To see why, let $e_i$ be the $i$-th edge that is processed, and consider the auxiliary bipartite graphs $F_i$ with parts $(W(e_1),W(e_2),\ldots,W(e_i))$ and $V(G)$, in which the clique $W(e_j)$ is connected to the vertex $v$ if and only if $v\in W(e_j)$.
We show that if $F_i$ contains a cycle of length at most $6$, then the algorithm stops while processing the first $i$ edges.
To illustrate this, suppose that $F_{i}$ contains a $4$-cycle while $F_{i-1}$ does not.
Then $W(e_i)$ shares two vertices $(x,y)$ with some previous clique $W(e_j)$, meaning that $(x,y)$ is a violating edge and every edge in $W(e_i)$ is $r$-revealing. Therefore, when processing $e_i$, the algorithm determines its revealing edge and stops, without processing any further edge.
For $6$-cycles, the argument is more complex but similar.
Classical bounds on edge count in unbalanced bipartite graphs with no $4$-cycle or $6$-cycle \cite{naor2005note,Neuwirth2001Girth8} imply that for $i\geq 64(n/r + n^2/r^3)$, the graph $F_i$ contains a cycle of length at most $6$.\footnote{Note that showing that either $F_i$ has no short cycle or we encounter a diamond does not seem to generalize to larger cycles (which would yield a better bound), as an $8$-cycle does not necessarily imply a diamond.} 
Therefore, the algorithm processes at most $O(n/r + n^2/r^3)$ edges.
To summarize, we find the set $L$ in $\To(r\cdot (n/r)^\omega)$ time and process only $O(n/r + n^2/r^3)$ edges from $L$, each in $O(r^2 + n)$ time, for a total time of $\tf(r\cdot (n/r)^\omega + n^3/r^3 +nr)$.

\subsubsection{Detecting \texorpdfstring{$r$-}~Light Diamonds}
Our second algorithm targets the regime where many induced diamonds are $r$-light.
Our starting point is the sampling framework of
\cite{tvetek2022approximate,CEW24,CEW25} that takes a $k$-vertex subgraph detection algorithm and converts it into a witness-sensitive algorithm.
We describe how it works for induced diamond detection, the algorithm needed to use it, why this is hard, and finally our main technical novelty for obtaining the required algorithm.

\subsubsection*{The Sampling Framework}
In the framework, we assign a random coloring $\phi: V(G)\to [4]$
and search for \emph{colorful} patterns, i.e., patterns whose $4$ vertices
have $4$ distinct colors.
This consists of two steps:
(1) sampling $\tf(1)$ induced subgraphs, and (2) searching for colorful induced diamonds in each one.

The sampling step uses the color classes to sample vertices with different probabilities.
Specifically, each sampled subgraph $H$ is obtained by sampling vertices of color $i$ with probability $p_i$ for $i\in[4]$.
We refer to the sampling probabilities as a \emph{sampling vector} denoted by $P=(p_1,p_2,p_3,p_4)$, and we denote the sampled subgraph by $H\gets G[P]$.
To obtain our collection of induced subgraphs, we consider all sampling vectors in $\FF=\{1,1/2,1/4,\ldots,1/n\}^4$ and sample one subgraph for each $P\in\FF$.
Note that the sampling vectors are independent of the underlying pattern we want to find.
The key guarantee of the framework is that if $G$ contains $t$ induced diamonds, then there exists a \emph{good} $P\in\FF$ such that $H\gets G[P]$ contains a colorful induced diamond with $\tg(1)$ probability and $w(P)=\tf(1/t)$, where $w(P)=\prod_{i=1}^4 p_i$ is the weight of $P$.

\newcommand{\Pbal}{P_{\mathrm{balanced}}}
\newcommand{\Pnbal}{P_{\mathrm{unbalanced}}}
The second step is to search for a colorful induced diamond in $H$.
Assume that $P$ is the promised good sampling vector.
To illustrate the potential speedup and challenges, consider two extreme cases, depending on whether $P$ is the balanced vector $\Pbal=(p,p,p,p)$ for $p=1/t^{1/4}$, or an unbalanced vector, such as $\Pnbal=(1,1,1,1/t)$.
For $H\gets G[\Pbal]$, the graph $H$ has $O(n/t^{1/4})$ vertices, so faster colorful induced diamond detection on $H$ sounds plausible.
For the unbalanced case, things are more complicated.
For $H\gets G[\Pnbal]$, the graph $H$ has $\Theta(n)$ vertices and possibly $\Theta(|E(G)|)$ edges, so $H$ is not sparser than $G$.

\paragraph{Colorful Induced Diamond Detection.}
Our task is thus to find an algorithm for colorful induced diamond detection that runs faster on unbalanced graphs than on $G$ itself.

Note that no colorful induced diamond detection algorithm was previously known and,
moreover, detecting an \emph{ordered-colorful} induced diamond (with predetermined colors for the vertices of its missing edge) is as hard as
$4$-clique detection~\cite{Marx2010}.

We observe, perhaps somewhat surprisingly, that colorful induced diamond detection is still possible. However, it is not fast on unbalanced graphs. To this end, we adapt the algorithm of~\cite{williams2014finding}, which detects induced diamonds, to detect \emph{colorful} induced diamonds (and, more generally, this applies to any nontrivial four-vertex induced subgraph that is neither a clique nor an independent set).
Also note that combined with \cite[Theorem 1.1]{DellLM22}, this yields a
$\To(n^\omega/\eps^{O(1)})$ time algorithm for $(1\pm\eps)$ approximate
counting of any nontrivial four-vertex induced subgraph, which is the first application of \cite{DellLM22} to induced subgraph counting.

\paragraph{The Challenge.}
The above colorful induced diamond detection algorithm is not sufficient for obtaining a fast witness-sensitive induced diamond detection algorithm, and in fact it does not yield any speedup. The issue is that the algorithm is not fast on unbalanced graphs.
To see why, we analyze its running time on each sampled graph $H\gets G[P]$ and show that it is $\tilde{\Theta}(\MM{n,n,n\cdot \frac{w(P)}{\min(P)}})$,
where $\frac{w(P)}{\min(P)}$ is the ratio between the weight of $P$ and its smallest coordinate.
Therefore, running the algorithm on $H\gets G[\Pnbal]$ takes time $\tilde O(\MM{n,n,n})=\tilde O(n^\omega)$, which is not faster than running~\cite{williams2014finding} directly on $G$.

This limitation is not merely an artifact of our analysis but is inherent to
the framework itself, as illustrated by the construction in
\Cref{fig:clique_extreme}, where every induced diamond contains the same triplet $(v_1,v_2,w_1)$.
If these vertices receive distinct colors (say $v_1,v_2,w_1$ have colors $1,2,3$ respectively), then for any sampled graph $H\gets G[P]$ with $P=(p_1,p_2,p_3,p_4)$, the probability that all three appear is $p_1p_2p_3$.
Unless $p_1p_2p_3=\Omega(1)$, the sampled graph $H$ is unlikely to contain any colorful induced diamond. However, if $p_1p_2p_3=\Omega(1)$, then $w(P)/\min(P)\ge p_1p_2p_3=\Omega(1)$, so no speedup is possible.

\newcommand{\rt}{\rmax/t}
\subsubsection*{Main New Technique: A Refined Analysis that Yields a Speedup}
We show that a refined analysis of the sampling framework \emph{does} yield a
speedup when many induced diamonds are $r$-light.
Previously, we argued that in the worst case, every induced diamond might contain the same triplet of vertices, e.g., $(v_1,v_2,w_1)$ in \Cref{fig:clique_extreme}, and therefore the only good sampling vector was $\Pnbal=(1,1,1,1/t)$.
In this case, however, the three repeated vertices lie in a large clique, so there exists an $r$-heavy diamond.
Specifically, let $W=N(v_1)\cap N(v_2)$ be the common neighborhood of $v_1,v_2$, and let $W'=W\setminus\set{w_1}$. If $(v_1,v_2,w_1)$ are the three vertices in all induced diamonds, then $W'$ is a clique, and every diamond is $|W'|$ heavy.
In this case, we use the $r$-heavy diamond detection algorithm.
On the other hand, for sufficiently small $r$ that is a function of $t$, the $r$-light diamond detection algorithm is faster. We do not know for which $r$ to switch between the two algorithms\footnote{
    Since rectangular matrix multiplication has no closed-form expression, we cannot find, for every $t$, the value of $r$ for which the two running times are equal.}
but we do not need this information; rather, we simply run both algorithms in parallel and stop when one of them detects an induced diamond.
To illustrate when the $r$-light diamond algorithm is faster, assume
that there is no $r$-heavy diamond for $r\leq |W|^{1-\eps}$ for some constant $\eps>0$.
Then $W'$ is not a clique, and by Tur\'an's theorem there are many non-edges in $G[W]$, specifically at least $\Omega(|W|^2/r)=\Omega(|W|^{1+\eps})$ non-edges, so there must be at least $\Omega(|W|^{\eps})$ vertices in $W$ that are endpoints of non-edges, proving that no vertex in $W$ is ``too'' important.
Therefore, we can hope to find a good sampling vector that is more balanced than $\Pnbal$. Specifically, we show that $P=(1,1,\sqrt{r/t},\sqrt{r/t})$ hits an induced diamond with $\tg(1)$ probability.
This illustrates how having no $r$-heavy diamonds simplifies the family of graphs we have to deal with. Yet, we still need to handle general graphs with no $r$-heavy diamonds, and we cannot assume that there are only two \degc vertices as in the above example.
To get our speedup, we show that either $x_3= \tg(\sqrt{t/\rmax})$, i.e., there are many \degc vertices, or that $P=(1,1,\sqrt{\rt},\sqrt{\rt})$ hits a colorful induced diamond with $\tg(1)$ probability, where in both cases we obtain a running time of $\tilde O(\MM{n,n,n\cdot \sqrt{\rmax/t}})$.
To prove that $P=(1,1,\sqrt{r/t},\sqrt{r/t})$ hits a colorful induced diamond with probability $\tg(1)$, our main technical contribution deviates from the black-box use of the sampling framework.
Specifically, we show that we can combine the framework (which is oblivious to $G$) with
structural properties implied by $r$-lightness to obtain a faster algorithm.
Our proof gradually refines the sampling vector from $P^{(1)}=(1/t)$ to $P^{(4)}=(p_1,p_2,p_3,p_4)$ by extending its dimension (i.e., the number of coordinates) in a very subtle manner.

\paragraph{Refined Analysis.}
We consider the graph $G$ with a random coloring $\phi: V(G)\to[4]$, and the set of ($r$-light) colorful induced diamonds $\DD_\phi$, containing $t$ elements.
We construct a sequence of hypergraphs, starting from 1-partite and ending with
4-partite:
\begin{align*}
    \GG^{(1)} & =((V_1\times V_2\times V_3\times V_4),\DD_{\phi})\;,      \\
    \GG^{(2)} & =((V_1 \times V_2) \sqcup (V_3 \times V_4),\DD_{\phi})\;, \\
    \GG^{(3)} & =((V_1 \times V_2) \sqcup V_3 \sqcup V_4,\DD_{\phi})\;,   \\
    \GG^{(4)} & =(V_1 \sqcup V_2 \sqcup V_3 \sqcup V_4,\DD_{\phi})\;.
\end{align*}
Each $\GG^{(i)}$ is $i$-partite and uses the same hyperedge set $\DD_\phi$,
where each hyperedge corresponds to a colorful induced diamond.
In $\GG^{(1)}$, each hyperedge is a single element from the product space
$V_1\times V_2\times V_3\times V_4$, so sampling each element with probability
$q$ hits a hyperedge with probability $1-(1-q)^t$.
As previously mentioned, \cite{CEW24,CEW25} implies that there exists a
sampling vector $P^{(4)}=(p_1,p_2,p_3,p_4)$ with $w(P^{(4)})=\tf(1/t)$ that
hits a hyperedge with probability $\tg(1)$.
To refine the guarantee, we define a sequence of sampling vectors as in
\Cref{fig-intro:refine-sequence}:
\begin{figure}[H]
    \centering
    \begin{tikzpicture}
        \node[draw, inner sep=4pt, rounded corners=3pt] (P1) {$P^{(1)}=\brak{\frac{1}{t}}$};

        \node[draw, inner sep=4pt, rounded corners=3pt, right=2.1cm of P1] (P2) {$P^{(2)}=(q_1,q_2)$};
        \node[above=2pt of P2] {\small $q_1\cdot q_2=\lambda/t$};

        \node[draw, inner sep=4pt, rounded corners=3pt, right=2.1cm of P2] (P3) {$P^{(3)}=(q_1,p,p)$};
        \node[above=2pt of P3] {\small $p=\sqrt{rq_2}$};

        \node[draw, inner sep=4pt, rounded corners=3pt, right=2.1cm of P3] (P4) {$P^{(4)}=(p_1,p_2,p,p)$};
        \node[above=2pt of P4] {\small $p_1\cdot p_2= \lambda\cdot q_1$};

        \node[draw, inner sep=4pt, rounded corners=3pt, below=1cm of P4, red] (P5) {$\hat{P}^{(4)}=(1,1,1,1/t)$};

        \draw[->, thick] (P1) -- node[above] {\footnotesize $1^{\text{st}}$ refinement} (P2);
        \draw[->, thick] (P2) -- node[above] {\footnotesize $2^{\text{nd}}$ refinement} (P3);
        \draw[->, thick] (P3) -- node[above] {\footnotesize $3^{\text{rd}}$ refinement} (P4);
        \draw[->, thick, red, dashed, bend right=10] (P1.south) to node[pos=0.6, xshift=8em, yshift=.8em, text width=6cm] {\small {\color{black!80!white} Using \cite{CEW25} as a black box}} (P5.west);
    \end{tikzpicture}
    \caption{Refinement sequence from $P^{(1)}$ to $P^{(4)}$. We set $\lambda=48\log n$.}
    \label{fig-intro:refine-sequence}
\end{figure}
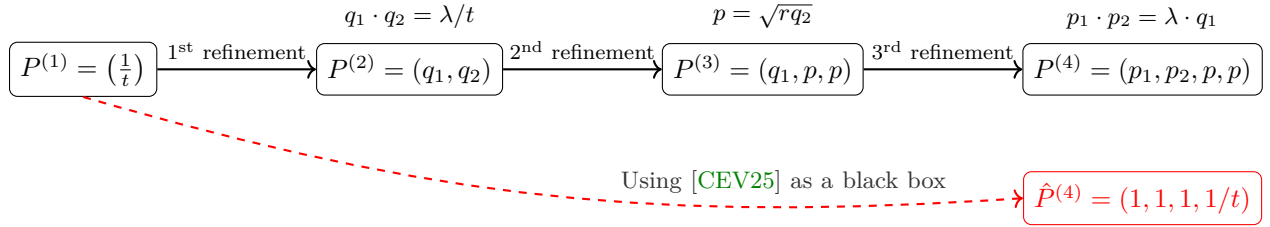
The first vector is $P^{(1)}=(1/t)$, which samples each element in
$\GG^{(1)}$ with probability $1/t$.
The first refinement is straightforward: view $\GG^{(2)}$ as a bipartite graph
with vertex sets $A=(V_1\times V_2)$ and $B=(V_3\times V_4)$ and $t$ edges.
For $0\le i\le \log t$, let $A_i$ be the set of vertices in $A$ with degree in
$[2^i,2^{i+1})$.
This lets us control degrees on the $A$-side.
We say that $A_i$ is \emph{heavy} if the number of edges incident to it is at
least $t/\log t$. At least one set must be heavy; otherwise, the total number of edges is less than $t$.
If there exists a heavy $A_i$ with $2^i\geq \sqrt{t/r}$, then $x_3
    \geq \sqrt{t/r}$ and we are in the easy case.
Otherwise, all heavy $A_i$ have $2^i<\sqrt{t/r}$. We fix one such $i$ and
set $P^{(2)}=(q_1,q_2)$ with $q_1=\lambda\cdot 2^i/t$ and
$q_2=4/2^i$. The choice of $i$ implies that $q_2\leq \sqrt{t/r}/t=1/\sqrt{rt}$.
To refine $P^{(3)}$ into $P^{(4)}=(p_1,p_2,p,p)$, we use a general refinement
theorem that does not use the additional
structure of $r$-light diamonds.
Refining $P^{(2)}$ into $P^{(3)}$ is the key step that uses the promise that there are no $r$-heavy diamonds, as previously explained.
We generalize the intuition from the running example with only two \degc vertices $v_1,v_2$ by refining $P^{(2)}=(q_1,q_2)$ into $P^{(3)}=(q_1,p,p)$, with $p=\sqrt{r\cdot q_2}$, where we note that because $q_2\leq 1/\sqrt{rt}$, we have $p\leq \sqrt{r/t}$.
This keeps two coordinates small, instead of just one small one, thereby decreasing the quantity $w(P^{(3)})/\min(P^{(3)})$ that governs the running time of the algorithm.
This comes at the cost of increasing the weight of the sampling vector by a factor of $r$, i.e., $w(P^{(3)})=r\cdot w(P^{(2)})$.

\subsection{Additional Results}
Our novel refinement framework yields two additional results.

First, we obtain an improved running time for detecting induced $4$-cycles when
many are $r$-light, i.e., no two vertices of the $4$-cycle lie in an $r$-clique.
The algorithm runs in time $\To(\MM{n,n,n\sqrt{r/t_r}})$,
where $t_r$ is the number of $r$-light induced $4$-cycles.
This requires different ideas, as induced-$C_4$-free graphs have different
structure than diamond-free graphs.
However, exploiting that every vertex in a $4$-cycle has the same role, we
achieve the same running time.

Second, we obtain a faster witness-sensitive algorithm for \dsum detection
when the number of solutions $t$ is at most $n$.
We reach $\tf(n^{2}/t^{1/3})$ time for \dsum detection and approximate
counting. 
For \dsum detection, there are three important regimes: the sparse regime with
$t\le n$, the medium regime with $n\le t\le n^3$, and the dense regime with
$t\ge n^3$.
Our algorithm is the first non-trivial witness-sensitive algorithm for \dsum in the sparse regime.
In the medium regime, it improves upon the previous best algorithm that takes $\tf(\min(n^2,n^{2}\cdot \sqrt{n/t}))$ time.
In the dense regime, the na\"{\i}ve algorithm that samples four uniform numbers and checks if they sum to zero is the fastest, taking $\tf(n^4/t)$ time, which is also sublinear when $t\gg n^3$.
To get the $\tf(n^2/t^{1/3})$-time algorithm for \dsum, we use the property that any three numbers participate in at most one solution, which allows us to refine the sampling vector similarly to the
induced-diamond and cycle cases, rather than using a black-box refinement
theorem.

Finally, we employ the sampling framework in combination with the structural analysis approach to obtain a combinatorial witness-sensitive algorithm for diamond detection:
\begin{restatable}{theorem}{ThmComb}\label{thm:comb}
    Let $G$ be an $n$-vertex graph with at least $t$ diamonds, where $t \leq n^2$. 
    There is a combinatorial algorithm that \whp runs in $\tf(n^3/\sqrt{t})$ time and finds an induced diamond in $G$.
\end{restatable}
We supplement this with a conditional lower bound against combinatorial algorithms, showing the result is tight up to polylogarithmic factors. 
    \begin{restatable}[Combinatorial Diamond Detection Lower Bound]{theorem}{CombLB}\label{thm:comb-lb}
        For every $0 \leq t \leq n^2/5$, every combinatorial randomized algorithm that finds an induced diamond in an $n$-vertex graph that has at least $t$ diamonds requires $n^{3-o(1)}/\sqrt{t}$ time under the Combinatorial Boolean Matrix Multiplication Conjecture.
    \end{restatable}

\subsection{Related Work}
When $H$ is the $k$-clique, the problem is solvable in $O(n^{\omega(a,b,c)})$ time for any partition $a+b+c=k$, where $\omega(a,b,c)$ denotes the exponent of multiplying an $n^a \times n^b$ matrix by an $n^b \times n^c$ matrix~\cite{itai1977finding,nesetril,eisenbrand2004complexity}.
In the sparse regime, $k$-clique detection can be solved in $O(m \cdot \alpha^{k-2})$ time, where $\alpha$ denotes the arboricity of $G$~\cite{ChibaN85}.
For other patterns such as paths and cycles, the color-coding technique yields $O((k!) \cdot n^\omega)$ time~\cite{alon1995color}.

We focus on \emph{induced} subgraphs: given $G$ and $\pat$, determine if $G$ contains $\pat$ as an induced subgraph.
Generally, detecting an $h$-vertex induced pattern reduces to $h$-clique detection on a graph with $hn$ vertices and $O(h^2 m)$ edges~\cite{nesetril}.
Significant attention has been devoted to $4$-vertex patterns~\cite{kloks2000finding,eisenbrand2004complexity,williams2014finding,dalirrooyfard2019graph,Dalirrooyfard2022InducedCA,abboud2026truly}.
Our primary interest lies in \emph{induced diamond} detection.
Kloks, Kratsch, and M{\"u}ller~\cite{kloks2000finding} characterized diamond-free graphs locally: $G$ is diamond-free if and only if for every vertex $v$, the induced subgraph on $N(v)$ contains no induced path on three vertices ($P_3$).
This characterization suggests a na\"{\i}ve algorithm: for each $v$, check if $G[N(v)]$ contains an induced $P_3$. This runs in $O(m+n)$ per vertex, or $O(n(m+n))$ total.
We refer to vertices whose neighborhoods contain an induced $P_3$ as \degc vertices.
Eisenbrand and Grandoni~\cite{eisenbrand2004complexity} improved this to $O(m^{3/2})$ by distinguishing between high- and low-degree vertices.
Subsequently, Vassilevska Williams, Wang, Williams, and Yu~\cite{williams2014finding} provided a deterministic algorithm running in $\tilde{O}(\min\{n^{\omega},m^{2\omega/(\omega+1)}\})$ time, which remains the state-of-the-art. 
Recently, Abboud, Akmal, and Fischer~\cite{abboud2026truly} introduced a purely \emph{combinatorial} algorithm for induced $4$-cycle detection running in $O(n^{2.84})$ time, without using fast matrix multiplication techniques.
This shows a separation between induced $4$-cycle detection and triangle detection under the standard Boolean matrix multiplication conjecture; see 
\cite{Dalirrooyfard2022InducedCA} for more details.

In \emph{property testing}, the goal is to read only $f(1/\eps)$ bits from the input, independent of input size, to determine whether a graph is $\pat$-free or $\eps$-far from being $\pat$-free (meaning at least an $\eps$-fraction of its edges must be removed to make it $\pat$-free) \cite{goldreich1998property,alon2000efficient,Goldreich2002-rm,Alon2015EasilyTG,Gishboliner2013DeterministicVN}.
Such an algorithm is called an $\eps$-tester, and properties admitting such testers are called \emph{testable}.
Alon, Fischer, Krivelevich, and Szegedy~\cite{alon2000efficient} proved that having a fixed pattern $\pat$ as an (induced) subgraph is testable.
Further work studies the dependency in $\eps$, and in particular for which properties the complexity $f(1/\eps)$ is polynomial in $1/\eps$.

In the \emph{sublinear} model, the goal is to detect $\pat$ using queries (degree, neighbor, pair, and sometimes random edge queries) in sublinear time~\cite{parnas2002testing,kaufman2004tight,aliakbarpour2018sublinear,eden2017approximately,assadi2018simple,Fichtenberger2020SamplingAS,Biswas2021TowardsAD,eden2025approximately}.
Assadi, Kapralov, and Khanna~\cite{AssadiKK19} provided an algorithm for detecting and approximately counting a fixed subgraph $\pat$ (induced or non-induced).
For any four-vertex pattern $\pat$ containing a $4$-cycle, they achieve query complexity $\TO{\min(m, m^2/t)}$ and runtime $\TO{m^2/t}$, where $t$ is the number of copies of $\pat$ in $G$.
The algorithm samples two random edges and checks if the graph induced on their endpoints contains $\pat$.
However, for induced diamonds, achieving $o(m+n)$ query complexity is impossible when $t=O(m)$; distinguishing a complete graph from a complete graph minus one edge (which contains $m$ induced diamonds) requires $\Omega(n)$ degree queries.
Without using random edge queries, Eden, Levi, Ron, and Rubinfeld~\cite{eden2025approximately} showed how to obtain query complexity and runtime of $\TO{\min(\frac{n}{t^{1/4}}+\frac{m^{2}}{t})}$.
Analogous questions have been explored in the \emph{distributed} setting; see~\cite{censor2022distributed} for a recent survey, and~\cite{le2021lower, Miyamoto25,nikabadi2022beyond} for specific results on induced subgraph detection.

\subsection*{Roadmap}
We begin with some preliminaries in \cref{sec:preliminaries}. 
In \cref{sec:r-heavy}, we present our algorithm for $r$-heavy diamonds and prove \cref{thm:r-heavy diamonds}. 
The algorithm for $r$-light diamonds and the proof of \cref{thm:r-light diamonds} follow in \cref{sec:r-light}. 
All lower bound proofs (for both general and combinatorial algorithms) appear in \cref{sec:lower-bounds}. 
\cref{sec:additional-results} presents our combinatorial algorithm for witness-sensitive diamond detection (\cref{thm:comb}), as well as results for other $4$-vertex patterns and the $4$-SUM problem. 
Finally, \cref{app:missing-proofs} contains proofs deferred from the main text.

\section{Preliminaries}
\label{sec:preliminaries}
Unless stated otherwise, throughout the paper by a diamond we mean an induced diamond.
A diamond's vertices have degree either 2 or 3; we call these \degb and \degc vertices, respectively (a vertex can be both).
We use $t,x$ and $x_3$ to denote the number of induced diamonds, the number of vertices that are part of an induced diamond, and the number of \degc vertices in $G$, respectively.
We say that an induced diamond is $r$-heavy if three of its vertices are contained in a clique of size at least $r$; otherwise, it is called $r$-light.
We use $\rmax$ to denote the largest integer $r$ such that $G$ contains an $r$-heavy induced diamond.
Given a four-coloring $\phi$ of the vertices of $G$, we say that an induced diamond is \emph{colorful} if all its vertices have different colors.

The following result will be used as a black-box in multiple sections across the paper. Its proof appears in \cref{app:missing-proofs}.
\begin{restatable}[\algVertexOne]{theorem}{ThmAlgVertexOne}\label{thm:is v in D}
    There is an algorithm \algVertexOne that, given a graph $G$ and a vertex $v\in V(G)$, either returns a diamond containing $v$ or reports that $v$ is not incident to a diamond. The algorithm runs in $O(m+n)$ time.
\end{restatable}

\paragraph{Graph Theory Notation.}
We use $\sqcup$ to denote disjoint union of vertex sets.
For every vertex $v$ in a graph $G$, we use $N(v)$ to denote the set of its neighbors in $G$. For an edge $e=(u,v)$, we use $N(e)\triangleq N(v)\cap N(u)$ to denote the set of common neighbors of $u$ and $v$, and $W(e)\triangleq N(e)\cup\{u,v\}$.

\paragraph{Matrix Multiplication.}
We use $\MM{n^a,n^b,n^c}$ to denote the time complexity of multiplying two matrices of sizes $n^a\times n^b$ and $n^b\times n^c$. This running time is also denoted by $n^{\omega(a,b,c)}$, where $\omega$ is the matrix multiplication exponent.
The function $\omega(a,b,c)$ is symmetric, meaning that for every permutation $\sigma:[3]\to[3]$ we have $\omega(x_1,x_2,x_3)=\omega(x_{\sigma(1)},x_{\sigma(2)},x_{\sigma(3)})$.
The following appears in \cite{abboud2024time,CEW25}.
\begin{claim}\label{claim:omega balanced}
    For $p_1,p_2,p_3\in[0,1]$ which may depend on $n$, we have:
    $\MM{n p_1, n p_2, n p_3} \leq
        \MM{n,n,n p_1 p_2 p_3}\;.$
\end{claim}

\noindent The following claim is a linear approximation of $\MM{n,n,np}$ for $p\in[1/n,1]$, and appears in \cite{huang1998fast,zwick2002all}.

\begin{claim}\label{claim:omega rect}
    $\MM{n,n,np}\leq n^\omega\cdot p^\beta + n^2$, for $p\in[1/n,1]$, where $\beta\approx 0.5475$.
\end{claim}
The constant $\beta=(\omega-2)/(1-\alpha)$ arises as follows. Here $\omega=\omega(1,1,1)\approx 2.372$ and $\alpha\approx 0.3214$ is the largest value satisfying $\omega(1,1,\alpha)=2$ \cite{moreasym}.
Since $\omega(1,1,x)$ is convex, the line through the points $(\alpha, 2)$ and $(1, \omega)$ gives a linear upper bound on $\omega(1,1,x)$ for $x\in[\alpha,1]$.

\paragraph{Probabilistic Tools.}

\begin{theorem}[Chernoff Bound {\cite{dubhashi2009concentration}}]\label{thm:chernoff 6}
    Let $X_1, \ldots, X_n$ be independent random variables with values in $[0,1]$ and $X=\sum_i X_i$. For $t\geq 6\Exp{X}$, and $\eps>0$ we have
    \begin{align*}
        \Pr{X\geq t}\leq 2^{-t}\;, &  &
        \Pr{X\leq (1-\eps)\Exp{X}}\leq \exp(-\eps^2\cdot \Exp{X}/2)
    \end{align*}
\end{theorem}

\begin{lemma}[Second Moment Method]
    Let $X$ be a non-negative integral random variable.
    Then
    \begin{align*}
        \Pr{X>0}\geq \frac{\Exp{X}^2}{\Exp{X^2}}\;.
    \end{align*}
\end{lemma}

\begin{lemma}[Reverse Markov's inequality {\cite[(1.6.4)]{doerr2019theory}}]\label{lem:rev mark}
    Let $X$ be a random variable with support contained in $[0,M]$.
    Then, for $R\in\mathbb{R}$ we have
    $\Pr{X>R}\geq \frac{\Exp{X}-R}{M-R}$.
\end{lemma}
All logarithms in this paper are base 2.

\section{Detecting $r$-Heavy Diamonds}
\label{sec:r-heavy}
\renewcommand{\W}{W}
\newcommand{\algDisc}{{\normalfont\textsc{FindHeavyHelper}}\xspace}
In this section we prove the following theorem:
\ThmHeavyD* 
Recall that $\rmax$ is the largest $r$ such that there is an $r$-heavy induced diamond in $G$, i.e., a diamond with three vertices contained in the same clique of size $r$.
Instead of working with $\rmax$, we prove this for any $r$:
\begin{theorem}
    \label{thm:r-heavy helper}
    There is a randomized algorithm $\algDisc(G, r)$ that finds an induced diamond in time $\To(\runtime)$ \whp, assuming $G$ has an $r$-heavy induced diamond and no $2r$-heavy induced diamond.
\end{theorem}

\begin{proof}[Proof of \Cref{thm:r-heavy diamonds} Using \Cref{thm:r-heavy helper}]
    Run $\algDisc(G, 2^i)$ for $i = 0, \ldots, \log n$ in round-robin, stopping when a diamond is found.
    Let $i^* = \lfloor \log_2 \rmax \rfloor$, so $2^{i^*} \le \rmax < 2^{i^*+1}$.
    By definition, $G$ has a $2^{i^*}$-heavy diamond but no $2^{i^*+1}$-heavy diamond, so the assumptions of \Cref{thm:r-heavy helper} hold for $r = 2^{i^*}$.
    Thus $\algDisc(G, 2^{i^*})$ finds a diamond in $\To((\rmax\cdot (n/\rmax)^\omega + (n/\rmax)^3+m))$ time \whp, and round-robin adds only $O(\log n)$ overhead.
\end{proof}

The main structural observation behind our algorithm is that in a diamond-free graph, every edge lies in exactly one maximal clique:
\begin{observation}\label{obs5:diamond from intersection}
    If two cliques $C_1,C_2\in \DD$ satisfy $C_1\neq C_2$ and $\abs{C_1\cap C_2}\geq 2$, then some vertex in $\Vp$ belongs to an induced diamond.
\end{observation}
\begin{proof}[Proof of \Cref{obs5:diamond from intersection}]
    Let $u_1,u_2\in C_1\cap C_2$ be two distinct vertices.
    Since $C_1\neq C_2$, we may assume without loss of generality that $C_1$ contains a vertex $u_3\notin C_2$.
    Let $v_2\in\Vp$ be a vertex whose partition contains $C_2$, i.e., $C_2\in\CC_{v_2}$.
    Then $\{u_1,u_2,u_3,v_2\}$ is an induced diamond, since it is a clique whose only missing edge is $(u_3,v_2)$.
    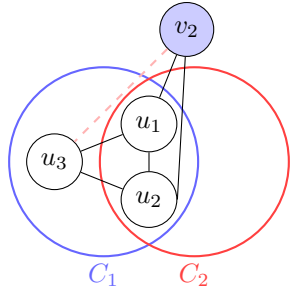
\begin{figure}[H]
        \centering
        \begin{tikzpicture}[scale=1.0, every node/.style={transform shape}]

            \node[circle,draw,fill=blue!20] (v2) at (.5, 1.75) {$v_2$};
            \node[circle,draw,] (u1) at (0, .5) {$u_1$};
            \node[circle,draw,] (u2) at (0, -.5) {$u_2$};
            \node[circle,draw,] (u3) at (-1.25, 0) {$u_3$};

            \def\r{1.25}
            \def\xa{.6}
            \def\ya{1.5}
            \draw[thick, blue!60, rotate=0] (-\xa,0) ellipse ({\r} and {\r});
            \node[blue!60] at (-\xa, -\ya) {$C_1$};
            \draw[thick, red!75, rotate=0] (\xa,0) ellipse ({\r} and {\r});
            \node[red!75] at (\xa, -\ya) {$C_2$};

            \draw[] (u1) -- (u2);
            \draw[] (u1) -- (u3);
            \draw[] (u2) -- (u3);
            \draw[] (v2) -- (u1);
            \draw[] (v2) -- (u2.east);
            \draw[thick,dashed,red!30] (v2) to (u3);
        \end{tikzpicture}
        \caption{$\{u_1,u_2,u_3,v_2\}$ is an induced diamond; the dashed red edge $(u_3,v_2)$ is missing.}
        \label{fig05:diamond from intersection}
    \end{figure}
\end{proof}
This observation also follows from \cite[Lemma~7]{chiarelli2021strong}: a graph is diamond-free if and only if every edge lies in a unique maximal clique.

\paragraph{Roadmap.}
The next paragraphs contain the notation and definitions used in the algorithm.
The first part of the algorithm includes finding a set of edges $L$, computed using fast matrix multiplication.
The second part of the algorithm (\Cref{ssec:find violating edge2}) explains how to process this set of edges efficiently.

\paragraph{Preliminaries and Definitions.}
For every edge $e=(u,v)\in E$, we use $\W(e)=(N(u)\cap N(v))\cup\set{u,v}$ to denote the set of common neighbors of $u$ and $v$, together with $u$ and $v$ themselves.
We refer to an edge that lies in two distinct maximal cliques as a \emph{violating edge}.
When $G$ is not diamond-free, it contains at least one violating edge, and if $G$ contains an $r$-heavy induced diamond, then there must be a violating edge that lies in a clique of size at least $r$.
The other direction is also true: if $G$ has a violating edge that lies in two maximal cliques $C_1,C_2$, then $G$ has a $|C_1|$-heavy induced diamond:
Consider an edge $e=(u,u')$ inside two maximal cliques $C_1,C_2$, and let $z\in C_2\setminus C_1$.
Let $v$ be any vertex in $C_1\setminus N(z)$. Then $(v,u,u',z)$ is an induced diamond, since the only missing edge is $(v,z)$, and $(v,u,u')$ lies in $C_1$, proving the existence of a $|C_1|$-heavy induced diamond.
We say that an edge $e'$ is \emph{revealing} if it lies in a maximal clique that contains a violating edge.
We restrict our attention to violating and revealing edges that are part of a clique of size at least $r$, referred to as $r$-violating edges and $r$-revealing edges, respectively.
Both $r$-violating and $r$-revealing edges are part of at least $r-2$ triangles.

Our algorithm consists of two steps: finding a set of edges $L$ that contains at least one $r$-revealing edge \whp and processing these edges one by one until we either find an induced diamond or exhaust all edges in $L$.
We next explain how to find the set $L$.

\paragraph{Finding $L$.}
Given a subset of vertices $\Vp$, define a subset of edges $L(\Vp,r)$ by
\begin{align*}
    L(\Vp,r)\triangleq \set{e\in G[\Vp] \mid |\W(e)| \in[r,2r]}\;.
\end{align*}
In words, this is the set of all edges with both endpoints in $\Vp$ whose common neighborhood has size in $[r-2,2r-2]$.
Computing for every edge $e$ in $G[\Vp]$ the exact number of triangles it is in, i.e., $\abs{\W(e)}-2$, can be done using fast matrix multiplication in time $\MM{|\Vp|,|\Vp|,n}$.
We show that if we sample $\Vp$ randomly, then \whp $L(\Vp,r)$ contains at least one $r$-violating edge.

\begin{lemma}
    \label{lemma:rheavy:exists r-violating}
    Let $\Vp$ be a random subset of vertices where each vertex is included independently with probability $p=\min(\frac{128\log n}{r},1)$.
    If $G$ has an $r$-heavy induced diamond and no $2r$-heavy induced diamond,
    then \whp $L(\Vp,r)$ contains at least one $r$-revealing edge.
    The randomness is only over the choice of $\Vp$.
\end{lemma}
\begin{proof}[Proof of \Cref{lemma:rheavy:exists r-violating}]
    Let $(v,u_1,u_2,z)$ be an $r$-heavy induced diamond with missing edge 
    $(v,z)$, where the vertices $v,u_1,u_2$ belong to the same maximal clique $Q$ whose size is in $[r,2r)$. Therefore, $f=(u_1,u_2)$ is an $r$-violating edge.

    Assume that $|Q\cap\Vp|\geq 2$, and let $s_1,s_2\in Q\cap\Vp$ be two distinct vertices in this set.
    We show that the edge $e=(s_1,s_2)$ belongs to $L(\Vp,r)$ and is $r$-revealing.
    If $e$ is violating, then it is also $r$-violating since it lies in $Q$, and therefore also $r$-revealing and we are done.
    Otherwise, $W(e)$ is a maximal clique that must contain $Q$.
    Therefore $W(e)$ contains the violating edge $f$, making $e$ $r$-revealing.
    Note that $|W(e)|<2r$, otherwise $f$ is $2r$-violating, contradicting the assumption that there are no $2r$-heavy diamonds in $G$.

    To see that $|Q\cap\Vp|\geq 2$ \whp, let $X$ be the random variable that counts the number of vertices from $Q$ that are included in $\Vp$.
    We have $\Exp{X}=p\abs{Q}\geq 128\log n$, so by Chernoff's inequality 
    $\Pr{X< 2}\leq \exp(-\Exp{X}/8)\leq 1/n^8$, which completes the proof.
\end{proof}

\subsection{Processing Edges in $L$}
\label{ssec:find violating edge2}
\newcommand{\algEdge}{\normalfont{\textsc{ProcessEdges}}\xspace}
For brevity, we write $L$ instead of $L(\Vp,r)$.
Below the algorithm $\algEdge(G,\Vp,L)$ is presented.

\begin{algorithm}[H]
    \newcommand{\gray}[1]{{\color{gray}{#1}}}
    \caption{Algorithm $\algEdge(G,\Vp,L)$: Process All Edges in $L$}
    \label{alg:process helper}
    \KwIn{The graph $G$, a set $\Vp$ of vertices and $L$ of edges.}
    \KwOut{A vertex in an induced diamond, or $\bot$ if no diamond is found.}
    \setcounter{AlgoLine}{0}
    \BlankLine
    \DontPrintSemicolon

    \While{$L$ is not empty}{
        Pop an edge $(x,y)$ from $L$.\;
        \smallskip
        \gray{\tcp{Part 1}}
        Let $W(x,y)\gets (N(x)\cap N(y))\cup\set{x,y}$.\;
        \smallskip
        \lIf{$G[W(x,y)]$ is not a clique}{
            \Return{$x$} \hfill $\triangleright$ \texttt{$(x,y)$ is part of a diamond.}
        }
        \medskip
        \gray{\tcp{Part 2}}
        Initialize an array $R$ of size $n$ with all entries $0$.\;
        \lFor{$v\in W(x,y)$}{$R[v] \gets -\infty$ \hfill $\triangleright$ \texttt{Mark vertices in $W(x,y)$.}}
        \smallskip
        \For{$v\in W(x,y)$}{
            \For{$z\in N(v)$ with $R[z]\neq -\infty$}{
                $R[z] \gets R[z]+ 1$.\;
                \lIf{$R[z] = 2$}{
                    \Return{$z$} \hfill $\triangleright$ \texttt{$z$ is in a diamond.}
                }
            }
        }
        \medskip
        \gray{\tcp{Part 3}}
        $S'\gets \Vp\cap W(x,y)$.\;
        \For{$x',y'\in S'$ such that $(x',y')\in L$}{
            Remove $(x',y')$ from $L$.\;
        }
    }
    \Return{$\bot$} \hfill $\triangleright$ \texttt{No diamond found.}
\end{algorithm}

We give an overview of the algorithm.
\paragraph{Part 1.}
We check whether $W(e)$ is a clique. If not, $e$ is violating and we have 
found an induced diamond. This takes $O(|W(e)|^2)=O(r^2)$ time, since edges 
in $L$ satisfy $|W(e)|\in [r,2r]$.

\paragraph{Part 2.}
We check whether some vertex $z\notin W(e)$ has at least two neighbors in 
$W(e)$. If so, letting $w,w'$ be two such neighbors, the vertex $z$ together 
with $w,w'$ and an endpoint of $e$ not adjacent to $z$ forms an induced 
diamond.
This step runs in $O(r^2+n)$ time: we initialize a counter $R[z]$ for each 
$z\notin W(e)$, and increment it for each edge between $z$ and $W(e)$, 
stopping when any counter reaches $2$.
Processing edges within $W(e)$ takes $O(r^2)$ time, and we process at most $n$ 
edges with exactly one endpoint in $W(e)$ before either finding a diamond or 
exhausting all such edges.

\paragraph{Part 3.}
Computing $S'=\Vp\cap W(e)$ takes $O(r)$ time, so removing all edges in $L$ 
with both endpoints in $S'$ takes $O(r^2)$ time.
~\\This completes the description of the algorithm.
We now prove its correctness and analyze its running time.

\begin{lemma}
    \label{lemma:rheavy:main}
    If $L$ contains an $r$-revealing edge and $G$ has no $2r$-heavy induced diamond,
    then $\algEdge(G,\Vp,L)$ finds a vertex in an induced diamond.
    Its running time is $O((n/r)^3 + n^2/r + nr)$.
\end{lemma}
\Cref{thm:r-heavy helper} follows from \Cref{lemma:rheavy:exists r-violating} and \Cref{lemma:rheavy:main}.
To prove the correctness of \Cref{alg:process helper}, we prove the following two claims:
\begin{claim}\label{claim:rh:part 12 correctness}
    If $e$ is an $r$-violating edge or an $r$-revealing edge, then Part 1 or Part 2 of \Cref{alg:process helper} finds a vertex in an induced diamond.
\end{claim}
\begin{proof}[Proof of \Cref{claim:rh:part 12 correctness}]
    If $e$ is an $r$-violating edge, then $\W(e)$ is not a clique, and therefore we find an induced diamond in Part 1 of the algorithm.
    We prove the claim for $r$-revealing edges.
    Let $e=(v_1,v_2)$ be an $r$-revealing edge in $L$, which lies in a maximal clique $Q$ together with an $r$-violating edge $h=(x,y)$.
    Let $\W(e)=(N(v_1)\cap N(v_2))\cup\set{v_1,v_2}$ be the set computed in Part 1 of the algorithm.
    If $Q\neq \W(e)$, then $\W(e)$ is not a clique, and we would have found an induced diamond in Part 1 of the algorithm. Assume that $Q=\W(e)$.
    Since $h$ is violating, it is also in a different maximal clique $Q'\neq Q$.
    Let $z\in Q'\setminus Q$, so $z\notin \W(e)$.
    Moreover, $z$ has two neighbors in $\W(e)$, namely $x$ and $y$, proving the claim.
\end{proof}

\begin{claim}\label{claim:rh:part 3 correctness}
    No $r$-revealing edge is removed from $L$ during Part 3 of \Cref{alg:process helper}.
\end{claim}
\begin{proof}[Proof of \Cref{claim:rh:part 3 correctness}]
    Let $f=(u_1,u_2)$ be an $r$-revealing edge in $L$, lying in a maximal clique $Q$ together with an $r$-violating edge $h=(x,y)$.
    Suppose we remove $f$ from $L$ after processing edge $e=(v_1,v_2)$.
    Then $f\in W(e)$, otherwise it would not be removed.
    If $W(e)\neq Q$, then $f$ is $r$-violating, so processing $e$ finds an induced diamond by \Cref{claim:rh:part 12 correctness}.
    If $W(e)=Q$, then $e$ is $r$-revealing because $\W(e)$ contains the $r$-violating edge $h$. Then \Cref{claim:rh:part 12 correctness} applies to $e$.
\end{proof}
The correctness of \Cref{alg:process helper} follows from \Cref{claim:rh:part 12 correctness,claim:rh:part 3 correctness}.
We analyze the total running time of processing all edges in $L$.
Consider the set of edges in $L$ processed by the algorithm $(e_1,e_2,\ldots,e_q)$, in the order they were processed.
Define $W_i=\W(e_i)$, and let $\QQ_i=\set{W_1,W_2,\ldots,W_i}$.
We define an auxiliary bipartite graph $F_i=(\QQ_i,V,E_i)$.
We add an edge between a clique $W\in \QQ_i$ and a vertex $v\in V$ if $v\in W$.
The following two lemmas are the main ingredients in the runtime analysis.
\begin{lemma}
    \label{lemma:rheavy:cycle implies diamond}
    If $F_i$ contains a $4$-cycle or a $6$-cycle, then the algorithm detects an induced diamond, no later than when processing edge $e_i$.
\end{lemma}

\begin{lemma}
    \label{lemma5:1:computation}
    If $F_i$ has no $6$-cycle, and $r\geq 60$, then
    \begin{align*}
        \abs{E(F_i)} < 64(n+ n^2/r^2)\;, &  & i < 64(n/r + n^2/r^3)\;.
    \end{align*}
\end{lemma}
\begin{proof}[Proof of \Cref{lemma:rheavy:main} Using \Cref{lemma:rheavy:cycle implies diamond} and \Cref{lemma5:1:computation}]
    The maximum number of cliques that the algorithm processes before the auxiliary graph $F_i$ contains a $6$-cycle is bounded by $64(n/r+n^2/r^3)$ by \Cref{lemma5:1:computation}.
    Since the algorithm finds an induced diamond as soon as $F_i$ contains a $6$-cycle by \Cref{lemma:rheavy:cycle implies diamond}, the algorithm processes at most $i\leq 64(n/r+n^2/r^3)$ edges, and since each edge requires $O(n+r^2)$ time, the total running time is $O((n+r^2)\cdot (n/r+n^2/r^3))=O((n/r)^3 + n^2/r + nr)$ as required.
\end{proof}

We prove the two lemmas.

\begin{proof}[Proof of \Cref{lemma:rheavy:cycle implies diamond}]
    \renewcommand{\Q}{W}
    \renewcommand{\QQ}{\mathcal{Q}_i}
    A $4$-cycle in $F_i$ corresponds to two cliques $W(e_i),W(e_j)\in \QQ$ and two vertices $u_1,u_2\in W(e_i)\cap W(e_j)$.
    This means that $(u_1,u_2)$ is an edge that lies in two distinct maximal cliques, making it a violating edge, and an $r$-violating edge since both $W(e_i)$ and $W(e_j)$ have size in $[r,2r]$.
    After processing the first edge among $e_i,e_j$, say $e_i$, the algorithm learns that $e_i$ is an $r$-revealing edge, and terminates by \Cref{claim:rh:part 12 correctness}.

    We prove that if $F_i$ contains a $6$-cycle, then the algorithm detects an induced diamond.
    We first show that if $F_i$ contains a $6$-cycle, then $G$ has an induced diamond, and then explain why this diamond implies the rest of the lemma.
    Assume that $F_i$ contains a $6$-cycle as a subgraph (see \Cref{fig05:c6-to-diamond}).
    \begin{figure}[h!]
        \centering
        \begin{tikzpicture}[scale=.42,
            vertex/.style={circle, draw, fill=white, minimum size=28pt, inner sep=1pt, font=\Large, thick},
            vertexe/.style={rectangle,rounded corners=8pt, draw, fill=gray!10, minimum size=28pt, inner sep=1pt, font=\Large, thick},
            edge/.style={thick, >={Stealth[round]}, shorten >=1pt, shorten <=1pt},
            blue edge/.style={thick, color=blue!60, >={Stealth[round]}, shorten >=1pt, shorten <=1pt},
            red edge/.style={thick, dashed, color=red, >={Stealth[round]}, shorten >=1pt, shorten <=1pt},
            oval d1/.style={draw=blue!60!cyan, fill=blue!30, fill opacity=0.5, line width=.2pt},
            oval d2/.style={draw=red!60,       fill=red!30,  fill opacity=0.5, line width=.2pt},
            oval d3/.style={draw=teal!70,      fill=teal!30, fill opacity=0.5, line width=.2pt},
            oval label/.style={font=\Large, font=\bfseries}
            ]

            \def\rU{3.2}      
            \def\rE{5.1}      
            \def\rC{2.7}      
            \def\rN{3.15cm}   
            \def\rD{4cm}    

            \coordinate (cu3) at ( 90:\rU); 
            \coordinate (cu2) at (210:\rU); 
            \coordinate (cu1) at (330:\rU); 

            \coordinate (cv1) at (270:\rE);  
            \coordinate (cv2) at (150:\rE);  
            \coordinate (cv3) at ( 30:\rE);  

            \coordinate (mid1) at (270:\rC);
            \draw[oval d1] (mid1) circle (\rD);
            \node[oval label, blue!60!cyan] at ($(mid1) + (+2.5, -1.65)$) {$W(e_1)$};

            \coordinate (mid2) at (150:\rC);
            \draw[oval d2] (mid2) circle (\rD);
            \node[oval label, red!60] at ($(mid2) + (-2.7, -1.0)$) {$W(e_2)$};

            \coordinate (mid3) at (30:\rC);
            \draw[oval d3] (mid3) circle (\rD);
            \node[oval label, teal!70] at ($(mid3) + (2.7, -1.0)$) {$W(e_3)$};

            \node[vertex] (u1) at (cu1) {$u_1$};
            \node[vertex] (u2) at (cu2) {$u_2$};
            \node[vertex] (u3) at (cu3) {$u_3$};

            \node[vertexe] (v1) at (cv1) {$e_1$};
            \node[vertexe] (v2) at (cv2) {$e_2$};
            \node[vertexe] (v3) at (cv3) {$e_3$};


            \draw[blue edge] (u1) -- (u2);
            \draw[blue edge] (u2) -- (u3);
            \draw[blue edge] (u3) -- (u1);

            \draw[blue edge] (v1) -- (u1);
            \draw[blue edge] (v1) -- (u2);

            \draw[red edge] (v1) to[out=110, in=290] (u3);

            \draw[edge] (v2) -- (u2);
            \draw[edge] (v2) -- (u3);

            \draw[edge] (v3) -- (u3);
            \draw[edge] (v3) -- (u1);

        \end{tikzpicture}
        \caption{A $6$-cycle in the graph $F_i$: $u_1\to W(e_1)\to u_{2}\to W(e_2)\to u_{3}\to W(e_3)\to u_{1}$.
        }
        \label{fig05:c6-to-diamond}
    \end{figure}
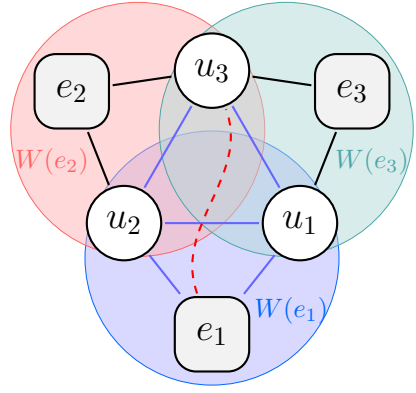

    Let $u_1\to W(e_1)\to u_{2}\to W(e_2)\to u_{3}\to W(e_3)\to u_{1}$, 
    be a $6$-cycle in $F_i$, where $W(e_j)\in \QQ$, and $u_{j}\in V$, for every $1\leq j\leq 3$.
    We use $(v_1,v_2)$ to denote the endpoints of $e_1$.
    We show that $e_1$ is an $r$-revealing edge, and therefore by \Cref{claim:rh:part 12 correctness} the algorithm finds an induced diamond while processing $e_1$.
    We consider two cases, based on whether $u_3\in W(e_1)$ or not:
    \begin{itemize}
        \item If $u_3\in W(e_1)$, then $(u_1,u_3)$ is a violating edge; both in $W(e_1)$ and $W(e_3)$.
              Thus, $e_1$ is an $r$-revealing edge.
        \item If $u_3\notin W(e_1)$, then it is not a neighbor of say $v_1$, but it has two neighbors $u_1,u_2\in W(e_1)$, thus $(u_1,u_2)$ is an $r$-violating edge, and $e_1$ is an $r$-revealing edge.
    \end{itemize}
\end{proof}

We prove \Cref{lemma5:1:computation}.
We need one more theorem from extremal combinatorics that bounds the number of edges in an unbalanced bipartite graph that does not contain a $6$-cycle as a subgraph.
\begin{theorem}[{\cite[Theorem 1]{naor2005note}}]\label{thm5:1:c6 extremal}
    $\mathrm{ex}(a,b,C_6)\leq 3\cdot \brak{(a\cdot b)^{2/3} + a + b}$.
\end{theorem}
In words, any bipartite graph with parts of sizes $a$ and $b$ and at least $3\cdot \brak{(a\cdot b)^{2/3} + a + b}$ edges contains a $6$-cycle as a subgraph.
Using \Cref{thm5:1:c6 extremal}, we show that if $F$ has no $6$-cycle, then it must be sparse.

\begin{proof}[Proof of \Cref{lemma5:1:computation}]
    \newcommand{\ee}{m_F}
    We use $\ee$ for $\abs{E(F_i)}$, and $N$ for $\abs{\QQ_i}$.
    Since every clique in $\QQ_i$ is of size at least $r$, it is incident to at least $r$ edges in $F_i$, so $\ee\geq r\cdot N$, implying that $N\leq \ee/r$.
    We plug this into \Cref{thm5:1:c6 extremal}, obtaining that if $F_i$ has no $6$-cycle, then
    \begin{align*}
        \ee           & \leq 3\cdot \brak{(n\cdot N)^{2/3} + n + N} \\
        \frac{\ee}{3} & \leq (n\cdot (\ee/r))^{2/3} + n + \ee/r\;.
    \end{align*}
    Since $r\geq 60$, we can replace $\ee/r$ by $\ee/60$. We get
    \begin{align*}
        \frac{19\ee}{60} & \leq (n\cdot\ee/r)^{2/3} + n\;.
    \end{align*}
    Assume that $\ee \geq 30n$, since otherwise the lemma holds trivially.
    \begin{align*}
        \frac{17\ee}{60} & \leq (n\cdot\ee/r)^{2/3}\;.
    \end{align*}
    By rearranging, we obtain:
    \begin{align*}
        (\ee)^{1/3} & \leq \frac{60}{17}\cdot (n/r)^{2/3} < 4\cdot (n/r)^{2/3} \\
        \ee         & < 64\cdot (n/r)^{2}\;.
    \end{align*}
    By plugging this back into $N\leq \ee/r$, we obtain
    $N\leq \ee/r < 64\cdot (n^2/r^3)$, which concludes the proof.
\end{proof}

\section{Detecting $r$-Light Diamonds}
\label{sec:r-light}
In this section we prove the following.
\ThmLightD* 

Recall that a diamond is $r$-light if it does not contain three vertices inside a clique of size $r$, and $\rmax$ is the largest integer $r$ such that $G$ contains an $r$-heavy diamond.

\paragraph{Roadmap.}
In \Cref{ssec:reduction} we present the color-coding reduction and the colorful diamond detection algorithm.
In \Cref{ssec:subsampling} we describe the sampling framework of~\cite{CEW24,CEW25} and our algorithm \algSensitive, which is based on the former.
In \Cref{ssec:main-light} we prove the main result of the section, \Cref{thm:r-light diamonds}.
In \Cref{ssec:refine} we prove several tools used in the previous subsection,
which may be of independent interest; specifically, we refine the sampling framework of~\cite{CEW24,CEW25}.
\subsection{Colorful Induced Diamond Detection}
\label{ssec:reduction}

This subsection has two main components. First, we reduce the problem of detecting an induced diamond in a graph $G$ with $n$ vertices, $t$ induced diamonds, and $x$ vertices that participate in an induced diamond to the problem of detecting a colorful induced diamond in a graph $G^\prime$ with $\TO{n}$ vertices, $\TO{m}$ edges, $\tilde{\Theta}(t)$ colorful induced diamonds, and $\tilde{\Theta}(x)$ vertices that participate in a colorful induced diamond.
Second, we provide an algorithm for detecting a colorful induced diamond in $G^\prime$. 
Note that a colorful induced diamond is an induced diamond with vertices of distinct colors.

\begin{lemma}\label{lemma:reduction}
    There exists a randomized $\TO{n+m}$ time reduction that, given a graph $G$ with
    $n$ vertices, $m$ edges, $t$ induced diamonds and $x$ vertices that participate in an induced diamond, outputs a graph $G^\prime$ with $\tilde{\Theta}(n)$ vertices, $\tilde{\Theta}(m)$ edges, $\tilde{\Theta}(t)$ colorful induced diamonds, and $\tilde{\Theta}(x)$ vertices that participate in a colorful induced diamond, with probability $1-O(1/n^8)$.
\end{lemma}
We describe the reduction and prove that it satisfies the desired properties \whp.
\begin{mdframed}[frametitle={The Reduction}, frametitlealignment=\centering, backgroundcolor=gray!10]
    \begin{enumerate}
        \item Take $r$ copies of $G$ denoted by $(G_1,\ldots,G_{r})$, where $r=100\log n$.
        \item Sample $r$ uniformly random colorings $\varphi_i:V(G_i)\to [4]$. \Comment{$\varphi_i$ is a random function that maps each vertex in $G_i$ to a color in $[4]$ uniformly at random.}
        \item Let the new graph be $G^\prime \triangleq \bigsqcup_{i=1}^r G_i$,
        and define a coloring $\phi:V(G^\prime)\to[4]$ as $\phi(v) = \varphi_i(v)$ for every $v \in V(G_i)$ and every $i \in [1,r]$.
    \end{enumerate}
\end{mdframed}

To prove \Cref{lemma:reduction} we use the following claim.
\begin{claim}\label{claim:z0}
    For every $i \in [1,r]$, denote by $t_i$ and $x_i$ the number of colorful induced diamonds and vertices that participate in a colorful induced diamond in $G_i$, respectively. Then, the following inequalities hold for every $i \in [1,r]$:
    \begin{enumerate}
        \item $\Pr{t_i\geq t/100}\geq 1/12$.
        \item $\Pr{x_i\geq x/100}\geq 1/12$.
    \end{enumerate}
\end{claim}
\begin{proof}[Proof of \Cref{claim:z0}]
    First, we prove that for every $i$, we have that $\Pr{t_i\geq t/100}\geq 1/12$ using \Cref{lem:rev mark}.
    We compute $\Exp{t_i}$. It suffices to compute $\Exp{t_1}$, as all $t_i$ are identically distributed.
    Let $Z_j$ be the indicator random variable that is equal to $1$ if the $j^{\text{th}}$ diamond in $G$ is colorful with respect to $\varphi_1$, and $0$ otherwise.
    We have that $t_1=\sum_{j=1}^t Z_j$, and therefore $\Exp{t_1}= \sum_{j=1}^t \Exp{Z_j}= t\cdot \Pr{Z_j=1}$.
    The probability that a fixed diamond is colorful with respect to $\varphi_1$ is $q\triangleq\frac{4!}{4^4}=\frac{6}{64}=\frac{3}{32}$.
    We therefore have that $\Exp{t_1}=q\cdot t$.
    We apply \Cref{lem:rev mark} (Reverse Markov's inequality) with $R=t/100$ and $M=t$ to get:
    \[
        \Pr{t_1\geq \frac{t}{100}}
        \geq \frac{\Exp{t_1}-\frac{t}{100}}{t-\frac{t}{100}}
        =    \frac{q\cdot t-\frac{t}{100}}{t(1-\frac{1}{100})}
        =    \frac{q-\frac{1}{100}}{1-\frac{1}{100}} \geq q-\frac{1}{100}\geq \frac{1}{12}\;.
    \]
    This completes the proof of the first part of \Cref{claim:z0}.
    The proof that $\Pr{x_i \geq \frac{x}{100}} \geq 1/12$ is nearly identical and is thus omitted.
\end{proof}

\begin{proof}[Proof of \Cref{lemma:reduction}]
    Clearly, $\Gp$ contains $\tilde{\Theta}(n)$ vertices, $\tilde{\Theta}(m)$ edges, $\TO{t}$ diamonds and $\TO{x}$ vertices that participate in a diamond.
    We show that $\Gp$ contains $\tilde{\Omega}(t)$ colorful induced diamonds and $\tilde{\Omega}(x)$ vertices that participate in a colorful induced diamond, with probability at least $1-O(1/n^8)$.

    We show that with high probability, there exist some $i,j \in [1,r]$ such that $t_i \geq t/100$ and $x_{j} \geq x/100$.
    Formally, we use $\EE_t$ to denote the event that there exists $i\in[1,r]$ such that $t_i\geq t/100$. Similarly, we define $\EE_x$ as the event that there exists $j\in[1,r]$ such that $x_j\geq x/100$. Since the colorings are independent, we get using \cref{claim:z0} that:

    \[
        \Pr{\EE_t} = 1- \prod_{i=1}^r \Pr{t_i < t/100} \geq 1 - (11/12)^r  \geq 1-1/n^8
    \]
    and similarly $\Pr{\EE_x} \geq 1-1/n^8$. Hence, we get:
    \[
        \Pr{\EE_t\land \EE_x}\geq \Pr{\EE_t}+\Pr{\EE_x}-1 \geq 1-2/n^8\;.
    \]
    This shows that in $\Gp$, there are $\Omega(t)$ colorful induced diamonds and $\Omega(x)$ vertices that participate in such diamonds. The reduction clearly takes $\TO{n+m}$ time, as it requires copying the graph $G$ for $r$ times and sampling $r$ random colorings.
\end{proof}

\begin{remark}
    In \Cref{ssec:main-light} we use the exact same reduction to get a graph $\Gp$ that is colored with $3$ colors instead of $4$, using a random $3$-coloring instead of a random $4$-coloring.
\end{remark}

We now provide an algorithm for detecting a colorful induced diamond in a $4$-colored graph, which we apply on the output of the reduction in \Cref{lemma:reduction}.
We provide an adaptation of the randomized diamond detection algorithm by \cite{williams2014finding} to the colorful and unbalanced setting.
The running time depends on the sizes $n_1 \geq n_2 \geq n_3 \geq n_4$ of the color classes.

\begin{lemma}\label{lem:colorful-det}
    There exists an algorithm $\algCD$ that determines whether a $4$-colored graph $G$ on $n$ vertices contains a colorful induced diamond in time $\TO{\MM{n_1,n_2,n_3}}$. The algorithm succeeds \whp.
\end{lemma}

\paragraph{Overview of the algorithm of \cite{williams2014finding}.}
Let $A$ denote the adjacency matrix of $G$. The algorithm begins by computing $Z(G)$, defined as the sum of the number of pairs of common neighbors for every edge.
As observed by Kloks--Kratsch--M\"uller~\cite{kloks2000finding}, this sum relates to the number of $K_4$'s and induced diamonds in $G$, as follows:
\[
    Z(G) = \sum_{(u,v)\in E} \binom{A^2[u,v]}{2} = 6\cdot \#K_4 + \#\mathsf{Diamonds}\,.
\]
This implies that if $Z(G) \not\equiv 0 \pmod 6$, the graph must contain an induced diamond. However, if $Z(G)$ is divisible by $6$, the test is inconclusive (e.g., the graph might contain $6$ diamonds). 
To resolve this, random noise is introduced by subsampling the vertices. 
A key lemma in \cite{williams2014finding} establishes that if diamonds exist, the count in a random subgraph will not be $0 \pmod 6$ with constant probability. 
The algorithm repeats this process $O(\log n)$ times to achieve high probability. 
The result is a randomized algorithm with one-sided error; the algorithm never reports a diamond if none exists, and if one does, it detects it with high probability.
\medskip
\noindent

Suppose now that the graph is $4$-colored as in the statement of \Cref{lem:colorful-det}. To prove the lemma, we define a colored variant of $Z(G)$ as follows.
Recall that $E$ denotes the edge set of $G$, and let $E_\phi$ denote the set of edges whose endpoints have different colors.
Consider an edge $e=(u,v)\in E_\phi$, where $u$ and $v$ have distinct colors $a,b \in [4]$. Let $c,d \in [4] \setminus \{a,b\}$ be the two remaining distinct colors.
We define $k_c(e)$ as the number of common neighbors of $u$ and $v$ that have color $c$. 
Let $y_e \triangleq k_{c}(e)\cdot k_{d}(e)$ be the number of pairs of common neighbors of $e$ with colors $c$ and $d$.
We define $Z_\phi(G)$ to be the sum of $y_e$'s over all edges:
\[
    Z_\phi(G) = \sum_{e\in E_{\phi}} y_e\;.
\]
The following claim is analogous to the observation in \cite{kloks2000finding}.

\begin{claim}\label{claim:colored exact counting}
    Let $s$ denote the number of colorful $K_4$ instances and $t$ denote the number of colorful induced diamonds in a $4$-colored graph $G$. Then, $Z_\phi(G) = 6s + t$.
\end{claim}

\begin{proof}[Proof of \Cref{claim:colored exact counting}]
    By definition, $y_e$ counts the pairs of common neighbors of $e=(u,v)$ whose colors are distinct from each other and from the colors of $u$ and $v$.
    Any such pair $\set{w,z}$ forms a colorful set of vertices $\set{u,v,w,z}$.
    If $w$ and $z$ are adjacent, they induce a colorful $K_4$. This $K_4$ contains $6$ edges; for each edge, the other two vertices form a valid pair, so the $K_4$ contributes exactly $6$ to the sum.
    If $w$ and $z$ are non-adjacent, they induce a colorful diamond with $(u,v)$ as the chord (the edge connecting the degree-$3$ vertices). A diamond has only one chord. For any other edge in the diamond, e.g. $e = (u,w)$, the diamond induced by $\set{u,v,w,z}$ is not counted in $y_e$ because $z$ is not a common neighbor of $u$ and $w$ in this diamond. Specifically, in a diamond, only the chord sees two common neighbors. Thus, every colorful induced diamond is counted exactly once.
\end{proof}

\begin{claim}[Algorithm $\algA$]\label{claim:alg for computing Z}
    There is an algorithm $\algA$ that, given a $4$-colored graph $G$, with color classes of sizes $n_1 \geq n_2 \geq n_3 \geq n_4$,
    computes $Z_\phi(G)$ in $O(\MM{n_1,n_2,n_3})$ time.
\end{claim}

\begin{proof}[Proof of \Cref{claim:alg for computing Z}]
    We first compute the $k_c(e)$ values for every edge $e \in E(G)$ and color $c \in [4]$ as follows.
    For distinct colors $i,j \in [4]$, let $A_{ij}$ be the $n_i \times n_j$ submatrix of $A$ representing edges between color classes $i$ and $j$.
    For an edge $e=(u,v)$ with colors $a,b$, the value $k_c(e)$ is the entry corresponding to $(u,v)$ in the product $A_{ac} A_{cb}$.
    We compute these products for all permutations of colors. The total time is bounded by the sum of matrix multiplication costs for all triplets of sizes. Since $\MM{\cdot}$ is symmetric and monotonic, the cost is dominated by the product of the three largest color classes:
    \[
        \sum_{a} \sum_{b \neq a} \sum_{c \neq a,b} \MM{n_a,n_c,n_b} = O(\MM{n_1, n_2, n_3})\;.
    \]
    Once these values are computed, $Z_\phi(G)$ is obtained by summing $y_e$ for all $e\in E_\phi$ in $O(n^2)$ time.
\end{proof}

\begin{proof}[Proof of \cref{lem:colorful-det}]
    The algorithm $\algCD$ is defined below.
    \begin{mdframed}[frametitle={\normalfont Algorithm $\algCD(G)$}, frametitlealignment=\centering, backgroundcolor=gray!10]
        \textbf{Input:} A $4$-colored graph $G$. \\
        \textbf{Output:} \yes if $G$ contains a colorful induced diamond, \no otherwise.

        \medskip
        \noindent
        Repeat $100\log n$ times:
        \begin{enumerate}[leftmargin=1.5cm]
            \item Sample a subset $U \subseteq V(G)$ by including each vertex with probability $1/2$.
            \item Invoke $\algA(G[U])$ to compute $Z_\phi(G[U])$.
            \item If $Z_\phi(G[U]) \not\equiv 0 \pmod 6$, output \yes and terminate.
        \end{enumerate}
        Output \no.
    \end{mdframed}
    \textbf{Running Time Analysis:} The algorithm runs for $\TO{1}$ iterations. Since $G[U]$ is an induced subgraph of $G$, computing $Z_\phi(G[U])$ using $\algA$ is not slower than computing $Z_\phi(G)$. By \Cref{claim:alg for computing Z}, this takes $O(\MM{n_1,n_2,n_3})$ time, where $n_1 \geq n_2 \geq n_3 \geq n_4$ are the color class sizes of $G$.

    \noindent
    \textbf{Correctness:}
    The correctness of computing $Z_\phi(G[U])$ using $\algA$ follows from \Cref{claim:alg for computing Z}.
    If $G$ contains no colorful diamonds ($t=0$), then neither does $G[U]$ for any $U$. By \Cref{claim:colored exact counting}, $Z_\phi(G[U]) = 6s'$ for some $s'$.
    Since $Z_\phi(G[U]) \equiv 0 \pmod 6$, the algorithm always outputs \no.

    If $G$ contains colorful diamonds ($t > 0$), we utilize the polynomial method from \cite{williams2014finding}. Let $P(x_1,\dots,x_n)$ be the polynomial:
    \[
        P(x_1,\ldots,x_n) \triangleq \sum_{\substack{u < v < w < y \\ \text{is a colorful induced diamond}}} x_u x_v x_w x_y.
    \]
    Let $(a_1,\ldots,a_n) \in \set{0,1}^n$ be the characteristic vector of $U$. Then $P(a_1,\ldots,a_n)$ is the number of colorful induced diamonds in $G[U]$.
    In particular, if $t > 0$ then $P$ is a non-zero multilinear polynomial of degree $4$. By \cite[Lemma 2.2]{williams2014finding}, we get:
    \[
        \mathrm{Pr}_{(a_1,\ldots,a_n) \in \set{0,1}^n}\left[P(a_1,\ldots,a_n) \not\equiv 0 \pmod 6 \right] \geq 1/2^4,
    \]
    which implies that with probability at least $1/16$ over the choice of $U \subseteq V(G)$, $Z_\phi(G[U]) \not\equiv 0 \pmod 6$. Hence, if $t > 0$, the success probability of the algorithm is at least $1-(15/16)^{100 \log n} \geq 1-1/n^{8}$.
\end{proof}

\subsection{The Algorithm \algSensitive}
\label{ssec:subsampling}
\newcommand{\rag}{100\log n}
In this subsection, we present the algorithm \algSensitive, which is the main algorithm we use in this section.
This algorithm follows the subsampling framework of~\cite{CEW24,CEW25} for witness-sensitive induced subgraph detection.
The framework needs two components:
a colorful input graph, i.e., a graph with a $4$-coloring and many colorful induced diamonds, and an algorithm for colorful induced diamond detection. Both components were provided in the previous subsection.
However, as a black box the subsampling framework yields no improvement over the $\tilde O(n^\omega)$ bound already achieved by~\cite{williams2014finding}.
To overcome this limitation, we develop a refined analysis for detecting $r$-light diamonds.
In the next two subsections we refine the analysis by crucially exploiting the additional structure in the graph: namely, that no large clique contains three vertices of a diamond.

\paragraph{The Subsampling Framework.}
A \emph{sampling vector} is a vector $P=(p_1,p_2,p_3,p_4) \in [0,1]^4$. A $4$-colored graph $G$ is sampled using $P$ to obtain a subgraph $H$ by keeping each vertex $v\in V_i$ with probability $p_i$, independently. We denote this by $H\gets G[P]$.
Let $\FF \triangleq \{2^{-j} \mid 0 \le j \le \lceil \log n \rceil\}$ be the set of inverse powers of two, and let $\FF^4$ be the set of all sampling vectors with entries in $\FF$.
The weight of a sampling vector $P$ is defined as $w(P)\triangleq \prod_{i=1}^4 p_i$. We also use $\wc(P)$ to denote the product of the three largest coordinates of $P$, or equivalently, $\wc(P) = \frac{w(P)}{\min(P)}$, where $\min(P) = \min\{p_1,p_2,p_3,p_4\}$.
Let
\[
    \alpha(P)\triangleq \mathrm{Pr}_{H \gets G[P]}\left[{H \text{ contains a colorful induced diamond}}\right]\;,
\]
be the probability that the sampled subgraph $H\gets G[P]$ contains a colorful induced diamond, when sampled according to $P$.
The algorithm \algSensitive is as follows:
\begin{mdframed}[frametitle={\normalfont Algorithm $\algSensitive(G)$}, frametitlealignment=\centering, backgroundcolor=gray!10]
    \textbf{Input:} An $n$-vertex graph $G$ \\
    \textbf{Output:} \yes if $G$ contains an induced diamond, \no otherwise.

    \medskip
    \noindent
    \begin{enumerate}[leftmargin=2em]
        \item (Color coding) Apply the algorithm from \cref{lemma:reduction} on $G$ to obtain a $4$-colored graph $G'$.
        \item (Subsampling) For every $P \in \FF^4$, create $\ell = 2000 \log^2 n$ independent random subgraphs $H_1^P, \ldots, H_\ell^P$, where $H_i^P \gets G'[P]$.
        \item Execute $\algCD(H_i^P)$ for all $P \in \FF^4$ and $i \in [\ell]$. This is done in parallel in a round-robin fashion, terminating immediately if any execution outputs \yes.
        \item If no execution outputs \yes, output \no.
    \end{enumerate}
\end{mdframed}

To analyze the running time of the algorithm, we need the following claim:
\begin{claim}\label{claim:runtime w3}
    The running time of $\algCD(H)$ where $H\gets G[P]$ is $\To(\MM{n,n,n\cdot \wc(P)})$ \whp.
\end{claim}
We need a few more definitions before bounding the running time of \algSensitive.
Let
\begin{align}
    \PP\triangleq\set{P\in \FF^4 \mid \alpha(P)\geq \frac{\rag}{\ell}}\;, \quad \text{and} \quad \wc(\PP) \triangleq \min\set{\wc(P)\mid P\in \PP }
    \;. \label{eq:PP}
\end{align}

The main result of this subsection is the following proposition.
\begin{restatable}{proposition}{propClaimRuntime}\label{claim:runtime}   
    The running time of the algorithm \algSensitive is $\To(\MM{n,n,n\cdot \wc(\PP)})$ \whp. The output is correct \whp.
\end{restatable}

\begin{proof}[Proof of \Cref{claim:runtime} using \Cref{claim:runtime w3}]
    Correctness follows from the one-sided error of $\algCD$.
    Suppose $t > 0$ and fix $P \in \PP$.
    By \Cref{claim:runtime w3}, each $\algCD(H_i^P)$ runs in time $\TO{\MM{n,n,n\cdot \wc(P)}}$ \whp.
    Since $\alpha(P) \geq \frac{\rag}{\ell}$, 
    the probability that no $H_i^P$ contains a colorful induced diamond is at most
    \[\left(1 - \frac{\rag}{\ell}\right)^\ell \leq e^{-\rag} = \frac{1}{n^{100}}\;.\]
    Therefore, \whp at least one $H_i^P$ contains a colorful induced diamond, in which case $\algCD(H_i^P)$ outputs \yes.

    Since the execution of the algorithms stops as soon as one of them outputs \yes, the total running time is dominated by the running time of $\algCD(H_{i}^P)$ times the number of executions, which is $\ell \cdot |\FF^4| = \TO{1}$. Therefore, \whp, the total running time is $\TO{\MM{n,n,n\cdot \wc(P)}}$.
    Since this holds for every $P\in \PP$, the running time is $\To(\MM{n,n,n\cdot \wc(\PP)})$ \whp. 
    This follows from a union bound over all $P\in \PP$, as $|\PP| = \To(1)$.
\end{proof}

\begin{proof}[Proof of \Cref{claim:runtime w3}]
    Let $P=(p_1,p_2,p_3,p_4)$, and assume without loss of generality that $p_1\geq p_2\geq p_3\geq p_4$.
    Let $H\gets G[P]$, and let $n_i$ be the size of the $i$th color class of $H$. 
    Standard concentration bounds imply that \whp, for every $i\in[4]$, $n_i\leq \hat{n}_i$, where $\hat{n}_i = 100 n p_i \cdot \log n$.
    Let $\EE$ be the event that $n_i\leq \hat{n}_i$, for every $i\in[4]$.
    Assuming $\EE$ holds, by \Cref{lem:colorful-det} the running time of $\algCD(H)$ is 
    \begin{align*}
        \MM{\hat{n}_1,\hat{n}_2,\hat{n}_3} 
        &\leq \MM{100 n p_1 \log n, 100 n p_2 \log n, 100 n p_3 \log n} \\
        &\leq \TO{\MM{n p_1, n p_2 , n p_3 }}\\
        &\overset{(\star)}\leq \TO{\MM{n,n,n\cdot (p_1p_2p_3)}}\\
        &= \TO{\MM{n,n,n\cdot \wc(P)}}\;,
    \end{align*}
    where $(\star)$ follows from the monotonicity of rectangular matrix multiplication (see \Cref{claim:omega balanced}).
\end{proof}

Previous work~\cite{CEW24,CEW25} showed that there exists $P\in\PP$ with $w(P) = \To(1/t)$.
However, if $P=(1,1,1,1/t)$, then while $w(P) = 1/t$, we have $\wc(P) = 1$, yielding no speedup. 
Our goal for the rest of this section is to prove that $\PP$ contains a sampling vector $P$ with $\wc(P) = \To(\sqrt{r/t})$.
From that we immediately get a detection algorithm with running time $\To(\MM{n,n,n\sqrt{r/t}})$, as stated in \Cref{thm:r-light diamonds}.
For the rest of this section we prove the existence of such a sampling vector.

\subsection{Main Result: Proof of \Cref{thm:r-light diamonds}}
\label{ssec:main-light}
\renewcommand{\cg}{100\log n}

\newcommand{\cliquebound}{\sqrt{t}}
\newcommand{\ts}{\sqrt{t}}
\newcommand{\tsr}{\sqrt{r/t}}
\newcommand{\tsi}{\frac{1}{\sqrt{t}}}
\newcommand{\Ua}{U_{1}}
\newcommand{\Ub}{U_{2}}
\newcommand{\hUa}{\hat{U}_{1}}
\newcommand{\hUb}{\hat{U}_{2}}
\newcommand{\Sa}{S_{1}}
\newcommand{\Sb}{S_{2}}
\newcommand{\Gx}{G_{\psi}}
\renewcommand{\Gp}{G_{\phi}}
\newcommand{\HHb}{\HH^{(2)}}
\newcommand{\HHc}{\HH^{(3)}}
\newcommand{\DDr}{\DD_{r}}
\newcommand{\ra}{\log n}

Recall that an induced diamond is $r$-light if no three vertices of the diamond are contained in an $r$-clique, and $r$-heavy otherwise.
Let $\rmax$ denote the largest integer $r$ such that $G$ contains an $r$-heavy induced diamond.
The main result of this section is the following:
\ThmLightD* 

We prove a slightly stronger version of \Cref{thm:r-light diamonds}:
\begin{theorem}\label{thm:rr-light}
    For every $r$, let $t_r$ denote the number of $r$-light induced diamonds in $G$.
    There exists an algorithm that, given an $n$-vertex graph $G$, detects an induced diamond in $G$ \whp, running in time $\To(\MM{n,n,n \cdot \sqrt{r/t_{r}}})$.
\end{theorem}
By applying \Cref{thm:rr-light} with $r=\rmax+1$, we get \Cref{thm:r-light diamonds}, since every induced diamond in $G$ is $(\rmax+1)$-light.

To prove \Cref{thm:rr-light}, we fix a specific $r$ and use $\DDr$ to denote the set of all $r$-light induced diamonds in $G$. 
Instead of working directly with $G$, we first apply the reduction from \Cref{ssec:reduction} to get a colorful graph with many colorful $r$-light diamonds.
Define a new graph $\Gx$ by taking $\TO{1}$ independent copies of $G$ with a random $3$-coloring (not a $4$-coloring) $\psi: V(\Gx)\to[3]$, as we did in \Cref{ssec:reduction} with a $4$-coloring.
Let $\DD_\psi$ be the subset of diamonds in $\Gx$ that correspond to diamonds in $\DDr$, where a diamond $D=(v_1,v_2,v_3,v_4)$ with missing edge $(v_3,v_4)$ belongs to $\DD_\psi$ if and only if $\psi(v_1)=1$, $\psi(v_2)=2$, and $\psi(v_3)=\psi(v_4)=3$. We refer to such diamonds as \emph{colorful} with respect to $\psi$.
We use $t=|\DD_\psi|$, where $t=\Omega(|\DDr|)$ by the same argument as in \Cref{ssec:reduction}, i.e., reverse Markov's inequality.

Later on, we extend $\psi$ to a coloring $\phi: V(\Gx)\to[4]$ randomly: every vertex $v\in V(\Gx)$ with $\psi(v)=3$ is assigned a random color $\phi(v)\in\{3,4\}$ independently.
The reason for this two-step coloring is that we want to define several hypergraphs based on $\psi$ first, and reveal the fourth color only later.
We define $\DD_\phi$ to denote the set of colorful diamonds in $\Gx$ with respect to $\phi$, and $E_\phi$ to denote the set of edges in $G$ whose endpoints have distinct colors with respect to $\phi$.

To prove \Cref{thm:rr-light}, we need the following theorem. Recall that
\[
    \alpha(P,\Gp)=\Pr{\Gp[P] \text{ contains a colorful diamond}}\;.
\]
\begin{theorem}
    \label{thm:helper}
    At least one of the following holds:
    \begin{enumerate}[label=(\arabic*)]
        \item There exists $P\in\FF^4$ such that $\wc(P)\leq \TO{\sqrt{r/t}}$ and $\alpha(P,\Gp)\geq \tg(1)$.
        \item The number of \degc vertices satisfies $x_3 \geq  \tg(\sqrt{t/r})$.
    \end{enumerate}
\end{theorem}
To handle the second case of \Cref{thm:helper}, we use the following theorem, which follows from a modification of the algorithm of \cite{williams2014finding}:
\begin{restatable}[\algVertexC]{theorem}{ThmFindDegc}\label{thm:find degc}
    There exists a randomized algorithm \algVertexC that detects a \degc vertex \whp in time $\To(\MM{n,n,n/x_3})$, where $x_3$ is the number of \degc vertices.
\end{restatable} 
The proof of this theorem is deferred to \Cref{app:missing-proofs}. 

\begin{proof}[Proof of \Cref{thm:rr-light} using \Cref{thm:helper}]
    Assume $(1)$ holds.
    Then we can sample a random induced subgraph $H\gets \Gp[P]$ and detect a diamond in $H$ using \algCD.
    $H$ has a colorful induced diamond with probability $\tg(1)$, since $\alpha(P,\Gp)\geq \tg(1)$. By \Cref{claim:runtime w3}, the detection runs in time $\TO{\MM{n,n,n\cdot \wc(P)}}=\TO{\MM{n,n,n\cdot \sqrt{r/t}}}$ \whp. By repeating the process $\TO{1}$ times, we find a colorful diamond in $\Gp$ \whp which proves the theorem.

    Assume $(2)$ holds. Then $x_3=\tg(\sqrt{t/r})$. Hence, we can use \Cref{thm:find degc} to detect an induced diamond in $G$ (ignoring the colors) in time $\TO{\MM{n,n,n/x_3}}=\TO{\MM{n,n,n\cdot \sqrt{{r}/{t}}}}$ \whp.
\end{proof}
We emphasize that we do not need to explicitly find such $P$; if such $P=(p_1,p_2,p_3,p_4)$ exists, then there is a sampling vector $P'=(p_1',p_2',p_3',p_4')$, where $p_i'$ is the smallest power of $2$ larger than or equal to $p_i$, for every $i\in[4]$, and clearly $\alpha(P',\Gp)\geq \alpha(P,\Gp)$ and $\wc(P')=\TO{\wc(P)}$.
Since the algorithm \algSensitive considers all sampling vectors with coordinates that are powers of $2$, it will consider $P'$ as well.

For the rest of this section, we prove \Cref{thm:helper}.
We define a sequence of hypergraphs, where
$V_i$ denotes the set of vertices with color $i$ in $\psi$, for every $i\in[3]$, and $\DD_{\psi}$ is the hyperedge set. We use $U_3$ and $U_4$ to denote the set of vertices with color $3$ and $4$ in $\phi$, respectively.
The final object $\Gp$ is defined using the coloring $\phi$, with edge set $E_{\phi}$ consisting of the edges whose endpoints have distinct colors.
Let:
\begin{align*}
    \GG^{(2)} & =((V_1 \times V_2) \sqcup (V_3 \times V_3),\DD_{\psi})\;, \\
    \GG^{(3)} & =((V_1 \times V_2) \sqcup V_3 \sqcup V_3,\DD_{\psi})\;,   \\
    \GG^{(4)} & =(V_1 \sqcup V_2 \sqcup V_3 \sqcup V_3,\DD_{\psi})\;,     \\
    \Gp       & =(V_1 \sqcup V_2 \sqcup U_3 \sqcup U_4,E_{\phi})\;.
\end{align*}
Note that we abuse the notation and use $\DD_{\psi}$ to denote the hyperedges in all hypergraphs. The two appearances of $V_3$ in a disjoint union are treated as separate copies.
Formally, $\GG^{(2)}$ is a bipartite graph with parts $A = V_1 \times V_2$ and $B = V_3 \times V_3$. An edge connects $(v_1,v_2) \in A$ to $(v_3,v_4) \in B$ if $(v_1,v_2,v_3,v_4)$ is a diamond in $G$ with $\psi(v_1) = 1$, $\psi(v_2)=2$, and $\psi(v_3)=\psi(v_4)=3$.
The hypergraph $\GG^{(3)}$ is the $3$-partite hypergraph on vertex set $V_1 \times V_2 \sqcup V_3 \sqcup V_3$, with a hyperedge connecting $(v_1,v_2),v_3,v_4$ for every diamond $(v_1,v_2,v_3,v_4)$ as before. The hypergraph $\GG^{(4)}$ is the $4$-partite hypergraph on vertex set $V_1 \sqcup V_2 \sqcup V_3 \sqcup V_3$, with a hyperedge connecting $v_1,v_2,v_3,v_4$ for every diamond $(v_1,v_2,v_3,v_4)$ as before.
Finally, $\Gp$ is the $4$-partite \emph{graph} on vertex set $V_1 \sqcup V_2 \sqcup U_3 \sqcup U_4$ and edges $E_\phi$ defined to be the subset of edges of $G$ whose endpoints have distinct colors.

For every $i\in[4]$, let $\alpha(P,\GG^{(i)})$ denote the probability that $\GG^{(i)}[P]$ contains a hyperedge, where $P$ is a sampling vector of dimension $i$. E.g., $\alpha(P,\GG^{(2)})$ is the probability that, when subsampling from the vertex parts of $\GG^{(2)}$ using vector $P = (p,q)$, an edge survives the subsampling. 
We sometimes add the superscript $(i)$ to $P$ to emphasize that it is an $i$-dimensional vector, e.g., $P^{(4)} = (p_1,p_2,p_3,p_4)$.
Recall that our final goal is to show the existence of a sampling vector $P$ with small $\wc(P)$
for $\Gp$ with $\alpha(P,\Gp) = \tg(1)$ (or to show that $x_3$ is large).

Before proving \Cref{thm:helper}, we explain the main idea behind the refinement of sampling vectors (see \Cref{fig:refine-sequence}).
Let $\XX=\XX_1\times\XX_2\times\cdots\times \XX_k$ be a Cartesian product of $k$ sets, and let $\UU\subseteq \XX$ be a subset of $\XX$.
Here $\XX$ is the set of all quadruples of vertices, and $\UU$ is the set of all colorful induced diamonds.
By sampling each element in $\XX$ independently with probability $q$, we get a $1$-dimensional sampling vector, which hits $\UU$ with probability $1- (1-q)^{|\UU|}$.
A $4$-dimensional sampling vector samples each element in $\XX_i$ independently with probability $p_i$, for every $i\in[4]$.
The work of \cite{CEW25} shows that there exists a sampling vector $P=(p_1,p_2,p_3,p_4)$ such that the sampled set contains at least one element from $\UU$ with probability $\tg(1)$ and $w(P)=\To(\frac1{|\UU|})$:
\begin{restatable}[{\cite[Lemma 22]{CEW25}}]{lemma}{LemDisc}\label{lemma:discovery}
    Let $\HH$ be a $k$-partite hypergraph with $n$ vertices and $m$ hyperedges.
    Then, there exists a \emph{simple} vector $P\in[0,1]^k$ such that $m\cdot w(P)\leq \cclog$, for which the random induced graph $C\gets \HH[P]$ contains at least one hyperedge with probability at least $(\eo)^{k}$.
\end{restatable}
Here, instead of refining a $1$-dimensional sampling vector into a $k$-dimensional sampling vector, we refine a $k$-dimensional sampling vector into a $(k+1)$-dimensional sampling vector, which requires a slightly more delicate analysis following the approach of \cite{CEW24,CEW25}. 
See \Cref{thm:refining sampling vector} for the formal statement.
This however is still too weak to get a faster running time.
We explain how to refine $P^{(1)}=(1/t)$ gradually into $P^{(4)}=(p_1,p_2,p_3,p_4)$ while controlling $\wc(P^{(4)})$, as depicted in \Cref{fig:refine-sequence}.

\begin{figure}[H]
    \centering
    \begin{tikzpicture}
        \node[draw, inner sep=4pt, rounded corners=3pt] (P1) {$P^{(1)}=\brak{\frac{1}{t}}$};

        \node[draw, inner sep=4pt, rounded corners=3pt, right=2cm of P1] (P2) {$P^{(2)}=(q_1,q_2)$};
        \node[above=2pt of P2] {\small $q_1\cdot q_2=\lambda/t$};

        \node[draw, inner sep=4pt, rounded corners=3pt, right=2cm of P2] (P3) {$P^{(3)}=(q_1,p,p)$};
        \node[above=2pt of P3] {\small $p=\sqrt{rq_2}$};

        \node[draw, inner sep=4pt, rounded corners=3pt, right=2cm of P3] (P4) {$P^{(4)}=(p_1,p_2,p,p)$};
        \node[above=2pt of P4] {\small $p_1\cdot p_2= \lambda\cdot q_1$};

        \node[draw, inner sep=4pt, rounded corners=3pt, below=1cm of P4, red] (P5) {$\hat{P}^{(4)}=(1,1,1,1/t)$};

        \draw[->, thick] (P1) -- node[above] {\footnotesize $1^{\text{st}}$ refinement} (P2);
        \draw[->, thick] (P2) -- node[above] {\footnotesize $2^{\text{nd}}$ refinement} (P3);
        \draw[->, thick] (P3) -- node[above] {\footnotesize $3^{\text{rd}}$ refinement} (P4);
        \draw[->, thick, red, dashed, bend right=10] (P1.south) to node[pos=0.6, xshift=7.5em, yshift=.8em, text width=6cm] {\small {\color{black!80!white} Using \Cref{lemma:discovery} as a black box}} (P5.west);
    \end{tikzpicture}
    \caption{Refinement sequence from $P^{(1)}$ to $P^{(4)}$. We set $\lambda=48\log n$.}
    \label{fig:refine-sequence}
\end{figure}
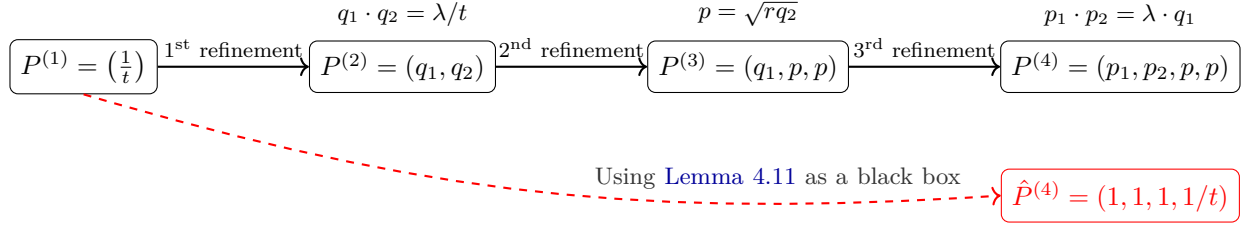
We explain how each refinement step is done. Each step follows the same ``recipe'', with slight deviations to get better control over $\wc(P)$.
\paragraph{From $P^{(1)}$ to $P^{(2)}$:}
Consider the bipartite graph $\GG^{(2)}$ with vertex set $A\sqcup B$ and $t$ edges.
Let $A_i$ denote the set of vertices in $A$ with degree in $[2^i,2^{i+1})$, for every $0\leq i\leq \log t$.
One of these sets must contain at least $t/\log t$ edges, let $A_i$ be such a set.
We work with the edge set of $E(A_i)$ instead of the entire edge set of $\GG^{(2)}$.
Define a sampling vector $P^{(2)}=(q_1=\frac{\lambda'\cdot 2^i }{t},q_2=\frac{4}{2^i})$, where $\lambda'=12\log n$. 
Instead of proving that $\alpha(P^{(2)},\GG^{(2)})\geq \tg(1)$,
following the same lines of the proof of \Cref{thm:refining sampling vector}, we prove that $\alpha(P^{(3)},\GG^{(2)}[E(A_i)])\geq \tg(1)$, directly.
We do not use \Cref{thm:refining sampling vector} to get $P^{(2)}$ as a black box, because we need to control the degrees in $A_i$ to further refine $P^{(2)}$ to $P^{(3)}$.
\paragraph{From $P^{(2)}$ to $P^{(3)}$:}
Here, we crucially use the fact that the diamonds we consider are $r$-light. 
This means that for any pair $(v_1,v_2)\in V_1\times V_2$, if $(v_1,v_2,v_3,u_i)$ is an $r$-light diamond for every $i=1,2,\ldots,r+2$, then there must be two indices $i\neq j$ such that $(v_1,v_2,u_i,u_j)$ is also an induced diamond, i.e., $(u_i,u_j)$ is a non-edge in $G$.
Our refinement increases the weight of the sampling vector, i.e.,
$w(P^{(3)})\geq w(P^{(2)})\cdot r$, which the general refinement in \Cref{thm:refining sampling vector} avoids, mostly because it is not clear how to ensure that this improves the probability of hitting a hyperedge.
\paragraph{From $P^{(3)}$ to $P^{(4)}$:}
Here we use the general refinement from \Cref{thm:refining sampling vector} to get $P^{(4)}$ from $P^{(3)}$.
This completes the overview of the refinement steps.

~\\
Before proving \Cref{thm:helper}, we state a lemma that allows us to refine 
$P^{(3)}$ into $P^{(4)}$.
This lemma is a special case of \Cref{thm:refining sampling vector}, which we 
prove in the next subsection.
\begin{restatable}[Refining Sampling Vectors]{lemma}{RefineSimple}\label{lemma:refine simple}
    Let $P^{(3)}=(q,p_3,p_4)$ be a sampling vector.
    Then, there exists a sampling vector $P^{(4)}=(p_1,p_2,p_3,p_4)$ with the following properties:
    \begin{enumerate}[label=(\arabic*)]
        \item $p_1\cdot p_2 = \To(q)$.
        \item $\alpha(P^{(4)},\GG^{(4)})\geq \alpha(P^{(3)},\GG^{(3)})/(2\log n)$.
    \end{enumerate}
\end{restatable}

~\\We proceed by proving several additional claims that we need for the proof of \Cref{thm:helper}.
The following claim shows that when using the same sampling vector $P^{(4)}$, hitting a diamond in $\Gp$ is at least as likely as hitting a hyperedge in $\GG^{(4)}$, up to a factor of $2$.
\begin{claim}\label{claim:xyz}
    $\alpha(P^{(4)},\Gp)\geq \frac{1}{2}\cdot \alpha(P^{(4)},\GG^{(4)})$.
\end{claim}
\begin{proof}[Proof of \Cref{claim:xyz}]
    Consider a diamond $D=(u_1,u_2,u_3,u_4)\in \DD_{\psi}$.
    Then the probability that it is also in $\DD_{\phi}$ is the probability that $\phi(u_3)\neq \phi(u_4)$, which is $1/2$.
    This completes the proof.
\end{proof}
By \Cref{claim:xyz}, it suffices to find $P^{(4)}$ such that $\alpha(P^{(4)},\GG^{(4)})\geq \tg(1)$, instead of $\alpha(P^{(4)},\Gp)\geq \tg(1)$ which is part of condition $(1)$ of \Cref{thm:helper}.
The following claim establishes a connection between $x_3$, the weight of a sampling vector $P^{(4)}$, and the probability that a diamond survives in $\GG^{(4)}[P^{(4)}]$.
\begin{claim}\label{claim2:z_1 and p12}
    For any $P^{(4)}=(p_1,p_2,p_3,p_4)$, it holds that $x_3\cdot p_i\geq \alpha(P^{(4)},\GG^{(4)})$, for $i=1,2$.
\end{claim}
\begin{proof}[Proof of \Cref{claim2:z_1 and p12}]
    Let $U_1$ and $U_2$ denote the set of vertices in $V_1$ and $V_2$, respectively, that participate in a hyperedge in $\DD_{\psi}$.
    By definition, every vertex in $U_1\cup U_2$ is a \degc vertex, so $|U_1|+|U_2|\leq x_3$.
    Define a random set $Y=V_1[p_1]\cup V_2[p_2]\cup V_3[p_3]\cup V_4[p_4]$, sampled according to the sampling vector $P^{(4)}$.
    Clearly the probability that $\GG^{(4)}[P^{(4)}]$ contains a hyperedge is equal to the probability that the induced subgraph $\GG^{(4)}[Y]$ contains a hyperedge.
    Define a random variable $Z=Y\cap V_1$.
    Clearly, if $Z$ does not contain any vertex from $U_1$, then $\GG^{(4)}[P^{(4)}]$ does not contain any hyperedge.
    Therefore,
    \begin{align*}
        \alpha(P^{(4)},\GG^{(4)}) & = \Pr{\GG^{(4)}[P^{(4)}] \text{ has a hyperedge}} \\
                                  & \leq \Pr{Z\cap U_1\nemp}                          \\
                                  & \leq \sum_{v\in U_1}\Pr{v\in Z}                   \\
                                  & = |U_1|\cdot p_1 \leq x_3\cdot p_1\;.
    \end{align*}
    The penultimate transition follows from the union bound.
    The same argument holds for $p_2$, completing the proof.
\end{proof}

~\\
We proceed by finding an induced subhypergraph of $\GG^{(2)}$ on the same vertex set (i.e., $V_1 \times V_2 \sqcup V_3 \times V_3$), where all vertices in $V_1 \times V_2$ have similar degree.
Recall that $t=|\DD_{\psi}|$, which is the number of hyperedges in $\GG^{(j)}$ for $j\in\set{2,3,4}$.
For every element $s\in V_1 \times V_2$ let $\deg^{(2)}(s)$ denote the degree of $s$ in the hypergraph $\GG^{(2)}$.
Partition the elements of $V_1 \times V_2$ into $\lceil\log t\rceil + 1$ buckets, where each bucket contains elements of similar degree in $\GG^{(2)}$:
\begin{align*}
    A_i = \set{s\in V_1 \times V_2:\deg^{(2)}(s)\in[2^{i-1},2^i)}\;.
\end{align*}
We say that a set $A_i$ is \emph{heavy} if $\abs{E(A_i)}\geq t/(6\log t)$,
where $E(A_i)$ is the set of hyperedges of $\GG^{(2)}$ that contain an element from $A_i$.
Note that there is at least one index $i$ such that $A_i$ is heavy. Otherwise, the total number of hyperedges would be less than $t/(6\log t)\cdot(\lceil\log t\rceil+1)<t$, which is a contradiction.
Let $i$ be the smallest index such that $A_i$ is heavy.
We restrict our attention to the subset of diamonds $\DD'_\psi$
defined as follows:
\begin{align*}
    \DD'_{\psi} =\set{(v_1,v_2,v_3,v_4)\in \DD_\psi \mid (v_1,v_2)\in A_i}\;.
\end{align*}
Let $\HH^{(j)}$ be the subhypergraph of $\GG^{(j)}$, with the same vertex set as $\GG^{(j)}$, and hyperedge set $\DD'_{\psi}$, for every $j\in[4]$.
Define
\begin{align}
    P^{(2)}=(q_1=\min(1,\frac{2^{i+1}\cdot \cg}{t}),q_2=\frac{1}{2^{i-1}})\;,\label{eq:P2}
\end{align}
where $i$ is the index of the chosen heavy set $A_i$. 
The following lemma allows us to prove \Cref{thm:helper}.
\begin{lemma}
    \label{prop2:balanced2 helper}
    Define $P^{(3)}=(q_1,p,p)$ where $p=\min(1,\sqrt{r\cdot q_2})$, and $q_1$ and $q_2$ are as in $P^{(2)}$.
    Then $\alpha(P^{(3)},\HH^{(3)})\geq \tg(1)$.
\end{lemma}

\begin{proof}[Proof of \Cref{thm:helper} using \Cref{prop2:balanced2 helper}]
    Let $P^{(3)}=(q_1,p,p)$ be as in \Cref{prop2:balanced2 helper}, meaning that $p=\min(1,\sqrt{r\cdot q_2})$, and additionally $q_1$ and $q_2$ are as in $P^{(2)}$.
    We refine $P^{(3)}$ into $P=(p_1,p_2,p,p)$ where $p_1\cdot p_2=\TO{q_1}$ using \Cref{lemma:refine simple}, which guarantees that $\alpha(P,\HH^{(4)})\geq \alpha(P^{(3)},\HH^{(3)})/(2\log n)$.
    By substituting the values of $q_1,q_2$ and $p$, we get
    \begin{align*}
        w(P)=\TO{w(P^{(3)})}=\TO{q_1\cdot p^2}\leq \TO{q_1\cdot r\cdot q_2}=\TO{r/t}\;.
    \end{align*}
    We have that
    \begin{align*}
        2\alpha(P,\Gp)
        \geq \alpha(P,\GG^{(4)})
        \geq \alpha(P,\HH^{(4)})
        \geq \frac{\alpha(P^{(3)},\HH^{(3)})}{2\log n}
        \geq \tg(1)\;.
    \end{align*}
    The inequalities follow for the following reasons.
    The first inequality follows by \Cref{claim:xyz}, the second follows because $\HH^{(4)}$ is a subhypergraph of $\GG^{(4)}$, the third follows by the guarantees of \Cref{lemma:refine simple}, and the last inequality follows by \Cref{prop2:balanced2 helper}.

    We show that one of the conditions of \Cref{thm:helper} holds.
    We use case analysis based on the smallest coordinate of $P$.
    \begin{itemize}
        \item Suppose first that $p_1$ (or $p_2$) is the smallest coordinate of $P$. If $p_1\leq \tsr$ or $p_2 \leq \tsr$, then by \Cref{claim2:z_1 and p12}, we have that $x_3\geq \tg(\sqrt{t/r})$, so (2) holds.
              Otherwise, we have $p_1,p_2\geq  \tg(\sqrt{r/t})$. Therefore,
              $\wc(P)=w(P)/p_1=\TO{\frac{r}{t}/ \sqrt{r/t}}=\TO{\tsr}$,
              as desired.
        \item Assume that the third or fourth coordinate of $P$ is the smallest coordinate of $P$. Both are equal to $p$.
              Therefore,
              \begin{align*}
                  \wc(P)=w(P)/p
                  =\To\left(\frac{r/t}{\sqrt{r\cdot q_2}}\right)
                  =\To\left(\frac{r/t}{\sqrt{r/2^{i-1}}}\right)
                  \leq \To\left(\frac{r/t}{\sqrt{r/t}}\right)
                  =\TO{\tsr}\;,
              \end{align*}
              as desired.
              The penultimate transition follows because $q_2 = 1/2^{i-1} \geq 1/t$.\qedhere
    \end{itemize}
\end{proof}

We are left with proving \Cref{prop2:balanced2 helper}. Recall that we sample from $\HH^{(3)}$ according to $P^{(3)}=(q_1,p,p)$.
We already showed that if we sample each vertex of $A_i$ with probability $q_1$ and each pair in $V_3\times V_3$ with probability $q_2$, then we hit a hyperedge in $\HH^{(2)}$ with constant probability.
We now replace the sampling of pairs in $V_3\times V_3$ with independent sampling of vertices in $V_3$ with probability $p=\min(1,\sqrt{r\cdot q_2})$ and show that we still hit a hyperedge in $\HH^{(3)}$ with constant probability.
In other words, each pair is sampled with probability $p^2=\min(1,r\cdot q_2)$, with the benefit of sampling vertices instead of pairs.

We define three random sets: $Y$ is $A_i[q_1]$, and $Z_1,Z_2$ are two independent copies of $V_3[p]$.
Let $\chi$ be the indicator random variable such that for every $U\subseteq V$ we have $\chi(U)=1$ if the induced hypergraph $\HH^{(3)}[U]$ contains a hyperedge, and $\chi(U)=0$ otherwise.
We have
\begin{align*}
    \alpha(P^{(3)},\HH^{(3)}) = \Pr{\chi(Y\cup Z_1\cup Z_2)=1}\;.
\end{align*}

For every subset $W\subseteq A_i$, we define two events:
\begin{enumerate}
    \item $\EE_1(W)$ is the event that $\set{Y=W}$, i.e., that the sampled elements from $A_i$ are exactly $W$.
    \item $\EE_2(W)$ is the event that $\set{\chi(W\cup Z_1\cup Z_2)=1}$.
\end{enumerate}
In words, the event $\EE_2(W)$, given that we sampled $W$ from $A_i$, states that if we sample vertices in $V_3$ with probability $p$ (independently twice, for $Z_1$ and $Z_2$), then the induced hypergraph on $W$, $Z_1$ and $Z_2$ contains a hyperedge.
We prove \Cref{prop2:balanced2 helper} using the following claims:
\begin{claim}\label{claim:ee0}
    $\alpha(P^{(3)},\HHc)\geq \sum_{\emptyset\neq W\subseteq A_i}\Pr{\EE_1(W)}\cdot \Pr{\EE_2(W)}$, where $P^{(3)}=(q_1,p,p)$ as defined in \Cref{prop2:balanced2 helper}.
\end{claim}
\begin{claim}\label{claim:ee1}
    $\Pr{A_i[q_1]\nemp}\geq 1-1/e$.
\end{claim}
\begin{claim}\label{claim:ee2}
    $\Pr{\EE_2(W)}\geq 1/8$, for every non-empty $W\subseteq A_i$.
\end{claim}
\begin{proof}[Proof of \Cref{prop2:balanced2 helper}]
    \begin{align*}
        \alpha(P^{(3)},\HHc)
         & \geq \sum_{\emptyset\neq W\subseteq A_i}\Pr{\EE_1(W)}\cdot \Pr{\EE_2(W)} & \text{by \Cref{claim:ee0}} \\
         & \geq \frac{1}{8}\sum_{\emptyset\neq W\subseteq A_i}\Pr{\EE_1(W)}        & \text{by \Cref{claim:ee2}} \\
         & = \frac{1}{8} \cdot \Pr{A_i[q_1]\nemp}                                          \\
         & \geq \frac{1}{16}\;. & \text{by \Cref{claim:ee1}}& \qedhere
    \end{align*}
\end{proof}

\Cref{claim:ee0} is a special case of a more general statement (see \Cref{claim:independence of e1e2}) that we use in the proof of \Cref{thm:refining sampling vector} and therefore is deferred to the next subsection.
We prove \Cref{claim:ee1,claim:ee2}.
\Cref{claim:ee1} follows from a standard hitting set argument:
\begin{proof}[Proof of \Cref{claim:ee1}]
    Sample each vertex of $A_i$ independently with probability $q_1$.
    Then
    \[
        \Pr{A_i[q_1]=\emptyset}=(1-q_1)^{|A_i|}\leq \exp(-q_1|A_i|)\;.
    \]
    It therefore suffices to show that $q_1|A_i|\ge 1$.
    Since $A_i$ is heavy, $|E(A_i)|\ge \frac{t}{6\log t}$.
    Also, every vertex in $A_i$ has degree $<2^i$, hence $|E(A_i)|\le |A_i|\cdot 2^i$, and so
    \[
        |A_i|\ge \frac{|E(A_i)|}{2^i}\ge \frac{t}{2^i\cdot 6\log t}\;.
    \]
    If $q_1=1$, then $q_1|A_i|\geq 1$ because $A_i$ is non-empty. Otherwise, using $q_1=\cg\cdot \frac{2^{i+1}}{t}$, we get
    \[
        q_1|A_i|\ge \frac{t}{2^i\cdot 6\log t}\cdot \cg\cdot \frac{2^{i+1}}{t}
        = \frac{100}{3}\cdot \frac{\log n}{\log t}\ge 1\;,
    \]
    where the last inequality holds since $t\le n^4$.
\end{proof}
It remains to prove \Cref{claim:ee2}, which is the only place we use the assumption that no three vertices of any diamond in $\DDr$ are contained in an $r$-clique.
The intuition behind the proof is as follows.
We look at a single element $s\in A_i$, where $s=(v_1,v_2)$, and at the set $M$ of vertices $u\in V_3$ such that $(v_1,v_2,u)$ is contained in some diamond in $\DDr$.
Using the fact that no three vertices of any diamond in $\DDr$ are contained in an $r$-clique, we conclude that many pairs of vertices in $M$ are non-edges and therefore correspond to diamonds in $\DDr$ with the edge $(v_1,v_2)$.
In other words, there are many pairs of distinct vertices in $M$ that together with $(v_1,v_2)$ form induced diamonds, and we hit such a pair with good probability.

\begin{proof}[Proof of \Cref{claim:ee2}]
    Fix an arbitrary $s=(v_1,v_2)\in A_i\subseteq V_1\times V_2$.
    We prove that $\Pr{\EE_2(\set{s})}\geq 1/8$.
    Since $\EE_2(\cdot)$ is monotone (i.e., $W\subseteq W'$ implies $\EE_2(W)\subseteq \EE_2(W')$), it follows that for every non-empty $W\subseteq A_i$ and any choice of $s\in W$,
    \[\Pr{\EE_2(W)}\geq \Pr{\EE_2(\set{s})}\geq 1/8\,,\]
    which is exactly the claim.

    Define $M\subseteq V_3$ to be the set of all vertices $u\in V_3$, such that $(v_1,v_2,u)$ is contained in some diamond in $\DD_\psi$.
    Let $\ell=2^{i-1}$, where recall that $s$ is in $[\ell,2\ell)$ diamonds in $\DD_\psi$, otherwise it would not be in $A_i$.
    Let $G_s=G[M]$, be the induced subgraph of $G$ on the vertices in $M$.
    Sample each vertex in $G_s$ independently w.p. $p=\sqrt{r\cdot q_2}=\sqrt{r / \ell}$, where the equation follows by \Cref{eq:P2}.
    We show that the sampled graph contains a pair of vertices with no edge between them, with probability at least $1/8$.

    ~\\Let $k$ be the number of vertices of $G_s$, $m$ the number of edges, where $m=\ch{k}{2}-d$, for some $d\geq 0$.
    Since every pair $(u,u')$ in $M$ such that $(v_1,v_2,u,u')\in \DD_\psi$, must satisfy $(u,u')\notin E(G_s)$, we have that $d\in[\ell,2\ell)$.
    Since every vertex in $M$ is in at least one diamond with $(v_1,v_2)$, it is incident to at least one non-edge in $G_s$, and therefore $k\leq 2d$.
    We claim that
    \begin{align}
        d\geq \frac{k^2}{4r}\;. \label{eq:turan}
    \end{align}
    If $k\leq 2r$, then \Cref{eq:turan} holds trivially.
    Assume $k>2r$.
    If $G_s$ has no $r$ clique, then this follows by Tur\'an's theorem:
    \begin{align*}
        d\geq \ch{k}{2}-\brak{1-\frac{1}{r-1}}\frac{k^2}2 = \frac{k^2}{2(r-1)}-\frac{k}{2}
        \geq \frac{k^2}{4r}
        \;.
    \end{align*}
    To use Tur\'an's theorem, we need to show that $G_s$ does not have an $r$-clique.
    Assume towards contradiction that $K$ is an $r$-clique in $G_s$, and let $u\in M$ be a vertex in $K$.
    Then, the triplet $(v_1,v_2,u)$ is contained in an $r$-clique in $G$ and therefore is not contained in any diamond in $\DD$, contradicting the definition of $M$.
    By rearranging \Cref{eq:turan}, we get $k\leq \sqrt{4d r}\leq \sqrt{8\ell r}$.

    We are now ready to prove that $\Pr{\EE_2(\set{s})}\geq 1/32$, 
    using the second moment method.
    Let $V_p\gets M[p]$ be the set of vertices sampled from $M$ with probability $p=\sqrt{r/\ell}$.
    Let $F=(f_1,\ldots, f_\ell)$ be an arbitrary subset of $\ell$ non-edges in $G_s$, and let
    \begin{align*}
        C=\set{(i,j)\in[\ell]^2\mid \abs{f_i\cap f_j}=1}\;,
    \end{align*}
    be the set of pairs of non-edges in $G_s$ that share exactly one vertex.
    We have $\abs{C}\leq 2\ell k$, as every pair of non-edges in $G_s$ can share at most one vertex.
    By \Cref{eq:turan} $k\leq \sqrt{8r\ell}$ and therefore $\abs{C}\leq 2\sqrt{8r\ell}\cdot \ell$.

    Define an indicator random variable $X_i$ for every $i\in[\ell]$, where $X_i=1$ if both endpoints of $f_i$ are sampled into $V_p$, and $X_i=0$ otherwise. Also, $X=\sum_{i=1}^\ell X_i$.
    We analyze $\Pr{X>0}$, where if $\set{X>0}$, then $\EE_2(\set{s})$ holds.
    By the second moment method we have $\Pr{X>0}\geq  \Exp{X}^2/\Exp{X^2}$.
    We get:
    \begin{align*}
        \Exp{X}=\ell\cdot p^2\;, \quad \Exp{X^2}=\ell p^2 + \abs{C}p^3 + \ell^2p^4\;.
    \end{align*}
    Therefore,
    \begin{align*}
        \Pr{X>0}
        \geq \frac{\Exp{X}^2}{\Exp{X^2}}
        =\frac{\ell^2p^4}{\ell p^2+\abs{C}p^3 + \ell^2p^4}
        \geq\frac{\ell^2p^4}{\ell p^2+2\ell\sqrt{8r\ell}p^3 + \ell^2p^4}
        \geq \frac{\ell p^2}{1+6\sqrt{r\ell}p + \ell p^2}\;.
    \end{align*}
    The penultimate inequality follows from plugging in $\abs{C}\leq 2\sqrt{8r\ell}\cdot \ell$.
    Since $p=\sqrt{r/\ell}$, we get:
    \begin{align*}
        \Pr{X>0}
        \geq \frac{\ell p^2}{1+6\sqrt{r\ell}p + \ell p^2}
        \geq \frac{r}{1+6r+r} \geq \frac{1}{8}\;.
    \end{align*}
    This completes the proof.
\end{proof}

We still need to prove \Cref{lemma:refine simple,claim:ee0}, which we do in the following subsection.

\subsection{Refinement Lemma}
\label{ssec:refine}
The main goal of this subsection is to prove \Cref{lemma:refine simple}:

\RefineSimple* 

Instead of proving \Cref{lemma:refine simple} directly, which requires \emph{refining} a $3$-dimensional sampling vector into a $4$-dimensional sampling vector, we explain how to refine a $k$-dimensional sampling vector into a $(k+1)$-dimensional sampling vector for general $k$.
The work in \cite{CEW24,CEW25} considers refining a $1$-dimensional sampling vector into a $k$-dimensional sampling vector.
Stated differently, let $\XX=\XX_1\times\XX_2\times\cdots\times \XX_k$ be a Cartesian product of $k$ sets, and let $\UU\subseteq \XX$ be a subset of $\XX$.
Then, by sampling every element in $\XX$ independently with probability $q$, we obtain a $1$-dimensional sampling vector that hits $\UU$ with probability $1-(1-q)^{|\UU|}$.
Here, $\XX$ is the set of all quadruples of vertices, and $\UU$ is the set of all colorful induced diamonds.
A $4$-dimensional sampling vector samples each element in $\XX_i$ independently with probability $p_i$, for every $i\in[4]$.
The following lemma shows that there exists a sampling vector $P=(p_1,p_2,p_3,p_4)$ such that the sampled set contains at least one element from $\UU$ with probability $\tg(1)$ and $w(P)=\To(\frac1{|\UU|})$:

\LemDisc* 
The analysis follows similar lines to those in \cite{CEW24,CEW25}, but requires a slightly more delicate argument.
\bigskip\noindent
Let $\GG$ be a $(k+1)$-partite $(k+1)$-uniform hypergraph with vertex sets $(V_1,\ldots,V_{k+1})$.
Let $\HH$ be a $k$-partite $k$-uniform hypergraph with vertex sets $(V_1,\ldots,V_{k-1}, V_k')$, where $V_k'=V_k\times V_{k+1}$.
In other words, the $k$th vertex set in $\HH$ is the Cartesian product of the $k$th and $(k+1)$th vertex sets in $\GG$.
That is,
\begin{align*}
    V(\GG)=\bigcup_{i=1}^{k+1} V_i\;, &  & \text{and } \quad V(\HH)=\brak{\bigcup_{i=1}^{k-1} V_i}\cup V_k'\;.
\end{align*}
Any hyperedge $e=(v_1,\ldots, v_{k+1})\in\GG$, where $v_i\in V_i$ for every $i\in[k+1]$, corresponds to a hyperedge $e'\in \HH$ obtained from $e$ by replacing $v_k$ and $v_{k+1}$ with the element $v_k'=(v_k,v_{k+1})\in V_k'$, i.e., $e'=(v_1,\ldots, v_{k-1},v_k')$.
In other words, the hyperedges of $\GG$ play the role of the elements of $\UU$, and the hyperedges of $\HH$ play the same role when these elements are viewed as $k$-tuples by contracting the last two coordinates into one.

\begin{definition}
    For every $s$-partite $s$-uniform hypergraph $\FF$ with vertex sets $(U_1,\ldots,U_s)$, and a sampling vector $R=(r_1,\ldots,r_s)\in[0,1]^s$, we define $\FF[R]$ to be the random induced subhypergraph obtained from $\FF$ by sampling each vertex $u\in U_i$ independently with probability $r_i$, for every $i\in[s]$.
    We also define $\alpha(R,\FF)$ to be the probability that $\FF[R]$ contains at least one hyperedge.
\end{definition}
The following theorem is the main result of this subsection and generalizes \Cref{lemma:refine simple}:

\begin{theorem}[Refining a Sampling Vector]
    \label{thm:refining sampling vector}
    Given a sampling vector $Q=(q_1,\ldots,q_k)$,
    define $\log n$ sampling vectors $P_i=(q_1,\ldots,q_{k-1},p_k^i,p_{k+1}^i)$ for $i\in[\log n]$, where
    \begin{align*}
        p_k^i=\min(1,12\log n\cdot q_k\cdot 2^i)\;, &  & \text{and } \quad p_{k+1}^i=\min(1,\frac{4}{2^i})\;.
    \end{align*}
    Then there exists $i\in[\ra]$ such that $\alpha(P_i,\GG)\geq \alpha(Q,\HH)/(2\log n)$.
    Note that $w(P_i)\leq w(Q)\cdot 48\log n$, for every $i\in[\ra]$.
\end{theorem}

We explain how to use \Cref{thm:refining sampling vector} to prove \Cref{lemma:refine simple}.
\begin{proof}[Proof of \Cref{lemma:refine simple} using \Cref{thm:refining sampling vector}]
    Recall that
    \begin{align*}
        \GG^{(3)} & =((V_1 \times V_2) \sqcup V_3 \sqcup V_3,\DD_{\psi})\;, \\
        \GG^{(4)} & =(V_1 \sqcup V_2 \sqcup V_3 \sqcup V_3,\DD_{\psi})\;,
    \end{align*}
    and that we want to refine the sampling vector $P^{(3)}=(q,p_3,p_4)$ into a sampling vector $P^{(4)}=(p_1,p_2,p_3,p_4)$ that satisfies the two properties in \Cref{lemma:refine simple}.

    \Cref{thm:refining sampling vector} allows us to refine a $k$-dimensional sampling vector into a $(k+1)$-dimensional sampling vector while maintaining the two properties.
    We define $P_i^{(4)}=(p_1^i,p_2^i,p_3,p_4)$, where
    \begin{align*}
        p_1^i=\min\!\bigl(1,12\ra\cdot q\cdot 2^i\bigr)\;, &  & p_2^i=\min\!\bigl(1,4/2^i\bigr)\;,
    \end{align*}
    for every $i\in[\ra]$, exactly as in \Cref{thm:refining sampling vector}.
    We have that $p_1^i\cdot p_2^i \le 48\ra\cdot q$, for every $i\in[\ra]$, so all these vectors satisfy the first property in \Cref{lemma:refine simple}.
    By \Cref{thm:refining sampling vector}, there exists $i\in[\ra]$ such that, for $P^{(4)}=P_i^{(4)}$,
    $\alpha(P^{(4)},\GG^{(4)})\geq \alpha(P^{(3)},\GG^{(3)})/(2\log n)$.
    Thus the second property in \Cref{lemma:refine simple} is also satisfied, which completes the proof.
\end{proof}

For the rest of this subsection, we prove \Cref{thm:refining sampling vector}.
We need the following definitions.
Let $A = \bigcup_{i=1}^{k-1} V_i$. For every subset of vertices $S\subseteq A$, we define $N_{\HH}(S)$ to be the set of all vertices $v'\in V_k'$ such that there exists a hyperedge $e'\in \HH$ that is contained in $S\cup\set{v'}$:
\begin{align*}
    N_{\HH}(S)=\set{v'\in V_k' \mid \exists e'\in E(\HH) \text{ s.t. } e'\subseteq S\cup\set{v'}}\;.
\end{align*}
For every non-empty $S\subseteq A$, we define the following events:
\begin{align*}
    \EE_1(S) & \triangleq \set{(V_1[q_1]\cup \cdots \cup V_{k-1}[q_{k-1}])=S}\;,                                       \\
    \EE_2(S) & \triangleq \set{N_\HH(S) \cap V_k'[q_k]\nemp}\;,                                                        \\
    \FF_i(S) & \triangleq \set{(V_k[p_k^i]\times V_{k+1}[p_{k+1}^i]) \cap N_\HH(S)\nemp}\;, \quad \forall i\in[\ra]\;.
\end{align*}
In words, $\EE_1(S)$ is the event that the set of sampled vertices from the first $k-1$ parts equals $S$.
The event $\EE_2(S)$ is that at least one vertex $v'\in V_k'\cap N_\HH(S)$ is sampled, and $\FF_i(S)$ is the event that, after sampling each vertex $v\in V_k$ independently with probability $p_k^i$ and each vertex $u\in V_{k+1}$ independently with probability $p_{k+1}^i$, there exists $(v_k,v_{k+1})\in N_\HH(S)$ such that both $v_k$ and $v_{k+1}$ are sampled.

We use the following two claims to prove \Cref{thm:refining sampling vector}.
\begin{claim}\label{claim2:y1}
    \begin{align*}
        \alpha(Q,\HH)   = & \sum_{\emptyset\neq S\subseteq A}\Pr{\EE_1(S)}\cdot \Pr{\EE_2(S)}\;, \\
        \alpha(P_i,\GG) = & \sum_{\emptyset\neq S\subseteq A}\Pr{\EE_1(S)}\cdot \Pr{\FF_i(S)}\;.
    \end{align*}
\end{claim}
\begin{proposition}[Preservation Step]\label{claim2:y4}
    For every non-empty $S\subseteq A$, there exists $i\in[\ra]$ such that
    $\Pr{\FF_i(S)}\geq \Pr{\EE_2(S)}\cdot (1-\frac1{e^2})$.
\end{proposition}
\begin{proof}[Proof of \Cref{thm:refining sampling vector} using \Cref{claim2:y1,claim2:y4}]
    We first show that $\sum_{i\in [\ra]}\alpha(P_i,\GG)\geq \alpha(Q,\HH)\cdot (1-\frac1{e^2})$:
    \begin{align*}
        \sum_{i\in[\ra]}\alpha(P_i,\GG)
         & =\sum_{i\in[\ra]}\ \sum_{\emptyset\neq S\subseteq A}
        \Pr{\EE_1(S)}\cdot\Pr{\FF_i(S)}
         &
         & \text{by \Cref{claim2:y1}}                                \\
         & =\sum_{\emptyset\neq S\subseteq A}
        \Pr{\EE_1(S)}\cdot\brak{\sum_{i\in[\ra]}\Pr{\FF_i(S)}}       \\
         & \ge \sum_{\emptyset\neq S\subseteq A}\Pr{\EE_1(S)}\cdot
        \brak{\Pr{\EE_2(S)}\cdot(1-\frac1{e^2})}
         &
         & \text{by \Cref{claim2:y4}}                                \\
         & = \alpha(Q,\HH)\cdot(1-\frac1{e^2})\;.                  &
         & \text{by \Cref{claim2:y1}}
    \end{align*}

    To complete the proof, we show that there exists $i\in[\ra]$ such that $\alpha(P_i,\GG)\geq \alpha(Q,\HH)/(2\log n)$, using a simple averaging argument.
    If the sum of $\ra$ non-negative numbers is at least $\alpha(Q,\HH)\cdot (1-\frac1{e^2})$, then at least one of them is at least $\alpha(Q,\HH)\cdot (1-\frac1{e^2})/\ra\geq \alpha(Q,\HH)/(2\log n)$. This completes the proof of the theorem.
\end{proof}

To prove \Cref{claim2:y1}, we prove the following more general statement.
\begin{proposition}[Decomposition Step]\label{claim:independence of e1e2}
    Let $\GG$ be a $k$-partite $k$-uniform hypergraph with vertex sets $(V_1,\ldots,V_k)$, where $E\subseteq V_1\times V_2\times \cdots \times V_k$.
    Let $P=(p_1,\ldots,p_k)$ be a sampling vector, where each vertex $v\in V_i$ is sampled independently with probability $p_i$, for every $i\in[k]$.
    Let $j\in\set{2,\ldots,k}$, and let $A=\bigcup_{i=1}^{j-1} V_i$.
    For every $S\subseteq A$, define two events and a set:
    \begin{align*}
        \EE_1(S) & \triangleq \set{(V_1[p_1]\cup \cdots \cup V_{j-1}[p_{j-1}]) =S}\;,                                                                         \\
        \EE_2(S) & \triangleq \set{(V_j[p_j]\times\cdots \times V_k[p_k])\cap N_\GG(S) \nemp}\;,                                                              \\
        N_\GG(S) & \triangleq \set{(v_j,\ldots,v_k)\in V_j\times \cdots \times V_k \mid \exists e\in E \text{ s.t. } e\subseteq S\cup\set{v_j,\ldots,v_k}}\;.
    \end{align*}
    Then, $\alpha(P,\GG)=\sum_{\emptyset\neq S\subseteq A}\Pr{\EE_1(S)}\cdot \Pr{\EE_2(S)}$.
\end{proposition}

\begin{proof}[Proof of \Cref{claim:independence of e1e2}]
    \newcommand{\AL}{L} 
    \newcommand{\BL}{R} 
    We use $\AL \triangleq \bigcup_{i=1}^{j-1}V_i$ and $\BL \triangleq \bigcup_{i=j}^{k}V_i$.
    Define two random variables $X$ and $Y$:
    \[
        X \triangleq V_1[p_1]\cup\cdots\cup V_{j-1}[p_{j-1}]
        \qquad\text{and}\qquad
        Y \triangleq V_j[p_j]\cup\cdots\cup V_k[p_k].
    \]
    Let $\chi$ be an indicator random variable, defined on subsets of vertices $V'\subseteq V(\GG)$, whose
    value is $\chi(V')=1$ if $\GG[V']$ contains an edge and is $0$ otherwise.
    Then, by definition,
    \[
        \alpha(P,\GG)=\Pr{\chi(X\cup Y)=1}.
    \]

    For each $S\subseteq \AL$, note that $\EE_1(S)=\set{X=S}$ and that the events
    $\set{\EE_1(S)}_{S\subseteq \AL}$ partition the sample space of $X$. Thus, by the law of total probability,
    \begin{align*}
        \alpha(P,\GG)
        = \Pr{\chi(X\cup Y)=1}
        = \sum_{S\subseteq \AL}\Pr{\EE_1(S)}\cdot \Pr{\chi(X\cup Y)=1\mid X=S}.
    \end{align*}

    The random set $Y$ is determined only by sampling vertices in $\BL$, while $X$ is determined
    only by sampling vertices in $\AL$, so $Y$ is independent of $X$. Therefore, for every $S\subseteq \AL$,
    \[
        \Pr{\chi(X\cup Y)=1\mid X=S}=\Pr{\chi(S\cup Y)=1}.
    \]
    Finally, by the definition of $N_\GG(S)$,
    \[
        \set{\chi(S\cup Y)=1}\iff
        \set{(V_j[p_j]\times\cdots\times V_k[p_k])\cap N_\GG(S)\nemp}=\EE_2(S),
    \]
    and thus $\Pr{\chi(S\cup Y)=1}=\Pr{\EE_2(S)}$. Substituting this gives
    \[
        \alpha(P,\GG)=\sum_{S\subseteq \AL}\Pr{\EE_1(S)}\cdot \Pr{\EE_2(S)}\;.
    \]
    For $S=\emptyset$, the probability $\Pr{\EE_2(S)}=0$, and therefore
    \[
        \sum_{S\subseteq \AL}\Pr{\EE_1(S)}\cdot \Pr{\EE_2(S)}
        =\sum_{\emptyset\neq S\subseteq \AL}\Pr{\EE_1(S)}\cdot \Pr{\EE_2(S)}
        \;,
    \]
    which completes the proof.
\end{proof}

We use \Cref{claim:independence of e1e2} to prove \Cref{claim2:y1}.
\begin{proof}[Proof of \Cref{claim2:y1}]
    The first equality is a special case of \Cref{claim:independence of e1e2} applied to $\HH$ with $j=k$.
    The second equality is the same decomposition applied to the $(k+1)$-partite hypergraph $\GG$ with $j=k$, so the first $k-1$ parts are sampled into $S$ and the last two parts are represented by $\FF_i(S)$.
\end{proof}
\Cref{claim:ee0} from the previous section also follows as a special case of \Cref{claim:independence of e1e2} with $k=3$ and $j=2$.
It remains to prove \Cref{claim2:y4}; for this we need the following lemma:
\begin{restatable}{lemma}{LemmaBipartite}\label{claim7:bipartite}
    Let $F=((A,B),E)$ be a bipartite graph with $m$ edges, where every vertex in $A$ has degree in $[d,2d]$. For any sampling vector $P=(p_1,p_2)$, define
    \begin{align*}
        \Lambda_1= \boldone_{\set{p_1<1}}\cdot \exp(-m\cdot p_1/(2d))\;, &  & \text{and } \quad
        \Lambda_2= \boldone_{\set{p_2<1}}\cdot \exp(-d\cdot p_2)\;.
    \end{align*}
    Then, $\Pr{F[P]\text{ contains an edge}}\geq (1-\Lambda_1)\cdot(1-\Lambda_2)$.
\end{restatable}

Recall that $F[P]$ is a random induced subgraph obtained from $F$ by sampling each vertex in $A$ independently with probability $p_1$ and each vertex in $B$ independently with probability $p_2$.
We prove \Cref{claim2:y4} using \Cref{claim7:bipartite}.
\begin{proof}[Proof of \Cref{claim2:y4} using \Cref{claim7:bipartite}]
    \newcommand{\mj}{m_S'}
    Fix a non-empty subset $S\subseteq A$, and define a bipartite graph $H_S$
    with vertex set $V(H_S)=(V_k\sqcup V_{k+1})$ and edge set $E(H_S)=N_\HH(S)\subseteq V_k^\prime =V_k\times V_{k+1}$, where $m_S=\abs{E(H_S)}$.
    Then,
    \[
        \Pr{\EE_2(S)} = 1-(1-q_k)^{m_S}\;.
    \]
    We show that there exists $i\in[\log n]$ such that
    \begin{align*}
        \Pr{\FF_i(S)}
        \geq \Pr{\EE_2(S)}\cdot (1-\frac1{e^2})\;.
    \end{align*}
    Recall that $\FF_i(S)$ is the event that $H_S[T_i]$ contains an edge, where $T_i=(p_k^i,p_{k+1}^i)$ and
    \begin{align*}
        p_k^i=\min(1,12\log n\cdot q_k\cdot 2^i)\;, &  & \text{and } \quad p_{k+1}^i=\min(1,\frac{4}{2^i})\;,
    \end{align*}
    as defined in \Cref{thm:refining sampling vector}, for every $i\in[\log n]$.
    Let $\alpha(T_i,H_S)$ denote the probability that $H_S[T_i]$ contains an edge.
    It remains to show that there exists $i\in[\log n]$ such that
    \begin{align*}
        \alpha(T_i,H_S)
        \geq \Pr{\EE_2(S)}\cdot (1-\frac1{e^2}) \;.
    \end{align*}

    Partition the vertices of $V_k$ into degree classes over $H_S$;
    for $j\in[\log n]$, define
    \begin{align*}
        U_j=\set{v\in V_k\mid \deg_{H_S}(v)\in[2^{j-1},2^{j})}\;.
    \end{align*}
    Let $E(U_j)\subseteq E(H_S)$ denote the set of edges in $H_S$ with one endpoint in $U_j$.
    We say that $U_j$ is \emph{heavy} if $\abs{E(U_j)}\geq m_S/(3\ra)$. At least one heavy set must exist; otherwise, the graph would have fewer than $m_S$ edges.
    Let $U_j$ be a heavy set, and let $H_S'=H_S[U_j\cup V_{k+1}]$ be the induced bipartite subgraph of $H_S$ on vertex sets $U_j$ and $V_{k+1}$.
    We claim that
    \begin{align}
        \alpha(T_j,H_S')\geq \Pr{\EE_2(S)}\cdot (1-\frac1{e^2})\;.
    \end{align}
    Since $H_S'$ is a subgraph of $H_S$, this claim implies the desired bound for $\alpha(T_j,H_S)$.
    To prove this, we apply \Cref{claim7:bipartite} to the bipartite graph $H_S'$, which has $m_S'=\abs{E(U_j)}$ edges. Every vertex in $U_j$ has degree in $[2^{j-1},2^j)$, so the minimum degree is $d=2^{j-1}$. We get
    \begin{align*}
        \alpha(T_j,H_S') \geq (1-L_1)\cdot (1-L_2)\;,
    \end{align*}
    where
    \begin{align*}
        L_1 & = \boldone_{\set{p_k^j<1}}\cdot \exp(-\mj\cdot p_k^j/(2d))\;, &  &
        L_2 = \boldone_{\set{p_{k+1}^j<1}}\exp(-d\cdot p_{k+1}^j)\;.
    \end{align*}
    It is enough to prove the following two inequalities:
    \begin{enumerate}
        \item $1-L_1 \ge \Pr{\EE_2(S)}$.
        \item $1-L_2\ge 1-e^{-2}$,
    \end{enumerate}
    These inequalities imply the claim.
    \paragraph{Proof of (1):}
    If $p_k^j=1$, then $L_1=0$ and $1-L_1=1\ge \Pr{\EE_2(S)}$.
    Otherwise, $p_k^j=12\ra q_k 2^j<1$, and in particular $q_k\leq 1/2$.
    Then,
    \begin{align*}
        L_1 = \exp(-\frac{m_S'\cdot p_k^j}{2d})
        \leq \exp(-\frac{m_S}{3\ra}\cdot 12\ra q_k\cdot 2^j/(2\cdot 2^{j}))
        = \exp(-2m_S\cdot q_k)
        \overset{(\star)}{\leq} (1-q_k)^{m_S}= 1-\Pr{\EE_2(S)}
        \;.
    \end{align*}
    The $(\star)$ inequality follows by taking the $m_S$-th root of both sides and using $(1-x)\geq e^{-2x}$ for $x\leq 1/2$, with $x=q_k$.

    \paragraph{Proof of (2):}
    If $p_{k+1}^j=1$, then $L_2=0$ and $1-L_2=1\ge 1-e^{-2}$.
    Otherwise, $p_{k+1}^j=4/2^j=2/d$.
    Thus, $d p_{k+1}^j=2$, and hence
    $1-L_2\ge 1-e^{-2}$.

    This completes the proof of \Cref{claim2:y4}.
\end{proof}

It remains to prove \Cref{claim7:bipartite}.
\begin{proof}[Proof of \Cref{claim7:bipartite}]
    \newcommand{\ca}{1-\exp(-d\cdot p_2)}
    \newcommand{\cbb}{1-\Lambda_2}
    \newcommand{\caa}{1-\Lambda_1}
    For every non-empty subset $W\subseteq A$, we define two events.
    We write $N(W)$ for the set of neighbors of $W$ in $B$.
    \begin{align*}
        \EE_1(W) = \set{A[p_1]=W}\;, &  &
        \EE_2(W) = \set{N(W)\cap B[p_2]\nemp}\;.
    \end{align*}
    By \Cref{claim:independence of e1e2}, it holds that
    \[
        \Pr{F[P]\text{ contains an edge}}
        =   \sum_{\emptyset\neq W\subseteq A}\Pr{\EE_1(W)}\cdot \Pr{\EE_2(W)}\;.
    \]
    We use the following two inequalities:
    \begin{enumerate}
        \item $\Pr{\EE_2(W)}\geq 1-\Lambda_2$, for every non-empty $W\subseteq A$.
        \item $\Pr{A[p_1]\nemp}\geq 1-\Lambda_1$.
    \end{enumerate}
    Therefore,
    \begin{align*}
        \Pr{F[P]\text{ contains an edge}}
         & =   \sum_{\emptyset\neq W\subseteq A}\Pr{\EE_1(W)}\cdot \Pr{\EE_2(W)}          \\
         & \overset{(1)}{\geq} (\cbb)\cdot \sum_{\emptyset\neq W\subseteq A}\Pr{\EE_1(W)}
        =   (\cbb)\cdot \Pr{A[p_1]\nemp}
        \overset{(2)}{\geq}   (\cbb)\cdot \brak{\caa}\;.
    \end{align*}
    It remains to prove $(1)$ and $(2)$.

    \paragraph{Proof of $(1)$:}
    \begin{align*}
        \Pr{\EE_2(W)}
        \geq 1-(1-p_2)^{\abs{N(W)}}
        \geq 1-\exp(-p_2\cdot \abs{N(W)})
        \geq \ca\;.
    \end{align*}
    The last inequality uses the fact that every vertex in $A$ has degree at least $d$, so $\abs{N(W)}\geq d$.
    Note that if $p_2=1$, then $1-(1-p_2)^{\abs{N(W)}}=1$, so $\Pr{\EE_2(W)}=1=1-\Lambda_2$.
    \paragraph{Proof of $(2)$:}
    \begin{align*}
        \Pr{A[p_1]\nemp}
        = 1-(1-p_1)^{\abs{A}}
        \geq 1-\exp(-p_1\cdot \abs{A})
        \geq 1-\exp(-mp_1/(2d))\;.
    \end{align*}
    The last inequality uses the fact that every vertex in $A$ has degree at most $2d$, and there are $m$ edges in total, so $\abs{A}\geq m/(2d)$.
    As before, if $p_1=1$, then $\Pr{A[p_1]\nemp}=1=1-\Lambda_1$.
    This proves $(1)$ and $(2)$, and therefore completes the proof.
\end{proof}

\section{Lower Bounds}
\label{sec:lower-bounds}
In this section, we present lower bounds for the sensitive detection of induced diamonds. We begin with a non-conditional lower bound showing that $\Omega(n^{2})$ time is needed in dense graphs (i.e., essentially reading the input) for every $t \leq n^2/5$. This lower bound holds in a model where access to the input graph is provided via:
\begin{description}
    \item[Adjacency queries:] Given a pair of vertices $u, v \in V(G)$, determine whether the edge $uv$ exists in $G$.
    \item[Degree queries:] Given a vertex $v \in V(G)$, return its degree $\deg(v)$.
    \item[Neighbor queries:] Given a vertex $v \in V(G)$ and an index $i \in [\deg(v)]$, return the $i$-th neighbor of $v$. 
\end{description}

We then present two conditional lower bounds for the problem: one for general algorithms and one that holds only against combinatorial algorithms. Our lower bounds are universal with respect to $t$, holding for every $t \in [0,n^2/5]$. We remark that throughout the section, the time bounds for $t = 0$ are defined as for $t = 1$ to avoid zero-division. 

We begin with the proof of the non-conditional lower bound.
\begin{theorem}
	\label{lem:quadratic-lb}
	For every $t \leq n^2/5$, every randomized algorithm that, given query access to an $n$-vertex graph $G$ with at least $t$ induced diamonds, outputs an induced diamond in $G$ with probability at least $2/3$, requires $\Omega(n^2)$ queries. 
\end{theorem}
\begin{proof}
	Fix $t \leq n^2/5$, and assume $n$ is sufficiently large. Consider an input distribution defined by starting with an $n/2 \times n/2$ biclique on vertex sets $U \sqcup V$. We select two pairs of distinct vertices, $\{u_i, u_j\} \in \binom{U}{2}$ and $\{v_i, v_j\} \in \binom{V}{2}$, uniformly at random. 
	We then modify the graph by removing the crossing edges $u_i v_j$ and $u_j v_i$, and adding the edges $u_i u_j$ and $v_i v_j$. 

Let us count the number of induced diamonds in any supported graph. An induced diamond requires a diagonal edge and two non-adjacent common neighbors. The only edges fully within $U$ or $V$ are $u_i u_j$ and $v_i v_j$. Crossing edges between $U$ and $V$ cannot serve as diagonals because any crossing edge has at most one common neighbor. Therefore, any induced diamond must use either $u_i u_j$ or $v_i v_j$ as its diagonal. 
The endpoints of $u_i u_j$ share the common neighborhood $V \setminus \{v_i, v_j\}$. Because this is an independent set of size $n/2 - 2$, any pair of these vertices forms an induced diamond with $u_i u_j$. This yields $\binom{n/2 - 2}{2}$ diamonds. By symmetry, $v_i v_j$ also forms $\binom{n/2 - 2}{2}$ diamonds with pairs from $U \setminus \{u_i, u_j\}$. Thus, the total number of induced diamonds is exactly $2 \binom{n/2 - 2}{2}$. For sufficiently large $n$, we have $2 \binom{n/2 - 2}{2} \geq n^2/5 \geq t$. 

Furthermore, the degrees in the modified graph remain uniformly $n/2$. Hence, we can assume without loss of generality that the algorithm does not make degree queries (as their answers can be hardcoded).

Now, consider a deterministic algorithm that makes $k$ adjacency or neighbor queries and (possibly depending on their answers) outputs a quadruple of vertices $(a,b,c,d)$. We analyze the probability that $(a,b,c,d)$ induces a diamond in $G$. The probability that one of the $k$ queries reveals any of the randomly altered pairs ($(u_i,u_j), (u_i,v_j), (v_i,v_j), (v_i,u_j)$) is $O(k/n^2)$. If a query hits these pairs, the algorithm may indeed correctly output a diamond. If none of the queried edges or neighbors hit the altered pairs, the randomized choice remains uniformly distributed across the remaining unqueried graph. In this case, the probability that the output $(a,b,c,d)$ successfully guesses a quadruple containing a diagonal edge is at most $O(1/(n^2-k))$. 

\begin{equation*}
	\begin{split}
		&\Pr{(a,b,c,d) \text{ is a diamond}} \leq \\
		&\Pr{\text{queries hit}} + \Pr{(a,b,c,d) \text{ is a diamond} \mid \text{queries not hit}} = \\ 
		& O\left( \frac{k}{n^2} + \frac{1}{n^2 - k} \right).
	\end{split}
\end{equation*}

Observe that for the success probability to be at least a constant (e.g., $2/3$), we require $k = \Omega(n^2)$ queries. By Yao's Minimax Principle, this lower bound extends to randomized algorithms, yielding the $\Omega(n^{2})$ time lower bound.
\end{proof}

We now continue to the conditional lower bounds. Our lower bound for general algorithms is based upon the following Unbalanced Triangle Detection (UTD) Hypothesis:

\begin{restatable}[Unbalanced Triangle Detection Hypothesis \cite{CEW24}]{hypothesis}{ConjUTD}
    \label{conj:utd}
    Let $G$ be a tripartite graph with vertex set $A \sqcup B \sqcup C$, where $|A| = |C| = n$ and $|B| = n^b$ for $b \leq 1$. Every randomized algorithm for triangle detection in $G$, in the $\mathsf{word\text{-}RAM}$ model with word-size $O(\log n)$, requires $\MM{n,n^b,n}/n^{o(1)}$ time.
\end{restatable}
This hypothesis was shown to hold true assuming that the current $k$-clique detection algorithms are optimal~\cite{CEW24}.
We now prove the following lower bound:

\SensLB*

\begin{proof}
    Under the UTD hypothesis,~\cite[Section 6]{CEW24} showed the following: for every $p \leq n$, the problem of distinguishing between a triangle-free graph and a graph with at least $p$ triangles (where $n$ is the number of vertices) requires $\MM{n,n/p,n}/n^{o(1)}$ time.

    Their reduction works as follows. Given a UTD instance $G$ with $|A|=|C|=n$ and $|B| = n/p$, it produces in $O(n^2)$ time a tripartite graph $G'$ by creating $p$ independent duplicates of $B$. Then, $G'$ has the following guarantees:
    \begin{itemize}
        \item (NO-case) If $G$ is triangle-free, then $G'$ is also triangle-free.
        \item (YES-case) If $G$ contains a triangle, then $G'$ contains at least $p$ triangles.
    \end{itemize}
    Note that any triangle $(a,b,c)$ in $G$ where $a\in A, b\in B, c\in C$ corresponds to $\binom{p}{2}$ induced diamonds in $G'$ formed by $(a,b_i,c,b_j)$ where $b_i,b_j$ are distinct copies of $b$.
    The only missing edge is $(b_i,b_j)$, since $G'$ is tripartite.
    On the other hand, in the NO-case, $G'$ is clearly diamond-free because a diamond contains a triangle.

    Now, suppose for contradiction that for some value $0 \leq t \leq n^2/3$ and some $\eps > 0$, there is an algorithm $\mathcal{A}$ that, given an $n$-vertex graph with at least $t$ diamonds, finds a diamond in $T(n,t) = \MM{n,n,n/\sqrt{t}}/n^{\eps}$ time.
    We use $\mathcal{A}$ to solve the unbalanced triangle detection problem as follows. We apply the above reduction with $p = \sqrt{2t}+1 \leq  \sqrt{2 n^2/3} + 1 \leq n$, and then run $\mathcal{A}$ on $G'$ for $T(n,t)$ steps.
    If a diamond is found, we report that $G$ contains a triangle; otherwise, we report that $G$ is triangle-free.
    Therefore:
    \begin{itemize}
        \item (NO-case) If $G$ was triangle-free, then $G'$ is diamond-free, thus our algorithm does not find a diamond and correctly reports that $G$ is triangle-free.
        \item (YES-case) If $G$ contains a triangle, then $G'$ contains at least $\binom{p}{2} \geq t$ induced diamonds. Hence, $\mathcal{A}$ finds a diamond and the algorithm correctly reports the existence of a triangle in $G$.
    \end{itemize}
    We conclude that $\mathcal{A}$ solves the UTD problem in
    \begin{align*}
        \MM{n,n/\sqrt{t},n}/n^{\eps} = \MM{n,n/p,n}/n^{\eps}\;,
    \end{align*}
    which is a contradiction to the UTD hypothesis.
       Finally, note that since $G'$ is tripartite, $G'$ contains no clique of size $4$ (or larger). This implies that no diamond in $G'$ contains three vertices that participate in a $K_4$.
\end{proof}

This shows, in particular, that our algorithm \Cref{thm:r-light diamonds} 
which runs in time $\tf(\MM{n,n,n\sqrt{\rmax/t}})$ is conditionally tight for constant $\rmax$; it cannot be improved to $\MM{n,n,n/t^{1/2+\eps_1}}/n^{\eps_2}$ for any $\eps_1,\eps_2 > 0$ even for graphs with no $4$-cliques.

~\\We now turn to lower bounds that hold against combinatorial algorithms. We base our results upon the Combinatorial Triangle Detection Conjecture, which is a popular conjecture known to be subcubically equivalent to the Combinatorial Boolean Matrix Multiplication Conjecture~\cite{williams2010subcubic}.
\begin{restatable}[Combinatorial Triangle Detection (BMM Conjecture)]{conjecture}{Conjbmm}
    \label{conj:bmm}
    There is no combinatorial triangle detection algorithm running in $O(n^{3-\eps})$ time, for any $\eps > 0$.
\end{restatable}
Based on this conjecture, we prove a lower bound on the running time of combinatorial witness-sensitive diamond detection. In the next section we provide a matching upper bound. 
\CombLB*
The proof idea is identical to that in the non-combinatorial setting: we wish to reduce an Unbalanced Triangle Detection instance to a graph where there are (roughly) $\sqrt{t}$ triangles sharing a common edge in the YES-case, and $0$ triangles otherwise. 
We show that the combinatorial analog of the UTD hypothesis is, in fact, implied by the BMM Conjecture (equivalently, Combinatorial Triangle Detection Conjecture).

\begin{lemma}[Combinatorial Unbalanced Triangle Detection]
    \label{lem:comb-utd}
    Let $G$ be a tripartite graph with vertex set $A \sqcup B \sqcup C$, where $|A| = |C| = n$ and $|B| = n^b$ for $b \leq 1$. Every combinatorial randomized algorithm for triangle detection in $G$, in the $\mathsf{word\text{-}RAM}$ model with word-size $O(\log n)$, requires $n^{2+b-o(1)}$ time, unless the BMM Conjecture is false.
\end{lemma}

\begin{proof}[Proof of \cref{lem:comb-utd}]
    Consider a triangle detection instance consisting of three parts $A, B, C$, each of size $n/3$. Suppose there exists a combinatorial algorithm for the unbalanced setting ($n, n^b, n$) running in time $O(n^{2+b-\eps})$ for some $\eps > 0$.
    We can partition $B$ into roughly $s=n^{1-b}$ sets $(B_1,\ldots, B_s)$ of size $n^b$ each. We then solve the instances $(A, B_i, C)$ for every $i \in [s]$ sequentially. 
    Given a combinatorial algorithm $\AC$ that solves the UTD problem on instances of size $(n,n^b,n)$ in $O(n^{2+b-\eps})$ time, 
    we solve all instances in $s\cdot n^{2+b-\eps} = n^{3-\eps}$ time, which refutes \cref{conj:bmm}.
\end{proof}

We are now ready to prove the lower bound against combinatorial algorithms.

\begin{proof}[Proof of \cref{thm:comb-lb}]
    We use the exact same construction and reduction as in the proof of \cref{thm:sens-lb}.
    Now the expression $\MM{a,b,c}$ is replaced with $a\cdot b \cdot c$ since we are only considering combinatorial algorithms.
\end{proof}

\section{Additional Results}
\label{sec:additional-results}
In this section, we provide additional results about combinatorial algorithms, induced $4$-cycles, and $4$-SUM. We begin with our results about combinatorial algorithms for witness-sensitive diamond detection. 

\subsection{Combinatorial Diamond Detection}
\label{ssec:comb}
In this subsection, we prove the following theorem:
\ThmComb*
The high-level approach is similar to the one for non-combinatorial algorithms:
We employ two algorithms analogously to \cref{thm:r-heavy diamonds} and \cref{thm:r-light diamonds}. However, here the win-win approach uses a different parameter $x$, defined as the number of vertices participating in a diamond in $G$ (rather than $\rmax$, the size of the largest clique containing three vertices of a diamond). The resulting algorithms are simpler than their non-combinatorial analogs.

Specifically, for large $x$, we employ the following algorithm analogous to \cref{thm:r-heavy diamonds}, based on the structural neighborhood analysis approach:
\begin{restatable}[\algVertex]{theorem}{ThmAlgVertex}
    \label{thm:vertex in diamond}
    The algorithm \algVertex finds an induced diamond in time $\To((n+m)\cdot (\frac{n}{x}+1))$ \whp, where $x$ is the number of vertices that participate in a diamond in $G$. 
\end{restatable}
Both the running time and correctness of \algVertex are guaranteed \whp.
The worst case running time of this algorithm is $O(n(n+m))$, and in that case the output is always correct. 
Moreover, the algorithm never reports that $G$ contains a diamond when it is diamond-free.
This yields an $\To(n^3/x)$ algorithm. We remark that throughout the section, the time bounds for $x = 0$ are defined as for $x = 1$ to avoid division by zero. For the case of small $x$, we prove the following:

\begin{restatable}[\algXlight]{theorem}{Thmxlight}\label{thm:x-light}
    There is an algorithm that, given a graph $G$ with $t$ diamonds and $x$ vertices that participate in a diamond, finds an induced diamond in time $\To(\MM{n,n,n \cdot (\frac{x}{t}+1)})$ \whp.
\end{restatable}
By employing the straightforward combinatorial algorithm for matrix multiplication, we get a combinatorial algorithm that detects a diamond in time $\To(n^3 \cdot \frac{x}{t})$.
We are now ready to complete the main proof of this subsection.
\begin{proof}[Proof of \cref{thm:comb}]
    We execute both algorithms from \cref{thm:vertex in diamond} and \cref{thm:x-light} (using straightforward combinatorial matrix multiplication) in a round-robin fashion, and halt whenever one of the executions outputs a diamond.
    With high probability, the running time is the minimum of their running times:
    \[
        \To\left(\min\left( \frac{n^3}{x}, n^3 \cdot \frac{x}{t} \right) \right).
    \]
    This is always bounded by $\To(n^3/\sqrt{t})$, achieved when $x = \sqrt{t}$.
\end{proof}

Let us begin with the proof of \Cref{thm:vertex in diamond}. It is a simple application of a hitting-set argument in combination with the following theorem:
\ThmAlgVertexOne*

\begin{proof}[Proof of \Cref{thm:vertex in diamond}]
    Sort the vertices in a random order $\pi$, and let $v_i$ be the $i$-th vertex in this order.
    \begin{enumerate}
        \item Check if $v_i$ participates in a diamond using \algVertexOne from \cref{thm:is v in D}. If a diamond is found, return it. Otherwise continue to the next vertex in the order.
    \end{enumerate}
    If no diamond is found after all vertices are checked, report that $G$ is diamond-free.

    \paragraph{Analysis and Correctness.} 
    If the graph is diamond-free, the algorithm always reports that $G$ is diamond-free and takes $O(n(n+m))$ time.
    Suppose that $x > 0$.
    Assume that the permutation $\pi$ is set in an online fashion, meaning that the $i$-th vertex $v_i$ is chosen uniformly at random from the set of vertices that have not been chosen in the first $i-1$ steps.
    Therefore, the probability that the first vertex $v_1$ is part of a diamond is $x/n$. For the $i$-th vertex, assuming no previous vertex was part of a diamond, the probability that $v_i$ is part of a diamond is at least $x/(n-i+1)\geq x/n$.
    Therefore, the probability that no vertex from $D$ is sampled within $i \leq \ell$ rounds is:
    \[
        \left(1- \frac{x}{n}\right)^i \leq \exp\left(-i \cdot \frac{x}{n}\right).
    \]
    For $i = \frac{n}{x} \cdot c\log n$, this probability is at most $1/n^c$, and thus the algorithm finds a diamond and terminates within $i$ rounds \whp. In this case, the running time is $O((n+m) \cdot i) = \To((n+m) \cdot \frac{n}{x})$.
\end{proof}

We now turn to prove \cref{thm:x-light}.
We proceed in two steps.
First, we aim to show that the running time of \algSensitive is $\To(\MM{n,n,n\cdot x/t})$ \whp. Then, we explain how to turn this into a combinatorial algorithm with running time $\To(n^3 \cdot x/t)$.
Ignoring combinatorial algorithms for a moment to show that the running time of \algSensitive is $\To(\MM{n,n,n\cdot x/t})$ \whp, 
For the former, it suffices to show that there exists a sampling vector $P=(p_1,p_2,p_3,p_4)$ such that $\alpha(P,G)=\tg(1)$ i.e., after sampling according to $P$, the probability of getting a colorful diamond is $\tg(1)$, and $\wc(P) = \To(x/t)$, where recall that $\wc(P) = \frac{w(P)}{\min(P)}$ and $w(P) = p_1 \cdot p_2 \cdot p_3 \cdot p_4$.
Recall that $\PP$ denotes the set of sampling vectors $P$ for which $\alpha(P,G)=\tg(1)$, and $\wc(\PP) = \min\{\wc(P) \mid P \in \PP\}$.
Previously we had the following non-combinatorial proposition:
\propClaimRuntime* 

To turn this framework into a combinatorial algorithm, we implement each call to \algCD using straightforward combinatorial matrix multiplication, instead of the fast matrix multiplication used in \Cref{claim:runtime w3}.
We obtain:
\begin{proposition}[Combinatorial \Cref{claim:runtime}]\label{claim:runtime comb}
    The running time of the \emph{combinatorial} algorithm \algSensitive is $\To(n^3\cdot \wc(\PP))$ \whp. The output is correct \whp.
\end{proposition}
To complete the proof of \cref{thm:x-light}, it remains to show that $\wc(\PP)\leq \tf(x/t)$, which together with \Cref{claim:runtime comb} implies the desired running time of $\To(n^3 \cdot \frac{x}{t})$.
\begin{claim}\label{claim:xt}
    $\wc(\PP)\leq \tf(x/t)$.
\end{claim}
To prove \Cref{claim:xt}, we need the following claim:
\begin{claim}\label{claim:min xt}
    For every sampling vector $P$, we have $\alpha(P,G) \leq x \cdot \min(P)$.
\end{claim}

\begin{proof}[Proof of \Cref{claim:xt} Using \Cref{claim:min xt}]
    By \Cref{lemma:discovery}, there exists a sampling vector $P$ such that $\alpha(P,G) = \tg(1)$, and thus $P \in \PP$, with $w(P) = \To(1/t)$.
    By \Cref{claim:min xt}, we have $\alpha(P,G) \leq x \cdot \min(P)$, and thus $\min(P) \geq \alpha(P,G)/x = \tg(1)/x$.
    This means that
    \begin{align*}
        \wc(P) = \frac{w(P)}{\min(P)} = \frac{\To(1/t)}{\min(P)} \leq \frac{\To(1/t)}{\tf(1/x)} = \tf(x/t)\;,
    \end{align*}
    as desired.
    Clearly $\wc(\PP) \leq \wc(P)$, and thus $\wc(\PP) \leq \tf(x/t)$.
\end{proof}
\begin{proof}[Proof of \Cref{claim:min xt}]
    Fix some index $i\in[4]$.
    Let $D \subseteq V_i$ be the subset of $i$-colored vertices that participate in at least one colorful induced diamond.
    Clearly, $|D|\leq x$.
    If $H \gets G[P]$ contains a colorful induced diamond, then at least one vertex of $D$ must be sampled into $H$.
    Hence,
    \[
        \{H \text{ has diamond}\} \subseteq \left\{\exists v \in D \text{ s.t. } v \in V(H)\right\}.
    \]
    By a union bound,
    \[
        \alpha(P,G)
        \leq \Pr{ \bigcup\limits_{v \in D} \set{v\text{ is sampled by }P}}
        \leq |D|\cdot p_i \leq x\cdot p_i\;.
    \]
    Since this holds for every $i\in[4]$, it must hold for the index minimizing $p_i$, so $\alpha(P,G) \leq x \cdot \min(P)$.
\end{proof}

\subsection{Other Four-Vertex Patterns}
\label{ssec:other-patterns}
The approach above extends to other $4$-vertex patterns.
In~\cite{kloks2000finding}, formulas are given that relate the counts of different $4$-vertex patterns, and in~\cite{williams2014finding} these are combined with subsampling to obtain one-sided error detection algorithms for every $4$-vertex pattern other than a clique or an independent set.
The same adaptation to the colorful setting applies to these patterns as well.

\newcommand{\Vb}{\ch{V}{2}}
\newcommand{\Pb}{P_\phi}
\newcommand{\Eb}{E_\phi}
\newcommand{\Fb}{F_\phi}
\paragraph{Colorful induced $C_4$ detection.}
We describe the adaptation for induced $C_4$, which we use in \Cref{ssec:induced-c4}.
Let $\Vb$ denote the set of all unordered pairs of vertices in $G$, and let $\Pb\subseteq \Vb$ be the set of pairs $(u,v)$ with $\phi(u) \neq \phi(v)$.
We partition $\Pb = \Eb \sqcup \Fb$, where $\Eb$ is the set of colorful edges, and $\Fb$ is the set of colorful non-edges.
For a pair $(u,v) \in \Pb$ with $\phi(u) = a$ and $\phi(v) = b$, let $c,d \in [4] \setminus \{a,b\}$ be the remaining colors.
Define $k_c(u,v)$ as the number of common neighbors of $u$ and $v$ with color $c$, and let $y_{uv} \triangleq k_{c}(u,v)\cdot k_{d}(u,v)$.
The counting formula for colorful induced $C_4$s is:
\begin{align*}
    Z_\phi^{C_4}(G) & =
    \sum_{(u,v)\in \Fb} y_{uv}
    -\sum_{(u,v)\in \Eb} y_{uv}\;.
\end{align*}
Computing $Z_\phi^{C_4}(G)$ has the same running time as computing $Z_\phi(G)$, yielding:
\begin{lemma}\label{lem:colorful-det c4}
    There exists an algorithm $\textsc{DetectColorfulInducedC4}$ that determines whether a $4$-colored graph $G$ on $n$ vertices contains a colorful induced $4$-cycle in time $\TO{\MM{n_1,n_2,n_3}}$
    where $n_i$ is the number of vertices of color $i$ in $G$, and $n_1\geq n_2 \geq n_3 \geq n_4$. The algorithm succeeds \whp.
\end{lemma}

\paragraph{Approximate counting.}
\newcommand{\ccount}{{\normalfont\textsf{Count}}\xspace}
Dell, Lapinskas, and Meeks~\cite{DellLM22} showed that approximate counting of colorful induced subgraphs reduces to colorful detection.
Given an oracle $O$ that takes as input a graph $G=(V,E)$ and a coloring $\phi:V\to[k]$, and determines whether $G$ contains a colorful induced copy of a fixed graph $H$ on $k$ vertices, they design an algorithm \ccount that outputs a $(1\pm \eps)$-approximation to the number of colorful induced copies of $H$ in $G$. This algorithm calls the oracle $O$ at most $T=(k\log n)^{O(k)}/\eps^{2}$ times and runs in total time $O(nT)$.

\begin{theorem}[Dell-Lapinskas-Meeks {\cite[Theorem 1.1]{DellLM22}}]\label{thm:dlm}
    Let $\GG$ be a $k$-partite hypergraph with vertex set $V=V_1\sqcup V_2\sqcup \ldots\sqcup V_k$, where $\abs{V}=n$,
    and a set $E$ of $m$ (unknown) hyperedges.
    There exists a randomized algorithm $\ccount(\GG,\eps)$ that takes as input the vertex set $V$ and has access to a {\em colorful independence oracle} $O$, which, given $X_1\subseteq V_1,\ldots,X_k\subseteq V_k$, returns whether the subhypergraph of $\GG$ induced by $\cup_i X_i$ contains a hyperedge. \ccount outputs $\hat{m}$ such that $\Pr{\hat{m}=m\apm}\geq 1-1/n^4\;.$
    \ccount runs in $O(nT)$ time and queries $O$ at most $T$ times, where
    $T=(k\log n)^{O(k)}/\eps^{2}$.
\end{theorem}
The hypergraph they consider is the following canonical hypergraph.
Given the graph $G$ and coloring $\phi:V(G)\to[k]$, define the $k$-partite hypergraph $\GG$ with parts $V_i=\{v\in V(G):\phi(v)=i\}$ for $i\in[k]$, and hyperedges corresponding to colorful induced copies of $H$ in $G$.

By combining this framework with our colorful detection algorithm for induced diamonds (\Cref{lem:colorful-det}), namely an efficient implementation for the oracle $O$, we obtain an approximate counting algorithm for induced diamonds, as well as all other $4$-vertex patterns excluding a clique and an independent set.
We state the result for induced diamonds:

\begin{corollary}\label{thm:colorful 4-vertex patterns}
    There exists a randomized algorithm ${\normalfont\textsc{ApproxDiamond}}(G,\eps)$ that outputs a value $\hat{t}$ satisfying $(1-\eps)t \leq \hat{t} \leq (1+\eps)t$ in time $\TO{n^\omega/\eps^{O(1)}}$, where $t$ is the number of induced diamonds in $G$. The algorithm succeeds \whp.
\end{corollary}

\subsection{Induced $C_4$ Detection}
\newcommand{\trt}{\sqrt{rt}}
\label{ssec:induced-c4}
In this subsection we extend the $r$-light induced-diamond analysis to $r$-light induced $C_4$.
We say that an induced $C_4$ is $r$-light if none of its edges is contained in an $r$-clique of $G$, or equivalently, no pair of its vertices belongs to an $r$-clique in $G$.
Let $t_r$ be the number of $r$-light induced $C_4$ copies in $G$.
We prove:
\begin{theorem}
    \label{thm:induced-c4}
    There exists a randomized algorithm that detects an induced $C_4$ in $G$ \whp, in time $\TO{\MM{n,n,n\sqrt{r/t_r}}}$, for any integer $r\geq 1$.
\end{theorem}

We show that this case is not resolved immediately by previous techniques.
\paragraph{Why Related Work is Insufficient.} 
Induced $C_4$ detection admits a randomized $\tilde O(n^\omega)$ algorithm via the algebraic approach of~\cite{williams2014finding} (reviewed earlier).
We ask whether one can obtain faster sensitive algorithms for $r$-light induced $C_4$ detection, following our $r$-light induced-diamond results.
Before proceeding, we note the known results for graphs with no $r$-cliques.
If $G$ is triangle-free, then every $C_4$ is induced. Since a $C_4$ can be found in time $O(n^2)$, this yields an $O(n^2)$-time induced-$C_4$ detector for triangle-free graphs.
For larger values of $r$, we also get improvements.
Let $\omega(G)$ denote the clique number of $G$, i.e., the size of the largest clique in $G$.
Let $G$ be a graph with $\omega(G)\leq r$, and average degree $d=2m/n$.
If $G$ is induced-$C_4$-free, then $\omega(G)=\Omega(d^2/n)$ by~\cite{gyarfas2015cliques}.
In other words, we can bound $d$ in terms of $r$ and $n$ as follows:
\begin{align*}
    r & \geq \omega(G) = \Omega(d^2/n) \implies d = O(\sqrt{rn})\;.
\end{align*}
Therefore, $m=nd/2=O(n^{3/2}\cdot \sqrt{r})$.
Since our graph is sparse, we can use the sparse induced $C_4$ detection algorithm of~\cite{williams2014finding} that runs in time $\tilde O(m^\beta)$, where
\[
    \beta=\frac{4\omega-1}{2\omega+1}<1.48\;,
\]
which yields a running time of
\[
    \tilde O\!\left((n^{3/2}\sqrt r)^\beta\right)
    = \tilde O\!\left(n^{3\beta/2}\, r^{\beta/2}\right)
    \approx \tilde O\!\left(n^{2.22}\, r^{0.74}\right)\;,
\]
which improves over $\tilde O(n^\omega)$ whenever $r=o(n^{2\omega/\beta-3})\approx o(n^{0.21})$.
However, even for $r=4$ this is still far from $O(n^2)$ time.
Finally, we note the following deterministic result. The work of Kowaluk and Lingas~\cite{kowaluk2019fast} gives a deterministic induced $C_4$ detection algorithm when the input graph has no clique of size $r$:
\begin{theorem}[{\cite{kowaluk2019fast}}]\label{thm:Kowaluk}
    Deterministic induced $C_4$ detection can be done in time
    \begin{align*}
        \begin{cases}
            O(n^{2.5719}r^{0.3176} + n^{2}r^{2})           & \text{for } r < n^{0.394}    \\
            O(n^{2.5}r^{0.5} \phantom{00000} + n^{2}r^{2}) & \text{for } r \geq n^{0.394}
        \end{cases}
    \end{align*}
    where $r$ is the size of the largest clique in $G$.
\end{theorem}
Note that the stated bound never goes below $n^{2.5719}$.
Consequently, this deterministic approach is asymptotically slower (in $n$) than the randomized $\tilde O(n^\omega)$ algorithm of~\cite{williams2014finding}.
This completes the related work discussion.

~\\The first tool we need is a fast colorful detection algorithm for induced $C_4$. Following \Cref{lem:colorful-det c4} and \Cref{claim:runtime w3}, we obtain:
\begin{claim}\label{claim-c4:runtime w3}
    Let $P$ be a sampling vector and let $H\gets \Gp[P]$.
    Then the running time of the algorithm ${\normalfont\textsc{DetectColorfulInducedC4}}(H)$ is $\To(\MM{n,n,n\cdot \wc(P)})$ \whp.
\end{claim}
To prove \Cref{thm:induced-c4}, it suffices to show a sampling vector $P$ with $\wc(P)=\TO{\tsr}$ and $\alpha(P,\Gp)=\tg(1)$, where $\Gp$ is the colorful graph obtained from $G$ via the random-coloring reduction of \Cref{ssec:reduction}.
\begin{theorem}[{$\alpha(P,\Gp)$}]
    \label{thm-c4:helper}
    There exists a sampling vector $P$ such that $\alpha(P,\Gp)=\tg(1)$, and $\wc(P)=\TO{\tsr}$.
\end{theorem}
\begin{proof}[Proof of \Cref{thm:induced-c4} using \Cref{thm-c4:helper}]
    By \Cref{claim-c4:runtime w3}, there exists a randomized algorithm that detects an induced $C_4$ in $G$ \whp in time $\TO{\MM{n,n,n\cdot \wc(P)}}$.
    Using $\wc(P)=\TO{\tsr}$ from \Cref{thm-c4:helper} gives running time $\TO{\MM{n,n,n\cdot \tsr}}$.
\end{proof}

One key difference from induced-diamond detection is that an induced $C_4$ has no distinguished \degc/\degb vertices: all four vertices play the same role.
Thus we cannot use the same win--win argument based on either many \degc vertices or a good sampling vector.
Instead, we exploit the symmetry of $C_4$ to show that a good sampling vector always exists.

As for $r$-light diamonds, we first reduce to a colorful instance (\Cref{ssec:reduction}) so that many $r$-light induced $4$-cycles become colorful.
Recall that an induced $4$-cycle is $r$-\emph{light} if none of its edges is contained in an $r$-clique of $G$.
Let $\Gx$ be the graph obtained by taking $\TO{1}$ independent copies of $G$ and applying a random $2$-coloring $\psi: V(\Gx)\to[2]$ (analogous to the four-coloring in \Cref{ssec:reduction}).
Let $V_i\triangleq \psi^{-1}(i)$ for $i\in[2]$.
Let $\CC_\psi\subseteq \CC_{<r}$ be the set of induced $4$-cycles in $\Gx$ such that for $C=(v,u,v',u')$ (with missing edges $(v,v')$ and $(u,u')$) we have $\psi(v)=\psi(v')=1$ and $\psi(u)=\psi(u')=2$.
We call such cycles \emph{colorful} with respect to $\psi$.
Let $t=|\CC_\psi|$; by the same argument as in \Cref{ssec:reduction}, we have $t=\Omega(|\CC_{<r}|)$.

We later extend $\psi$ to a random coloring $\phi: V(\Gx)\to[4]$.
For every $v\in V(\Gx)$ let $\phi(v)\triangleq 2(\psi(v)-1)+ \gamma(v)$, where $\gamma: V(\Gx)\to[2]$ is a random $2$-coloring.
This means that every vertex $v\in V(\Gx)$ with $\psi(v)=1$ is assigned a random color $\phi(v)\in\{1,2\}$ independently, and every vertex $v\in V(\Gx)$ with $\psi(v)=2$ is assigned a random color $\phi(v)\in\{3,4\}$ independently.

The reason for this two-step coloring is that we want to define several hypergraphs based on $\psi$ first, and reveal the $4$-coloring only later.
Let $\CC_\phi$ denote the set of colorful $4$-cycles in $\Gx$ with respect to $\phi$; it is a random subset of $\CC_\psi$.
Let $E_\phi$ be the set of edges in $\Gx$ whose endpoints receive distinct colors under $\phi$, and let $U_i$ denote the $i$th color class under $\phi$.
We now define the hypergraphs used in the analysis:
\begin{align*}
    \GG^{(2)} & =(((V_1\times V_1)\sqcup(V_2\times V_2)),\CC_{\psi})\;, \\
    \GG^{(3)} & =(((V_1\times V_1)\sqcup V_2\sqcup V_2),\CC_{\psi})\;,  \\
    \GG^{(4)} & =(V_1\sqcup V_1\sqcup V_2\sqcup V_2,\CC_{\psi})\;,      \\
    \Gp       & =(U_1\sqcup U_2\sqcup U_3\sqcup U_4,E_{\phi})\;.
\end{align*}
Note that we abuse notation and use $\CC_{\psi}$ to denote the hyperedges in all hypergraphs.
For every $i\in[4]$, let $\alpha(P,\GG^{(i)})$ denote the probability that $\GG^{(i)}[P]$ contains a hyperedge, where $P$ is a sampling vector of dimension $i$.
We also write $\alpha(P,\Gp)$ for the probability that $\Gp[P]$ contains a colorful $4$-cycle.
Recall that our goal is to show the existence of a light sampling vector $P$ for $\Gp$ with $\alpha(P,\Gp) = \tg(1)$, i.e., \Cref{thm-c4:helper}.

~\\The following claim shows that for the same sampling vector $P^{(4)}$, hitting a cycle in $\Gp$ is within a factor $8$ of hitting a hyperedge in $\GG^{(4)}$.
\begin{claim}\label{claim2:xyz}
    $\alpha(P^{(4)},\Gp)\geq \frac{1}{4}\cdot \alpha(P^{(4)},\GG^{(4)})$.
\end{claim}
\begin{proof}[Proof of \Cref{claim2:xyz}]
    Expose the random sample according to $P^{(4)}$ before exposing the auxiliary coloring $\gamma$.
    If $\GG^{(4)}[P^{(4)}]$ contains no hyperedge, there is nothing to prove.
    Otherwise, fix one hyperedge $C=(v_1,u_1,v_2,u_2)\in\CC_{\psi}$ contained in $\GG^{(4)}[P^{(4)}]$, chosen deterministically from the sampled instance.
    Since $C\in\CC_{\psi}$, we have $v_1,v_2\in V_1$ and $u_1,u_2\in V_2$.
    Conditional on the sampled instance, the coloring $\gamma$ is still independent.
    With probability $1/4$, we have $\gamma(v_1)\neq\gamma(v_2)$ and $\gamma(u_1)\neq\gamma(u_2)$.
    On this event, $C$ is colorful with respect to $\phi$, so $\Gp[P^{(4)}]$ contains a colorful $4$-cycle.
    Thus $\alpha(P^{(4)},\Gp)\geq \frac{1}{4}\alpha(P^{(4)},\GG^{(4)})$, which implies the stated bound.
\end{proof}

By \Cref{claim2:xyz}, it suffices to find $P^{(4)}$ such that $\alpha(P^{(4)},\GG^{(4)})\geq \tg(1)$, instead of $\alpha(P^{(4)},\Gp)\geq \tg(1)$.

\begin{lemma}[{$\alpha(P,\GG^{(4)})$}]
    \label{thm2:helper}
    There exists a sampling vector $P=(p_1,p_2,p_3,p_4)$ such that $\alpha(P,\GG^{(4)})\geq \tg(1)$, and $\wc(P)=\TO{\sqrt{r/t}}$.
\end{lemma}

Let $A\triangleq V_1\times V_1$ and $B\triangleq V_2\times V_2$ denote the two parts of the hypergraph $\GG^{(2)}$.
The first step is to choose an induced subhypergraph of $\GG^{(2)}$ in which all vertices in $A$ (or in $B$) have similar degrees.
For every element $a\in A$ let $\deg^{(2)}(a)$ denote the degree of $a$ in the hypergraph $\GG^{(2)}$, and similarly for every $b\in B$.
Recall that $t$ is the total number of hyperedges in $\GG^{(2)}$.
Partition the elements of $A$ and $B$ into $\log t + 1$ buckets, where each bucket contains elements of similar degree in $\GG^{(2)}$:
\begin{align*}
    A_i & = \set{a\in A:\deg^{(2)}(a)\in[2^{i-1},2^i)}\;, \\
    B_i & = \set{b\in B:\deg^{(2)}(b)\in[2^{i-1},2^i)}\;.
\end{align*}
We say that a set $A_i$ is \emph{heavy} if $\abs{E(A_i)}\geq t/(6\log t)$,
where $E(A_i)$ is the set of hyperedges of $\GG^{(2)}$ that contain an element from $A_i$, and similarly for $B_i$.
Note that there is at least one index $i\in[\log(t)]$ such that $A_i$ is heavy; otherwise $E(A)< t/(6\log t)\cdot \log(t)<t/6$, a contradiction.
Let $i$ be the smallest index such that $A_i$ is heavy.
There is also at least one index $j\in[\log(t)]$ such that $B_j$ is heavy.

\begin{claim}\label{claim:small heavy}
    At least one of the following holds:
    \begin{enumerate}
        \item There exists an index $i\geq \log(\trt)$ such that $E(A_i)\geq \frac{t}{36\log^2 t}$.
        \item There exists an index $j\geq \log(\trt)$ such that $E(B_j)\geq \frac{t}{36\log^2 t}$.
        \item There exists a subgraph $\FF^{(2)}\subseteq\GG^{(2)}$ with at least $\tg(t)$ edges and maximum degree at most $\trt$.
    \end{enumerate}
\end{claim}

\begin{proof}[Proof of \Cref{claim:small heavy}]
    Let $i$ be the largest index such that $A_i$ is heavy.
    If $i\geq \log(\trt)$, then condition $(1)$ holds, and we are done.
    Otherwise, let $B'=N_{\GG^{(2)}}(A_i)$ be the set of neighbors of $A_i$ in $\GG^{(2)}$, and let $\HH^{(2)}=\GG^{(2)}[A_i\cup B']$ be the induced hypergraph on $A_i\cup B'$.
    The hypergraph $\HH^{(2)}$ has at least $t/(6\log t)$ hyperedges.
    We partition the set $B'$ into $\log(t)$ sets $B_1',\ldots, B_{\log(t)}'$, exactly as before, but with respect to the degrees in $\HH^{(2)}$:
    \begin{align*}
        B_j' = \set{b\in B':\deg_{\HH^{(2)}}(b)\in[2^{j-1},2^j)}\;.
    \end{align*}
    Let $j$ be the largest index such that $B_j'$ is heavy in $\HH^{(2)}$, meaning it contains at least a $1/(6\log t)$ fraction of the hyperedges of $\HH^{(2)}$, and therefore at least $t/(36\log^2 t)$ hyperedges.
    If $j\geq \log(\trt)$, then condition $(2)$ holds, and we are done.
    Otherwise, let $A'=N_{\HH^{(2)}}(B_j')$ be the set of neighbors of $B_j'$ in $\HH^{(2)}$. Let $\FF^{(2)}$ be a subgraph of $\HH^{(2)}$ with the same vertex set, and with the edge set $E(A',B_j')$.
    The graph $\FF^{(2)}$ has at least $t/(36\log^2 t)$ hyperedges and maximum degree at most $\trt$, so condition $(3)$ holds.
\end{proof}
Assuming that condition $(3)$ of \Cref{claim:small heavy} holds, we prove \Cref{thm2:helper} using the following claim.
\begin{restatable}[Low Maximum Degree]{claim}{ClaimSimpler}
    \label{claim:simpler}
    Let $H$ be a graph with $m$ edges and maximum degree at most $\sqrt{rm}$.
    Keep each vertex of $V(H)$ independently with probability $q=\sqrt{r/m}$, and let $H[q]$ be the induced subgraph on the kept vertices.
    Then, $\Pr{H[q]\text{ contains an edge}}\geq 1/4$.
\end{restatable}
\begin{proof}[Proof of \Cref{thm2:helper} using \Cref{claim:simpler}]
    Let $\FF^{(2)}$ be the subgraph of $\GG^{(2)}$ as stated in condition $(3)$ of \Cref{claim:small heavy}.
    Let $q=\tsr$, and define $P^{(2)}=(q,q)$.
    Then, by \Cref{claim:simpler}, we have that $\alpha(P^{(2)},\FF^{(2)})\geq 1/4$.
    We refine $P^{(2)}$ into $P=(p_1,p_2,p_3,p_4)$, where $p_1\cdot p_2=\tf(q)$, and $p_3\cdot p_4=\tf(q)$, such that
    \begin{align*}
        \alpha(P,\HH^{(4)})
        \geq \alpha(P^{(2)},\FF^{(2)})/(2\log n)^2\;,
    \end{align*}
    where such refinement is possible by \Cref{thm:refining sampling vector} (applied twice) and it also guarantees that $p_1,p_2,p_3,p_4\geq q$, so $\min(P)\geq q$.
    Since $\alpha(P^{(2)},\FF^{(2)})\geq \tg(1)$, we get $\alpha(P,\HH^{(4)})\geq \tg(1)$.
    Finally, we show that $\wc(P)=\TO{\tsr}$.
    \begin{align*}
        \wc(P)=\frac{w(P)}{\min(P)}\leq \frac{w(P)}{\tf{q}}=\tf(q)=\tf(\tsr)\;. & \qedhere
    \end{align*}
    \end{proof}
The proof of \Cref{claim:simpler} uses a simple second moment argument, and is given towards the end of this subsection.

\renewcommand{\cg}{\lambda}
We prove \Cref{thm2:helper} under the assumption that either condition $(1)$ or $(2)$ of \Cref{claim:small heavy} holds.
We may assume without loss of generality that (1) holds, i.e., there exists an index $i\geq \log(\trt)$ such that $E(A_i)\geq \frac{t}{36\log^2 t}$.
If the second condition holds, we can rename the vertices in $V_1$ and $V_2$, as well as the roles of $A$ and $B$, and proceed with the same analysis.
Let $E(A_i)$ be the set of hyperedges of $\GG^{(2)}$ that contain an element from $A_i$.
We restrict our attention to the hypergraph $\HH^{(j)}$ with the same vertex set as $\GG^{(j)}$, but with the edge set $E(A_i)$, for every $j\in\{2,3,4\}$.
Define
\begin{align}
    P^{(2)}=(q_1=\frac{2^{i+1}\cdot \cg}{t},q_2=\frac{1}{2^{i-1}})\;,\label{eq2:P2}
\end{align}
where $\cg\triangleq 5000\log^{3}n$, and $i$ is the index of the chosen heavy set $A_i$. Note that $q_1\geq \cg\cdot \tsr$, and $\Omega(1/t)\leq q_2\leq 1/\trt$.
\begin{lemma}
    \label{prop3:balanced2 helper}
    Define $P^{(3)}=(q_1,p,p)$ where $p=\sqrt{r\cdot q_2}$, and $q_1$ and $q_2$ are as in $P^{(2)}$.
    Then $\alpha(P^{(3)},\HH^{(3)})\geq \tg(1)$.
\end{lemma}

\begin{proof}[Proof of \Cref{thm2:helper} using \Cref{prop3:balanced2 helper}]
    Let $P^{(3)}=(q_1,p,p)$ be as in \Cref{prop3:balanced2 helper}, meaning that $p=\sqrt{r\cdot q_2}$, and additionally $q_1$ and $q_2$ are as in $P^{(2)}$.
    We refine $P^{(3)}$ into $P=(p_1,p_2,p,p)$, where $p_1\cdot p_2=\TO{q_1}$, using \Cref{thm:refining sampling vector}, which guarantees that $\alpha(P,\HH^{(4)})\geq \alpha(P^{(3)},\HH^{(3)})/(2\log n)$.
    By substituting the values of $q_1,q_2$ and $p$, we get
    \begin{align*}
        w(P)=\TO{w(P^{(3)})}=\TO{q_1\cdot p^2}=\TO{q_1\cdot r\cdot q_2}=\TO{r/t}\;.
    \end{align*}
    We have that
    \begin{align*}
        8\alpha(P,\Gp)
        \geq \alpha(P,\GG^{(4)})
        \geq \alpha(P,\HH^{(4)})
        \geq \frac{\alpha(P^{(3)},\HH^{(3)})}{2\log n}
        \geq \tg(1)\;,
    \end{align*}
    The inequalities follow for the following reasons.
    The first inequality follows by the same argument as \Cref{claim2:xyz}, the second follows because $\HH^{(4)}$ is a subhypergraph of $\GG^{(4)}$, the third follows by the guarantees of \Cref{thm:refining sampling vector}, and the last inequality follows by \Cref{prop3:balanced2 helper}.
    We are left with showing that $\wc(P)=\TO{\tsr}$.

    \begin{itemize}
        \item
              If $p_1$ (or $p_2$) is the smallest coordinate of $P$,
              then
              \begin{align*}
                  \wc(P)=\frac{w(P)}{\min(P)}\leq \frac{w(P)}{q_1} =\frac{\To(r/t)}{\tsr}=\To(\tsr)\;.
              \end{align*}
              Here we used $p_1\cdot p_2= \tg(q_1)$ by \Cref{thm:refining sampling vector}, and $q_1=\tg(\tsr)$ by the definition of $P^{(2)}$.
        \item If $p=\min(P)$, then $\wc(P)= \frac{\To(r/t)}{p}\leq\frac{\To(r/t)}{\tsr}=\To(\tsr)$.
              The first inequality follows from $q_2\geq \Omega(1/t)$, which follows by the definition of $P^{(2)}$.
    \end{itemize}
\end{proof}
Note that the statement of \Cref{prop3:balanced2 helper} is identical to \Cref{prop2:balanced2 helper}, and so is the proof.
The only difference is that in \Cref{prop2:balanced2 helper} we use the assumption that no three vertices in some hyperedge belong to the same $r$-clique in $G$.
Here, we have a stronger assumption: no \emph{two} vertices from some hyperedge belong to the same $r$-clique in $G$.
We now complete the proof of \Cref{thm2:helper} by proving \Cref{claim:simpler}.
\begin{proof}[Proof of \Cref{claim:simpler}]
    We use the second moment method.
    Let $E=(e_1,\ldots,e_{m})$ be the set of edges in $H$.
    Let $X_i$ be the indicator random variable that is $1$ if edge $e_i$ is sampled into $H[q]$, and $0$ otherwise.
    Define $X=\sum_{i=1}^{m} X_i$, where $\Exp{X}=m\cdot q^2$.
    We bound $\Exp{X^2}$ by $\Exp{X}(1+ 2\Delta\cdot q + \Exp{X})$, as we previously did in \Cref{prop2:balanced2 helper}.
    Now we apply the second moment method to show that $\Pr{X>0}\geq 1/4$.
    \begin{align*}
        \Pr{X>0}
        \geq \frac{\Exp{X}^2}{\Exp{X^2}}
        \geq \frac{\Exp{X}^2}{\Exp{X}\brak{1 + 2\Delta\cdot q + \Exp{X}}}
        = \frac{\Exp{X}}{1 + 2\Delta\cdot q + \Exp{X}}\;.
    \end{align*}
    If we show that $\Delta\cdot q\leq \Exp{X}$, we get that
    \begin{align*}
        \Pr{X>0}
        \geq \frac{\Exp{X}}{1 + 2\Delta\cdot q + \Exp{X}}
        \geq 1/4\;,
    \end{align*}
    which completes the proof.
    To complete the proof, we plug in $q=\sqrt{r/m}$ and $\Delta=\sqrt{rm}$, to obtain
    \begin{align*}
        \Delta\cdot q = \sqrt{rm}\cdot \sqrt{r/m} = r \leq m\cdot q^2 = \Exp{X}\;. & \qedhere
    \end{align*}
\end{proof}

\subsection{4-SUM}
\newcommand{\algSum}{{\normalfont\textsf{Find4SUM}}\xspace}
\renewcommand{\PP}{\QQ}
The main result of this subsection is as follows.
\begin{theorem}
    \label{thm:4-sum}
    There is an approximate counting algorithm for \dsum that runs in time $n^{2-x/3+o(1)}$, where $t=n^x$ for $x\in[0,1]$ is the number of solutions to the \dsum instance.
\end{theorem}
We consider distinguishing between no solution to \dsum and $t$ solutions, for $t\in[1,n]$:
\begin{theorem}[Detection]
    \label{thm-4sum:detection}
    There is a randomized algorithm that distinguishes between no solution to \dsum and $t$ solutions for $t\in[1,n]$, and runs in time $\TO{n^{2}/t^{1/3}+n}$ \whp.
\end{theorem}
This improves upon \cite{CEW25}, which runs in time $\max(n^2,\tf(n^2\cdot \sqrt{n/t}))$, for $t\leq n^3$. 
In the regime $t\leq n$, our algorithm is the first non-trivial witness-sensitive algorithm for \dsum.
For $t\geq n^3$ the na\"ive algorithm of sampling four elements and checking if they sum to zero runs in time $O(n^4/t)$, which is better than both algorithms, and also sublinear for $t\gg n^3$.

\paragraph{Preliminaries.} We review the \dsum problem and the actual variant we consider.
A \dsum instance consists of one array $A$ of $n$ integers, and the goal is to determine whether there exist four elements $a_1,a_2,a_3,a_4\in A$ such that $\sum_{i=1}^4 a_i=0$, or to approximate the number $t$ of such quadruples, which we refer to as \emph{solutions}.
A variant of the problem considers four arrays $A_1,A_2,A_3,A_4$ of integers
where the goal is to determine whether there exist four elements
$a_i\in A_i$ such that $\sum_{i=1}^4 a_i=0$, i.e., one element from each array.
We refer to \dsum instances on one array as single-array \dsum instances, and to instances on four arrays as four-array \dsum instances.
There is a well-known reduction from single-array \dsum to four-array \dsum: assuming the entries of $A$ lie in $[-U,U]$, define new arrays $A_i=A + 3^i \cdot 10U$ for $i\in[3]$, and let $A_4=A-U'$, where $U'\triangleq 10U\cdot(3+3^2+3^3)$.
Then, every solution $(a_1,a_2,a_3,a_4)$ to the single-array \dsum instance corresponds to $4!$ solutions in the four-array \dsum instance.

\Cref{thm:4-sum} follows by combining \Cref{thm-4sum:detection} with the following reduction by \cite{CEW25} from approximate counting to detection:
\begin{theorem}[{Approximate Counting to Detection \cite{CEW25}}]
    \label{thm:4-sum-approx-counting}
    Given a randomized algorithm that detects a solution to four-array 
    \dsum in time $f(n,t)$, with probability at least $2/3$, assuming there are at most $t$ solutions, there exists a randomized algorithm that approximates the number of solutions to four-array \dsum within a $(1\pm \eps)$ factor, \whp, in time $\poly(\log n,1/\eps)\cdot f(n,\poly(\log n/\eps))$.
\end{theorem}

For the rest of this section, we prove \Cref{thm-4sum:detection}.
Before that, we provide some points on the difference between witness-sensitive detection for \dsum and for induced diamond or $C_4$ detection.
Recall that in \Cref{ssec:subsampling} we had a graph $G$, with $4$-coloring $\phi:V(G)\to[4]$, and $t$ colorful induced diamonds.
We consider the set of all sampling vectors $\PP$ defined by:
\begin{align*}
    \PP \triangleq\set{P\in \FF^4 \mid \alpha(P)\geq \gamma}\;,
\end{align*}
in words, the set of all sampling vectors $P\in \PP$ such that $H\gets G[P]$ contains an induced colorful diamond with probability at least $\gamma$, where $\gamma=\tg(1)$.
We then defined a ``measure'' $\mu(P)$ for every sampling vector $P\in \PP$ by $\mu(P)=\wc(P)\triangleq w(P)/\min(P)$, and we defined $\mu(\PP)=\min_{P\in \PP}\mu(P)$.
We also showed a colorful detection algorithm \textsc{DetectDiamond} that on $H\gets G[P]$ runs in time $\tf(\MM{n,n,n\cdot \mu(P)})$ \whp.
Finally, in \Cref{ssec:main-light} we showed that $\PP$ has a sampling vector $P$ with $\mu(P)=\TO{\sqrt{r/t}}$, which gave us the desired running time.

Here, we follow the same approach with the following differences.
First, we consider \dsum instances instead of colorful graphs with induced diamonds. The main property that we use is that every triplet of elements in $A$ participates in at most one solution, since the fourth element is determined by the first three. This is not the same assumption as $r$-light diamonds or $4$-cycles.
Second, given an instance $A^{(4)}\triangleq(A_1,A_2,A_3,A_4)$ of four-array \dsum, and a sampling vector $P=(p_1,p_2,p_3,p_4)$, we define the sampled instance $B^{(4)}=(B_1,B_2,B_3,B_4)$, where each $B_i$ is obtained by sampling each element of $A_i$ independently with probability $p_i$.
We can again define $\PP$ to be all sampling vectors $P$ such that the sampled instance $B^{(4)}\gets A^{(4)}[P]$ contains a solution with probability at least $\gamma=\tg(1)$.
However, we have to define a new measure $\mu(P)$ for the time it takes to solve $B^{(4)}\gets A^{(4)}[P]$.
Such a measure was already defined in \cite{CEW25} for unbalanced $k$-array \ksum detection; here, we restate it for four-array \dsum.
\begin{theorem}\label{cor:53y}
    Let $A_1,A_2,A_3,A_4$ be four arrays of integers, with sizes $n_i=|A_i|$ for $i\in[4]$.
    Let $\Pi$ be the set of all permutations of $[4]$.
    Detection takes
    \begin{align*}
        \min_{\pi\in \Pi} O(\max(n_{\pi(1)}n_{\pi(2)},n_{\pi(3)}n_{\pi(4)}))\;.
    \end{align*}
\end{theorem}
We use a simpler bound, which is immediately implied by \Cref{cor:53y}.
\begin{corollary}
    \label{lem:4-sum-detection}
    Assume that $n_i=|A_i|$ and that $n\geq n_1\geq n_2\geq n_3\geq n_4$.
    Then detection takes $O(\max(n_1\cdot n_4,n_2\cdot n_3))$ time.
\end{corollary}
The detection algorithm simply computes the Minkowski sum $A_1+A_4$ and $-(A_2+A_3)$ and checks for an intersection using hashing, in time $O(n_1n_4+n_2n_3)$.
We are ready to define our measure for four-array \dsum:
for every sampling vector $P=(p_1,p_2,p_3,p_4)$ we define a \emph{measure} $\mu_2(P)$ by
\begin{align}
    \mu_2(P)\triangleq\min_{\pi\in \Pi}\max(p_{\pi(1)}\cdot p_{\pi(2)},p_{\pi(3)}\cdot p_{\pi(4)})\;. \label{eq:mu2}
\end{align}
We also define $\mu_2(\PP)=\min_{P\in \PP}\mu_2(P)$.
Using simple concentration bounds we get:
\begin{claim}\label{claim:53mu2 single}
    The running time of $\algSum(B^{(4)})$ where $B^{(4)}\gets A^{(4)}[P]$ is bounded by $\TO{n^2\cdot \mu_2(P)+n}$ \whp.
\end{claim}

\paragraph{The Algorithm.}
We are now ready to describe our detection algorithm for four-array \dsum, which is the \emph{canonical} algorithm: simply try all sampling vectors $P\in\set{1,1/2,1/4,\ldots, 1/n}^4$ in a round-robin fashion.
That is, given an instance $A^{(4)}\triangleq(A_1,A_2,A_3,A_4)$ of four-array \dsum, consider all sampling vectors $P=(p_1,p_2,p_3,p_4)$, and let $B^{(4)}=(B_1,B_2,B_3,B_4)$ be the sampled instance in which each $B_i$ is obtained by sampling each element of $A_i$ independently with probability $p_i$.
We use $B^{(4)}\gets A^{(4)}[P]$ to denote this sampling process.
Then we run the detection algorithm from \Cref{lem:4-sum-detection} on the sampled instance $B^{(4)}$.
We use $\algSum$ to denote the algorithm that we run on each sample $B^{(4)}$.
\begin{claim}\label{claim:53can}
    The running time of the canonical algorithm is $\TO{n^2\cdot \mu_2(\PP)+n}$ \whp.
\end{claim}
\begin{proof}[Proof of \Cref{claim:53can}]
    We run the detection algorithm $\algSum$ on $\tilde{\Theta}(1)$ independent samples of $B^{(4)}\gets A^{(4)}[P]$ for each $P\in \PP$.
    Overall, there are $\tg(1)$ instances.
    Let $P\in \PP$ be a sampling vector with $\mu_2(P)=\mu_2(\PP)$.
    For every $B^{(4)}\gets A^{(4)}[P]$ the running time of $\algSum(B^{(4)})$ is $\TO{n^2\cdot \mu_2(P)+n}$ \whp, by \Cref{claim:53mu2 single}.
    Moreover, at least one of the $\tilde{\Theta}(1)$ samples taken by $P$
    contains a solution \whp, meaning that the algorithm detects a solution in time $\TO{n^2\cdot \mu_2(P)+n}$ \whp.
    The overall running time overhead of the round-robin over all $P$ is only a $\poly(\log n)$ factor.
\end{proof}

We move on to the analysis of the canonical algorithm, where the novelty of this subsection lies.
Specifically, we prove that there exists a sampling vector $P\in \PP$ with $\mu_2(P)=\TO{1/t^{1/3}}$.
\paragraph{Analysis.}
The key structural observation here is that every triplet of elements in $A$ participates in at most one solution to \dsum, since the fourth element is determined by the first three.
This allows us to refine the sampling vector in a way that is not possible in general hypergraphs, and therefore to obtain a faster detection algorithm for four-array \dsum.
As before, we translate the problem to hypergraphs, so that we can use our refined sampling technique.
Let $\EE$ denote the set of such quadruples, and let $t=|\EE|$.
We define the following hypergraph to represent the \dsum instance.
\begin{align*}
    \GG^{(1)} & =((A_1\times A_2\times A_3\times A_4),\EE)\;,      \\
    \GG^{(2)} & =((A_1 \times A_2) \sqcup (A_3 \times A_4),\EE)\;, \\
    \GG^{(3)} & =((A_1 \times A_2) \sqcup A_3 \sqcup A_4,\EE)\;,   \\
    \GG^{(4)} & =(A_1 \sqcup A_2 \sqcup A_3 \sqcup A_4,\EE)\;.
\end{align*}

Recall that $\QQ$ is the set of all sampling vectors $P^{(4)}=(p_1,p_2,p_3,p_4)$ for $A^{(4)}=(A_1,A_2,A_3,A_4)$ or $\GG^{(4)}$, with $\alpha(P^{(4)},\GG^{(4)})=\tg(1)$, and that $\mu_2(\QQ)=\min_{P^{(4)}\in \QQ}\mu_2(P^{(4)})$.
Finally, recall that the running time of our canonical algorithm is bounded by $\TO{n^2\cdot \mu_2(\QQ)+n}$ \whp.
By \Cref{claim:53can} it suffices to prove that $\mu_2(\QQ)\leq \tf(1/t^{1/3})$, to complete the proof of \Cref{thm-4sum:detection}.
\begin{lemma}
    \label{lemma-4sum:sp}
    $\mu_2(\PP)\leq 1/t^{1/3}$.
\end{lemma}
For the rest of this section, we prove \Cref{lemma-4sum:sp}.
Let $\PP^{(2)}$ be the set of all sampling vectors $P^{(2)}=(q_1,q_2)$ for $\GG^{(2)}$, with $\alpha(P^{(2)},\GG^{(2)})=\tg(1)$.
Define
\begin{align*}
    \rho(P^{(2)})\triangleq \max(q_1,q_2) \quad \text{and} \quad \rho(\PP^{(2)})\triangleq \min_{P^{(2)}\in \PP^{(2)}}\rho(P^{(2)})\;.
\end{align*}

The quantity $\rho(\PP^{(2)})$ is useful because $(i)$ it is easier to bound than $\mu_2(\PP)$, and $(ii)$ it upper bounds $\mu_2(\PP)$.

\begin{claim}\label{claim-4sum:sp2}
    $\mu_2(\PP)\leq \tf(\rho(\PP^{(2)}))$.
\end{claim}
\begin{proof}[Proof of \Cref{claim-4sum:sp2}]
    Given a sampling vector $P^{(2)}=(q_1,q_2)$ for $\GG^{(2)}$ with $\alpha(P^{(2)},\GG^{(2)})=\tg(1)$, we can use \Cref{thm:refining sampling vector} twice, to obtain a sampling vector $P^{(4)}=(p_1,p_2,p_3,p_4)$
    where $p_1\cdot p_2=\tf(q_1)$ and $p_3\cdot p_4=\tf(q_2)$, such that $\alpha(P^{(4)},\GG^{(4)})=\tg(1)$.
    After sampling subarrays using $P^{(4)}$, we get four arrays $B_1,B_2,B_3,B_4$ where $|B_i|=O(np_i)$, for $i\in[4]$.
    We have
    \begin{itemize}
        \item $\abs{B_1}\cdot \abs{B_2}=\tf(n^2p_1p_2)=\tf(n^2 q_1)$
        \item $\abs{B_3}\cdot \abs{B_4}=\tf(n^2p_3p_4)=\tf(n^2 q_2)$
    \end{itemize}
    Therefore
    \begin{align*}
        \mu_2(\PP)\leq \mu_2(P^{(4)})\leq \tf(\max(q_1,q_2))=\tf(\rho(P^{(2)}))\;.
    \end{align*}
    We showed that for every $P^{(2)}\in \PP^{(2)}$, it holds that $\mu_2(\PP)\leq \tf(\rho(P^{(2)}))$. Therefore, $\mu_2(\PP)\leq \tf(\rho(\PP^{(2)}))$, which completes the proof.
\end{proof}

\renewcommand{\trt}{\sqrt{t}}
\renewcommand{\tsr}{\sqrt{1/t}}
Next, we show how to use \Cref{claim-4sum:sp2} to prove \Cref{lemma-4sum:sp}.
Let $X\triangleq A_1\times A_2$ and $Y\triangleq A_3\times A_4$ denote the two parts of the hypergraph $\GG^{(2)}$.
The first step is to choose an induced subhypergraph of $\GG^{(2)}$ in which all vertices in $X$ (or in $Y$) have similar degrees.
For every element $x\in X$ let $\deg^{(2)}(x)$ denote the degree of $x$ in the hypergraph $\GG^{(2)}$, and similarly for every $y\in Y$.
Recall that $t$ is the total number of hyperedges in $\GG^{(2)}$.
Partition the elements of $X$ (and $Y$) into $\log t + 1$ buckets, where each bucket contains elements of similar degree in $\GG^{(2)}$:
\begin{align*}
    X_i & = \set{x\in X:\deg^{(2)}(x)\in[2^{i-1},2^i)}\;, \\
    Y_i & = \set{y\in Y:\deg^{(2)}(y)\in[2^{i-1},2^i)}\;.
\end{align*}
We say that a set $X_i$ is \emph{heavy} if $\abs{E(X_i)}\geq t/(6\log t)$,
where $E(X_i)$ is the set of hyperedges of $\GG^{(2)}$ that contain an element from $X_i$, and similarly for $Y_i$.
Note that there is at least one index $i\in[\log(t)]$ such that $X_i$ is heavy; otherwise $E(X)< t/(6\log t)\cdot \log(t)<t/6$, a contradiction.
Let $i$ be the smallest index such that $X_i$ is heavy.
There is also at least one index $j\in[\log(t)]$ such that $Y_j$ is heavy.

\begin{claim}\label{claim-4sum:small heavy}
    At least one of the following holds:
    \begin{enumerate}
        \item There exists an index $i\geq \log(\trt)$ such that $E(X_i)\geq \frac{t}{36\log^2 t}$.
        \item There exists an index $j\geq \log(\trt)$ such that $E(Y_j)\geq \frac{t}{36\log^2 t}$.
        \item There exists a subgraph $\FF^{(2)}\subseteq\GG^{(2)}$ with at least $\tg(t)$ edges and maximum degree at most $\trt$.
    \end{enumerate}
\end{claim}
The proof of \Cref{claim-4sum:small heavy} follows the proof of \Cref{claim:small heavy}, and is therefore omitted.

Assuming that condition $(3)$ of \Cref{claim-4sum:small heavy} holds, we prove \Cref{lemma-4sum:sp} using \Cref{claim:simpler}, which we restate below:
\ClaimSimpler* 
\begin{proof}[Proof of \Cref{lemma-4sum:sp} Assuming Condition (3)]
    Let $\FF^{(2)}$ be the subgraph of $\GG^{(2)}$ as stated in condition $(3)$ of \Cref{claim-4sum:small heavy}.
    Let $m'=\abs{E(\FF^{(2)})}$, and apply \Cref{claim:simpler} with $r=1$.
    Let $q=1/\sqrt{m'}=\tf(\tsr)$, and define $P^{(2)}=(q,q)$.
    Then, by \Cref{claim:simpler}, we have that $\alpha(P^{(2)},\FF^{(2)})\geq 1/4$, meaning that $P^{(2)}\in \PP^{(2)}$, and therefore that $\rho(\PP^{(2)})\leq \max(P^{(2)})=q=\tf(\tsr)$, proving
    \Cref{lemma-4sum:sp} by \Cref{claim-4sum:sp2}.
\end{proof}

\renewcommand{\cg}{\lambda}
We prove \Cref{lemma-4sum:sp} under the assumption that either condition $(1)$ or $(2)$ of \Cref{claim-4sum:small heavy} holds.
We may assume without loss of generality that (1) holds, i.e., there exists an index $i\geq \log(\ts)$ such that $E(X_i)\geq \frac{t}{36\log^2 t}$.
If the second condition holds, we can rename the vertices in $A_1,A_2$ and $A_3,A_4$, as well as the roles of $X$ and $Y$, and proceed with the same analysis.
Let $E(X_i)$ be the set of hyperedges of $\GG^{(2)}$ that contain an element from $X_i$.
We restrict our attention to the hypergraph $\HH^{(j)}$ with the same vertex set as $\GG^{(j)}$ and the edge set $E(X_i)$ for every $j\in\{2,3,4\}$.
Define
\begin{align}
    P_i^{(2)}=(q_1=\frac{2^{i+1}\cdot \cg}{t},q_2=\frac{1}{2^{i-1}})\;,\label{eq24:P2}
\end{align}
where $\cg\triangleq 64\log^{2}n$, and $i$ is the index of the chosen heavy set $X_i$.
\begin{claim}\label{claim7:p2}
    $\alpha(P_i^{(2)},\GG^{(2)})=\tg(1)$.
\end{claim}
By \Cref{claim7:p2} we have $P_i^{(2)}\in \PP^{(2)}$.
Since $2^i\geq \ts$ by our choice of $i$, it follows that $\max(P_i^{(2)})=q_1$, which implies $\rho(\PP^{(2)})\leq q_1$.
In the case where $2^{i}\leq \tf(t^{2/3})$, we obtain $q_1\leq \tf(1/t^{1/3})$ and hence $\rho(\PP^{(2)})\leq \tf(1/t^{1/3})$; \Cref{lemma-4sum:sp} then follows from the inequality $\mu_2(\PP)\leq \rho(\PP^{(2)})$ given by \Cref{claim-4sum:sp2}.
It therefore suffices to consider the remaining case $2^{i}>\tf(t^{2/3})$, which gives $q_2<\tf(1/t^{2/3})$.

We now leverage the fact that every triplet of vertices participates in at most one hyperedge, which permits a balanced refinement of the sampling vector.
\begin{lemma}
    \label{lemma7:refine}
    Refine $P_i^{(2)}=(q_1,q_2)$ from \Cref{eq24:P2} to $P^{(3)}=(q_1,p,p)$, where $p=5\log n \sqrt{q_2}$.
    Then $\alpha(P^{(3)},\GG^{(3)})=\tg(1)$.
\end{lemma}
We prove \Cref{lemma-4sum:sp} using \Cref{lemma7:refine}.
\begin{proof}[Proof of \Cref{lemma-4sum:sp} Using \Cref{lemma7:refine}]
    By \Cref{lemma7:refine}, there exists a sampling vector $P^{(3)}=(q_1,p,p)$ for $\GG^{(3)}$ with $\alpha(P^{(3)},\GG^{(3)})=\tg(1)$, where $p=5\log n\sqrt{q_2}$.
    ``Refine'' $P^{(3)}$ to $P^{(4)}=(1,1,p,p)$ for $\GG^{(4)}$.
    It follows that $\alpha(P^{(4)},\GG^{(4)})=\tg(1)$ by a coupling argument and \Cref{thm:refining sampling vector}.
    Clearly, $\mu_2(P^{(4)})=p$, as desired.
\end{proof}
We are left with proving \Cref{claim7:p2} and \Cref{lemma7:refine}.

\begin{proof}[Proof of \Cref{claim7:p2}]
    Let $\HH^{(2)}$ be the subhypergraph of $\GG^{(2)}$ with the same vertex set, and with edge set $E(X_i)$.
    By \Cref{claim7:bipartite} we have that
    \begin{align*}
        \alpha(P_i^{(2)},\HH^{(2)})\geq (1-\exp(-mq_1/(2d)))\cdot(1-\exp(-1))\;.
    \end{align*}
    By plugging in the values of $m\geq t/(36\log^2 t)$, $q_1=\frac{2^{i+1}\cdot \cg}{t}$, and $d\leq 2^i$, we get
    \begin{align*}
        mq_1/(2d)\geq \frac{t}{36\log^2 t}\cdot \frac{2^{i+1}\cdot \cg}{t}\cdot \frac{1}{2\cdot 2^i}=\frac{\cg}{36\log^2 t}
        =\frac{50\log^{2}n}{36\log^2 t}\geq 1/4\;,
    \end{align*}
    where the last inequality holds since $t\leq n^4$.
    Therefore,
    \begin{align*}
        \alpha(P_i^{(2)},\HH^{(2)})\geq (1-\exp(-1/4))\cdot(1-\exp(-1))
        =\Omega(1)\;.
    \end{align*}
\end{proof}

\begin{proof}[Proof of \Cref{lemma7:refine}]
    Consider the random subhypergraph $\GG^{(3)}[P^{(3)}]$.
    By the same estimate used in the proof of \Cref{claim7:p2}, the set of retained vertices from the $A_1\times A_2$ part intersects $X_i$ with probability $\Omega(1)$.
    Condition on this event, and fix a retained element $x=(a_1,a_2)\in X_i$.
    Let $N(x)$ be the set of pairs $(a_3,a_4)\in A_3\times A_4$ such that $(x,a_3,a_4)\in \EE(\GG^{(3)})$.
    We define a bipartite graph $G_x$ with vertex set $A_3\sqcup A_4$ and edge set $N(x)$, and then remove its isolated vertices.
    We claim that this graph is a matching, i.e., it has maximum degree at most $1$.
    Indeed, if there are two edges $(a_3,a_4)$ and $(a_3',a_4)$ in $E(G_x)$, then the triplet $(a_1,a_2,a_4)$ participates in two hyperedges $(a_1,a_2,a_3,a_4)$ and $(a_1,a_2,a_3',a_4)$ in $\GG^{(3)}$, contradicting the assumption that every triplet participates in at most one hyperedge.
    Since $x\in X_i$, we have that it is in $[2^{i-1},2^i)$ hyperedges in $\GG^{(2)}$, and therefore $G_x$ has $2^{i-1}$ edges, i.e., $|E(G_x)|=\Theta(1/q_2)$.
    By setting $p=5\log n \sqrt{q_2}$, and sampling each part of $G_x$ independently with probability $p$, we hit an edge with constant probability: each edge is sampled with probability $p^2$, the edges are disjoint, and therefore the probability that we do not hit any edge is at most
    \begin{align*}
        (1-p^2)^{|E(G_x)|}\leq \exp(-p^2|E(G_x)|)=\exp(-25\log^2 n\cdot q_2\cdot 1/q_2)=\exp(-25\log^2 n)\;,
    \end{align*}
    Therefore, conditioned on retaining such an element $x\in X_i$, the sampled vertices from $A_3$ and $A_4$ complete a hyperedge with $x$ with probability $1-o(1)$.
    Since the conditioning event holds with probability $\Omega(1)$, we get $\alpha(P^{(3)},\GG^{(3)})=\tg(1)$.
\end{proof}

\appendix

\section{Obtaining the Exact Running Time in \Cref{thm:simpler}}
\newcommand{\Tl}{T_\mathrm{light}}
\newcommand{\Th}{T_\mathrm{Heavy}}
In this subsection, we explain how to obtain the exact running time in \Cref{thm:simpler} from the two bounds given by \Cref{thm:r-heavy diamonds} and \Cref{thm:r-light diamonds}, which are restated below for convenience.
\ThmHeavyD*
\ThmLightD*

We use the standard linear approximation for rectangular matrix multiplication:
Let $t=n^\tau$ and $\rmax=n^\rho$.
We use $\Th(n,\rho)$, and $\Tl(n,\rho,\tau)$ to denote the running times of the heavy and light algorithms, respectively, as functions of $n$, $\rho$, and $\tau$.
Our goal is to get an upper bound on the running time of the form $n^a/t^b$ which depends only on $n$ and $t$, without any dependence on $\rmax$.
Thus, for every $t$, we find the worst-case value of $\rho$ that maximizes the running time of the faster of the two algorithms.
We have
\begin{align*}
    \Tl(n,\rho,\tau) & \overset{\text{\Cref{claim:omega rect}}}=O(n^{\omega+\frac{\beta}{2}(\rho-\tau)}+n^2)\;, \\
    \Th(n,\rho)      & = O(n^{\rho+(1-\rho)\omega}+n^{3-3\rho}+n^{1+\rho})\;.
\end{align*}
We want to find $\rho(\tau)$ such that $\Th(n,\rho(\tau))=\Tl(n,\rho(\tau),\tau)$.
First, note that if $\rho(\tau)> 1/3$ then the running time of the heavy algorithm is at most $O(n^2)$, while the running time of the light algorithm is always at least $\Omega(n^2)$. Thus, $\rho(\tau)$ is never larger than $1/3$,
which means that the dominant term in $\Th(n,\rho)$ is $n^{3-3\rho}$, and we can ignore the other two terms.
We thus get:
\begin{align*}
    3-3\rho & = \omega+\frac{\beta}{2}(\rho-\tau)\;,                           \\
    \rho    & =\frac{3-\omega+\beta\tau/2}{3+\beta/2} = 0.19177+0.08363\tau\;.
\end{align*}
Substituting this value of $\rho$ into $3-3\rho$ gives the exponent $2.4241-0.25085\tau$. Thus, the running time of the best of the two algorithms is $O(n^{2.4241}/t^{0.25085}+n^2)$.
This also shows that if $t=n^{\tau_1}$ for $\tau_1=1.690665$, the running time $\tf(n^2)$.

Here, we used simplified bounds for rectangular matrix multiplication, and thus the first value of $t$ for which the running time is faster than $n^\omega$ is not tight. 
We show by a simpler calculation that the running time is faster than $n^\omega$ for every $t\geq n^{\tau_0}$, where $\tau_0=(3-\omega)/3$.

For any $(\rmax,t)$ such that $\rmax\leq \frac{t}{\log^8 n}$ and $t\geq \log^{100} n$, the light algorithm already improves upon the $O(n^\omega)$ bound.
Therefore, we can assume that $\rmax=\tilde{\Theta}(t)$, and thus the heavy algorithm takes time $\tf((n/\rmax)^3)$, obtaining $\tf(n^\omega)$ when $\rmax=\tilde{\Theta}(t)$, and $t=n^{\tau_0}$.
This completes the proof of \Cref{thm:simpler}.

\section{Missing Proofs}\label{app:missing-proofs}
In this appendix, we present two missing proofs.
The first one presents an algorithm to check if a given vertex participates in a diamond in linear time:
\ThmAlgVertexOne* 
The second one takes a set of vertices, and checks if any of them contain a \degc vertex, which is a vertex with induced $P_3$ in its neighborhood, using fast matrix multiplication.
\ThmFindDegc* 

The proof of \Cref{thm:is v in D} is inspired by the characterization of diamond-free graphs as graphs where each neighborhood is $P_3$-free~\cite{kloks2000finding}, and the diamond detection algorithm of~\cite{eisenbrand2004complexity}.

\begin{proof}[Proof of \Cref{thm:is v in D}]
    Given a vertex $v$ we want to check if it participates in an induced diamond.
    The algorithm first checks if the connected components of $G[N(v)]$ the graph induced by $N(v)$ are cliques.
    If not, then $v$ is a \degc vertex, and we are left with finding an explicit $P_3$ in one of the components, which together with $v$ must form a diamond.
    If all components are cliques, then $v$ cannot be a \degc vertex, but it may still be a \degb vertex.  In this case, we search for a vertex $z$ outside $N(v)\cup\{v\}$ that is connected to two neighbors of $v$ in the same component, which would yield a diamond together with $v$.

    We can find the connected components of $G[N(v)]$ in $O(n+m)$ time, as well as verifying whether each component is a clique or not.
    If some component $C_i$ is not a clique, we find an induced $P_3$ in $C_i$ as follows.
    We take any vertex $u\in C_i$ that has degree strictly less than $|C_i|-1$ in $G[C_i]$. Since $C_i$ is a connected component but not a clique, there must exist some vertex in $C_i$ at distance exactly $2$ from $u$. We find such a vertex $z\in C_i$ by running a BFS from $u$ inside $C_i$, which takes $O(m)$ time. This yields a path $u\to w\to z$ for some $w\in C_i$. Note that $u\to w\to z$ is an induced $P_3$ because $u$ and $z$ are at distance $2$ (implying no edge $(u,z)$). We report $(u,v,w,z)$ and halt. This takes $O(n+m)$ time in total, dominated by constructing the subgraph $G[N(v)]$ and traversing it.
    If all connected components are a clique, we check if there exists a vertex $z$ outside $N(v)\cup\{v\}$ that is connected to two vertices in the same connected component in $G[N(v)]$.
    If $z$ has two neighbors $(u_1,u_2)\in C_i$, then $(v,u_1,u_2,z)$ form a diamond with only the edge $(v,z)$ missing.
    Checking this for every vertex $z$ outside $N(v)\cup\{v\}$ can be done in $O(\deg(z))$ time, by iterating over $z$ edges and maintaining a count of how many neighbors of $z$ are in each component $C_i$, stopping as soon as we find a component $C_i$ with at least two neighbors of $z$.
    Iterating over all vertices $z$ outside $N(v)\cup\{v\}$ and performing this check takes $\sum_{v\notin N(v)\cup\{v\}}\deg(v)\leq O(m)$ time, since we iterate over the edges of the graph at most once.
\end{proof}

We move on to the proof of \Cref{thm:find degc}.
The crux of the proof is a modification of the algebraic algorithm of~\cite[Theorem 5.1]{williams2014finding}.
Given a subset of vertices $\Vp$, the algorithm finds all
\degc vertices in time $\tf(\MM{n,n,\abs{\Vp}})$:
\begin{restatable}{theorem}{ThmAlgPartition}
    \label{thm:alg Partition}
    There exists a deterministic algorithm \algPart that, given as input a graph $G$ and a subset of vertices $\Vp\subseteq V(G)$, outputs a subset $L\subseteq \Vp$ containing all vertices in $\Vp$ that are \degc vertices in $G$. Furthermore, for every vertex $v\in \Vp\setminus L$, it outputs a partition of $N(v)\setminus \Vp$ into a disjoint union of cliques. Specifically, for each $v \in \Vp \setminus L$, it outputs $\CC_v =\{ C_1 , C_2, \ldots , C_k \}$, where each $C_i$ induces a clique, and there are no edges between $C_i$ and $C_j$ for $i \neq j$. The running time is $O(\MM{n,n,\abs{\Vp}}\cdot \log^2 n)$.
\end{restatable}
Assuming this theorem holds, proving \Cref{thm:find degc} is straightforward:
\begin{proof}[Proof of \Cref{thm:find degc} using \Cref{thm:alg Partition}]
    If we knew the value of $x_3$ (the number of \degc vertices in $G$), we could sample each vertex independently with probability $p=\frac{10\log n}{x_3}$ to obtain a set $\Vp$. Standard hitting set arguments show that with high probability, $\abs{\Vp}=\TO{n/x_3}$ and $\Vp$ contains at least one \degc vertex. We could then run \algPart on $\Vp$ to find a \degc vertex in time $\TO{\MM{n,n,n/x_3}}$.

    Since we do not know $x_3$, we proceed in phases using a doubling search on the probability. We run $\log n$ instances of the algorithm with vertex sets $V_{p_i}$ obtained by sampling each vertex with probability $p_i=\min(1,\frac{10\log n}{2^i})$, for $i=\log n, \log n - 1, \ldots, 1$ (in this order).

    We execute $\algPart(G,V_{p_i})$ (from \cref{thm:alg Partition}) in each phase until we find a \degc vertex or finish all phases. We stop the execution as soon as a \degc vertex is found.

    Let $i^*$ be the first index in the sequence (i.e., the largest $i$) such that $2^{i^*} \leq x_3$.
    We claim that with high probability, $V_{p_{i^*}}$ contains a \degc vertex.
    The probability that no \degc vertex is sampled in $V_{p_{i^*}}$ is:
    \begin{align*}
        (1-p_{i^*})^{x_3} \le \exp(-p_{i^*}\cdot x_3) = \exp\left(-\frac{10\log n}{2^{i^*}} \cdot x_3\right) \leq \exp(-10\log n) = \frac{1}{n^{10}}\;,
    \end{align*}
    where the last inequality follows because $x_3 \ge 2^{i^*}$.
    Thus, the algorithm detects a diamond with high probability at phase $i^*$ (or earlier). The running time of phase $i^*$ is $O(\MM{n,n, |V_{p_{i^*}}|} \cdot \log n) = O(\MM{n,n, n/2^{i^*}} \cdot \log n)$. Since $2^{i^*} = \Theta(x_3)$, this is bounded by $\TO{\MM{n,n,n/x_3}}$, as required.
\end{proof}

The rest of the appendix is devoted to proving \cref{thm:alg Partition}.

\subsection{The Fast Clustering Algorithm}\label{ssec:algv algeb}
\newcommand{\Inc}{I}
In this subsection, we prove \cref{thm:alg Partition}:
\ThmAlgPartition* 
We provide an algorithm that finds all \degc vertices in a subset of vertices $\Vp\subseteq V$, and for every vertex in $\Vp$ that is not a \degc vertex, it computes a partition of its neighborhood into a union of disjoint cliques, where we refer to such partition as a \emph{clustering}, and say that the clustering is \emph{valid} if it is indeed a partition into a union of disjoint cliques.
The algorithm proceeds in two steps: $(i)$ Compute a clustering for every vertex $v\in \Vp$, using the algorithm \algCluster, and $(ii)$ verify that the clustering is valid using the algorithm \algClusterVeri.
Both steps use fast matrix multiplication and take $O(\MM{n,n,\abs{\Vp}}\cdot \log n)$ time.
The following lemma proves the correctness of \algCluster.
\begin{restatable}{lemma}{LemmaAlgCluster}\label{lemma:alg part}
    There is a deterministic algorithm \algCluster that takes a subset of vertices $\Vp$ and computes a clustering for every $v\in \Vp$ in time $O(\MM{n,n,\abs{\Vp}}\cdot \log n)$. For every $v\in \Vp$, if there is no induced $P_3$ in $N(v)\setminus \Vp$, the computed clustering is valid.
\end{restatable}

A subtle point here is that the computed clustering is only of the subsets $N(v) \setminus \Vp$ rather than the full neighborhood $N(v)$. This is due to technical difficulties and might pose issues in detecting diamonds that contain more than one vertex from $\Vp$. We will soon see how to overcome this issue.

The verification algorithm is defined as follows:
\begin{restatable}{lemma}{LemmaAlgClusterVeri}\label{lemma:alg partveri}
    There is a deterministic algorithm \algClusterVeri that takes a subset of vertices $\Vp$ and a partition $\CC_v$ of the neighborhood of every vertex $v\in \Vp$.
    The algorithm outputs for every $v\in \Vp$ whether the clustering is valid or not, in time $O(\MM{n,n,\abs{\Vp}}\cdot \log n)$.
\end{restatable}

We now address the issue of diamonds having multiple vertices in $\Vp$. Dealing with this can be done using randomness by taking a poly-logarithmic number of subsamples of $\Vp$, ensuring that every diamond is ``separated'' (i.e., a \degc vertex $v$ is in the subsample, but the $P_3$ is outside).
Achieving this separation deterministically is well-known in the literature:
\begin{restatable}[{\cite{naor1990small}}]{lemma}{lemmaThreeDishunct}\label{lemma:3 disjunct}
    Let $U\subseteq V$.
    There exists a family $\FF=(F_1,\ldots, F_\ell)$ of size $\ell=O(\log n)$, such that for every four distinct vertices $a,b,c,d\in U$, there exists some set $F_i\in \FF$ where $a\in F_i$ and $\{b,c,d\} \cap F_i = \emptyset$.
    This family can be constructed explicitly in $\To(n)$ time.
\end{restatable}
The family $\FF$ in \Cref{lemma:3 disjunct} is a special case of $(n,k)$-universal sets (here with $k=4$). Constructions of $\FF$ with $|\FF|=O(\log n)$ are given in \cite{naor1990small,alon1992simple,NaorSS95,alon2002construction}.
Using \Cref{lemma:alg part,lemma:alg partveri,lemma:3 disjunct}, we can prove \Cref{thm:alg Partition}.

\begin{proof}[Proof of \Cref{thm:alg Partition}]
    Let $\FF=(F_1,\ldots, F_\ell)$ be the family from \cref{lemma:3 disjunct} defined over the universe $\Vp$, with $\ell=O(\log n)$.
    For each set $F_i \in \FF$, we execute \algCluster followed by \algClusterVeri on the input set $F_i$. Let $L_i \subseteq F_i$ denote the set of vertices for which the verification fails. By \Cref{lemma:alg part,lemma:alg partveri}, every vertex in $L_i$ is a \degc vertex in $G$.

    We return the union $L' = \bigcup_{i=1}^\ell L_i$.
    We claim that $L'$ contains exactly the set of all \degc vertices in $\Vp$.

    First, if a vertex $v \in \Vp$ is not a \degc vertex, its neighborhood $N(v)$ is a union of disjoint cliques. Consequently, for any $F_i$, the subset $N(v) \setminus F_i$ is also a union of disjoint cliques. By the correctness of \algCluster and \algClusterVeri, $v$ will never be flagged as \degc vertex, so $v \notin L'$.

    Conversely, consider a \degc vertex $v\in \Vp$. By definition, there exists an induced $P_3$, denoted $(a,b,c)$, in $N(v)$.
    By \cref{lemma:3 disjunct}, there exists some set $F_i \in \FF$ such that $v \in F_i$ and $\{a,b,c\} \cap F_i = \emptyset$.
    This implies that $a,b,c \in N(v) \setminus F_i$. Thus, $(a,b,c)$ remains an induced $P_3$ in the graph visible to the algorithm in iteration $i$.
    Specifically, $N(v) \setminus F_i$ cannot be partitioned into disjoint cliques. Therefore, the valid clustering check for $v$ must fail, and $v$ will be included in $L_i \subseteq L'$.

    The total running time is $\sum_{i=1}^{\ell} O(\MM{n,n, |F_i|} \log n) = O(\ell \cdot \MM{n,n,\abs{\Vp}} \log n) = O(\MM{n,n,\abs{\Vp}}\cdot \log^2 n)$.
\end{proof}
Before we explain how to implement \algCluster and \algClusterVeri, we define an auxiliary directed graph $H$ that we work on.
\paragraph{Auxiliary Directed Graph $H$.}
The graph $H$ contains two vertex sets $V_1$ and $V_2$, where $V_1$ is a copy of $\Vp$, and $V_2$ is a copy of $V\setminus \Vp$. For every edge $(u,v)\in E$ between two vertices in $V\setminus \Vp$, we add two directed edges $(u,v)$ and $(v,u)$. For every edge $(u,v)\in E$ where $u\in \Vp$ and $v\notin \Vp$, we add a directed edge $(u,v)$ from $V_1$ to $V_2$.
This completes the description of the directed graph $H$.
Note that $H[V_1]$ is an independent set, $H[V_2]$ is isomorphic to $G[V\setminus \Vp]$, and the edges in $E(V_1,V_2)$ correspond to edges $E_G(\Vp,V\setminus \Vp)$.
The vertices of $V_1$ are the vertices for which we want to compute a clustering of their neighborhood in $H$.

~\\Instead of using one adjacency matrix of size $n\times n$, we use one adjacency matrix $A$ of size $n\times n$ to encode the edges in $H[V_2]$ and one incidence rectangular matrix $\Inc$ of size $\abs{V_1}\times \abs{V_2}$ to encode the edges in $E(V_1,V_2)$, where $\Inc[i,j]=1$ if and only if the $i$-th vertex in $V_1$ is connected to the $j$-th vertex in $V_2$.

We note that we remove $\Vp$ from $V_2$ to avoid the detection of induced diamonds that contain two copies of the same vertex from $\Vp$:
If $(a,b,c)$ is a triangle and $a$ is sampled to $\Vp$, then the four vertices $\set{a,b,c,a'}$, where $a'$ is the copy of $a$ in $V_2$, form an induced diamond in $H$.

\paragraph{Intuition For Clustering.}
Consider a vertex $v\in V_1$ that is not a \degc vertex.
Then, its neighborhood $N(v)$ can be partitioned into a union of disjoint cliques $(C_1,\ldots, C_k)$.
Let $u$ be some vertex in $N(v)\cap C_1$.
Then, the number of paths of length at most two from $v$ to $u$ is exactly $\abs{C_1}$.
If all cliques have different sizes, then by computing the number of paths of length at most two from $v$ to every vertex in $N(v)$, we can recover the clustering of $N(v)$.
To deal with the general case, we iteratively remove edges from $E(v,N(v))$, until there is only a single edge between $v$ and every clique in the clustering of $N(v)$.

After computing this initial clustering, we need to verify that it is valid.
In other words, we proceed as if all vertices in $V_1$ are not \degc vertices, and check whether this assumption is correct.

\subsubsection{The Clustering Algorithm \algCluster}
The algorithm starts with a directed graph $H_0=H$ and incidence matrix $\Inc_0=\Inc$, and proceeds iteratively in phases. We denote the directed graph obtained after $i$ phases by $H_i$ and its incidence matrix by $\Inc_i$, where $H_i$ is a subgraph of $H_{i-1}$ and $\Inc_i$ is obtained from $\Inc_{i-1}$ by removing some of its $1$ entries.
We never remove edges from $H[V_2]$, and therefore the adjacency matrix $A$ of the graph $H[V_2]$ remains unchanged throughout the algorithm.
We use $\CC_v^0 = \set{N(v)}$ to denote the initial clustering, and $\CC_v^i$ to denote the clustering after the $i$-th phase.
The clustering $\CC_v^i$ is a refinement of $\CC_v^{i-1}$, meaning that every cluster in $\CC_v^i$ is a (possibly trivial) subset of some cluster in $\CC_v^{i-1}$.

Each phase has three steps:
\begin{enumerate}
    \item Removing some edges from $E(V_1,V_2)$.
    \item Refining the clustering of every $v \in V_1$.
    \item Restoring some of the edges that were removed in the first step of the current phase.
\end{enumerate}
The algorithm always runs for $r=2\log n$ phases.
\paragraph{Step 1: Removing Edges.}
See \Cref{alg:remove edges} for the pseudocode.
At the beginning of the $i$-th phase, the input is a directed graph $H_{i-1}$ encoded using an incidence matrix $\Inc_{i-1}$, the adjacency matrix $A$, and a clustering $\set{\CC_v^{i-1}}_{v\in V_1}$.
The algorithm removes edges as follows.
For every vertex $v \in V_1$, and every cluster $C \in \CC_v^{i-1}$, remove half of the edges from $E(v,C)$, where $E(v,C)$ is the set of edges between $v$ and $C$, w.r.t. to the graph $H_{i-1}$.
If there is only one edge between $v$ and $C$, we do not remove it.
The choice of the edges to remove is arbitrary.
Denote the resulting graph by $H_{i}'$ and its incidence matrix by $\Inc_{i}'$.
We note that $H_i'$ is not the final graph of the $i$-th phase.
Some of the edges that were removed might be restored in the third step of the current phase.

\begin{algorithm}[H]
    \caption{Removing Edges (Step 1)}
    \label{alg:remove edges}
    \KwIn{The graph $H_{i-1}$ and clustering $\set{\CC_v^{i-1}}_{v\in V_1}$.}
    \KwOut{A subgraph $H_i'$ of $H_{i-1}$.}

    \setcounter{AlgoLine}{0}
    \medskip
    \For{$v\in V_1$}{
        \For{$C\in \CC_v^{i-1}$}{
            Remove half of the edges from $E(v,C)$;
        }
    }
    Output the resulting graph $H_i'$;
\end{algorithm}
\paragraph{Step 2: Refining the Clustering.}
See \Cref{alg:refine Cluster} for the pseudocode.
The input consists of the graph $H_i'$ with incidence matrix $\Inc_i'$, the adjacency matrix $A$, and the clustering $\set{\CC_v^{i-1}}_{v\in V_1}$.

The algorithm computes the matrix $R_i = \Inc_i' + \Inc_i'\cdot A$.
For each vertex $v \in V_1$, we refine the clustering $\CC_v^{i-1}$ using $R_i$ as follows.
Fix a cluster $C \in \CC_v^{i-1}$. We partition $C$ based on the values $R_i[v,u]$. Let $C(y)$ denote the subset of vertices $u \in C$ where $R_i[v,u] = y$:
\[ C(y)=\set{u\in C\mid R_i[v,u] = y}. \]
The new clustering $\CC_v^{i}$ is formed by collecting all non-empty sets $C(y)$, except for $C(0)$.

\paragraph{Step 3: Restoring Edges.}
We restore the edges between $v$ and every vertex in the cluster corresponding to $0$. Specifically, for every $u \in C(0)$, we set $\Inc_i[v,u] = 1$.
We refer to the resulting graph as $H_i$ and to its adjacency matrix as $\Inc_i$.

~\\This completes the description of all three steps of the algorithm.

\begin{algorithm}[H]
    \caption{Refine Cluster and Restore Edges (Steps 2 and 3)}
    \label{alg:refine Cluster}
    \KwIn{A graph $H_i'$ with matrices $\Inc_i',A$, and clustering $\set{\CC_v^{i-1}}_{v\in V_1}$.}
    \KwOut{A graph $H_i$ with adjacency matrix $\Inc_i$ and a refined clustering.}

    \setcounter{AlgoLine}{0}
    \medskip
    $\Inc_i\gets \Inc_i'$;\\
    \lFor{$v\in V_1$}{$\CC_v^{i}\gets \emptyset$ }
    $R_i \gets \Inc_i' + \Inc_i'\cdot A$;\\
    \For{$v\in V_1$}{
    \For{$C\in \CC_v^{i-1}$}{
    \For{$u\in C$}{
    $y\gets R_i[v,u]$; \\
    Add $u$ to a (new) cluster $C(y)$. \Comment{$C(y)=\set{u\in C\mid R_i[v,u] = y}$,}
    }
    Add every non-empty new cluster $C(y)$ with $y\neq 0$ to $\CC_v^{i}$. \\
    \lFor{$u\in C(0)$}{ $\Inc_i[v,u] \gets 1$} \Comment{Restore every edge between $v$ and $u\in C(0)$.}
    }
    }
    \Return{$\set{\CC_v^{i}}_{v\in V_1}$;}
\end{algorithm}

~\\The following claim summarizes the running time of the algorithm.
\begin{claim}\label{claim:running time}
    The algorithm runs in time $O(\MM{n,n,\abs{\Vp}}\cdot r)$.
\end{claim}
\begin{proof}[Proof of \Cref{claim:running time}]
    Removing half of the edges takes $O(m)$ time.
    Computing the matrix $R_i$ takes $O(\MM{\abs{\Vp},n,n})$ time.
    We show that refining the clusters for all $v \in V_1$ takes $O(n^2 \log n)$ time.

    Fix $v \in V_1$ and a cluster $C \in \CC_v^{i-1}$. We refine $C$ in time $O(\abs{C} \log n)$.
    Iterate over all $u \in C$ and let $y = R_i[v,u]$.
    We maintain the new clusters $C(y)$ in a balanced binary search tree indexed by $y$.
    For each $u$, we search for $y$; if it exists, we insert $u$ into $C(y)$, otherwise we create a new node for $y$.
    Since the search and insertion take $O(\log n)$, refining $C$ takes $O(\abs{C} \log n)$ time.

    Summing over all clusters in $\CC_v^{i-1}$, the time for a fixed $v$ is $O(n \log n)$ as the clusters form a partition of $N(v)$, and therefore $\sum_{C \in \CC_v^{i-1}} \abs{C} = \abs{N(v)} \leq n$.
    Summing over all $v \in V_1$, the total refinement time is $O(n^2 \log n)$.
\end{proof}

A key property of the output of the algorithm is stated in the following claim.
\begin{claim}\label{claim:reps}
    In the graph $H_r$, there is exactly one directed edge between every vertex $v \in V_1$ and every cluster $C \in \CC_v$.
\end{claim}
We assume that every vertex $v \in V_1$ has at least one neighbor in $V_2$, otherwise it is not a \degc vertex and its clustering is trivial.
\begin{proof}[Proof of \Cref{claim:reps}]
    We follow the proof of \cite[Lemma 5.2]{williams2014finding}.
    We prove the claim by induction on the phase number $i$.
    The invariant is that after phase $i$, for every vertex $v \in V_1$ and every cluster $C \in \CC_v^i$, the number of edges directed from $v$ to $C$ is at most $n/2^i$.
    This holds trivially for the base case $i=0$.

    Consider the $i$-th phase. Fix a vertex $v$ and a cluster $C \in \CC_v^{i-1}$.
    Let $E_{total}$ denote the set of edges from $v$ to $C$ at the beginning of the phase. By the induction hypothesis, $\abs{E_{total}} \le n/2^{i-1}$.
    The algorithm partitions $E_{total}$ into two sets: $E_{keep}$ (edges retained in the first step) and $E_{cut}$ (edges removed in the first step).
    We have $\abs{E_{keep}} \le \abs{E_{total}}/2 \le n/2^i$ and similarly $\abs{E_{cut}} \le n/2^i$.

    The clustering is refined into new clusters $C(y)$. We bound the number of edges incident to a new cluster $C' = C(y)$ according to two cases:
    \begin{itemize}
        \item \textbf{Case $y \neq 0$:} The algorithm retains only the edges from $E_{keep}$. Thus, the edges from $v$ to $C'$ are a subset of $E_{keep}$. The number of edges is at most $\abs{E_{keep}} \le n/2^i$.
        \item \textbf{Case $y = 0$:} The algorithm restores the edges that were removed, meaning the edges from $v$ to $C'$ are exactly those in $E_{cut}$ (restricted to the vertices in $C'$). Thus, the number of edges is at most $\abs{E_{cut}} \le n/2^i$.
    \end{itemize}
    In both cases, the number of edges from $v$ to the new cluster $C'$ satisfies the invariant.
\end{proof}

We are ready to prove \Cref{lemma:alg part}, which we restate here for convenience:
\LemmaAlgCluster* 

\begin{proof}[Proof of \Cref{lemma:alg part}]
    \renewcommand{\DD}{\mathcal{P}}
    The running time follows from \Cref{claim:running time}.
    We prove the correctness.
    We need to show that for every $v\in \Vp$ that has no induced $P_3$ in $N(v)\setminus \Vp$, the clustering $\CC_v^r$ is valid.
    Fix such a vertex $v$, and let $\DD=(D_1,\ldots, D_k)$ be the partition of $N(v)$ into vertex-disjoint cliques.
        We show that the algorithm produces a clustering $\CC_v^r=(C_1,\ldots, C_k)$ which is equal to $\DD$
        up to renaming; there exists a permutation $\pi:[k]\to[k]$ such that $C_i = D_{\pi(i)}$ for every $i\in[k]$.
        Let $D$ be a maximal clique in $\DD$, with no edges between $D$ and any vertex outside $D$.
    We prove two claims on $D$:
    \begin{enumerate}
        \item For any $i\in [r]$, there exists $C\in \CC_v^i$ such that $D\subseteq C$. That is, $D$ is contained in some cluster $C \in \CC_v^i$.
        \item Fix two distinct cliques $D,D'\in \DD$. Then $D$ and $D'$ are not in the same cluster in the final partition $\CC_v^r$.
    \end{enumerate}
    From both claims, we conclude that there exists a cluster $C\in \CC_v^r$ such that $C\subseteq D\subseteq C$, meaning that $D=C$.

    \paragraph{Proof of (1).}
    Recall that a cluster $C\in \CC_v^i$ is partitioned to sub clusters based on the value of $R_i[v,u]$ for every vertex $u \in C$, where $R_i[v,u]$ counts the number of paths of length one or two from $v$ to $u$ in the graph $H_i'$.
    For every two vertices $a,b\in D$, the number of such paths is equal to the number of edges between $v$ and $D$ in the graph $H_i'$.
    Since no two vertices in $D$ obtain different values of $R_i[v,u]$, they are always in the same sub-cluster.

    \paragraph{Proof of (2).}
    Consider the first iteration $i$ where $D$ is in some cluster $C$ where there is exactly one edge between $C$ and $v$.
    Such an iteration exists, since $D$ is always in some cluster $C$ by (1). Moreover, by \Cref{claim:reps} after $\log n$ iterations, there is a single edge between $v$ and every cluster $C$ in its neighborhood.
    If in this iteration $C=D$ then we are done, since $D$ is its own cluster.
    Otherwise, there is another maximal clique $D'\in \DD$, which is also in $C$. But there is only one edge between $C$ and $v$, assume that this edge is $(v,u)$.
    If $u\in D$, then the number of paths of length one or two from $v$ to $w$ for every vertex $w\in D$ is $1$, while the number of such paths between $v$ and every vertex in $D'$ is $0$. Therefore, in the next phase, $D$ and $D'$ are separated into different clusters.
    If $u\notin D$, then we get the same conclusion, by symmetry.
\end{proof}
This completes the proof of \Cref{lemma:alg part}. We proceed to explain how the verification algorithm \algClusterVeri works, and prove \Cref{lemma:alg partveri}.

\subsubsection{The Verification Algorithm \algClusterVeri}
The algorithm \algClusterVeri first constructs the auxiliary graph $H$ and its incidence matrix $\Inc$. Then, it takes as input a clustering $\set{\CC_v^r}_{v\in V_1}$ and verifies whether the clustering of each vertex is valid.
That is, for every vertex $v \in V_1$, it checks:
\begin{enumerate}
    \item that there are no inter-cluster edges, i.e., no edges between two distinct clusters in $\CC_v^r$, and
    \item that every cluster $C\in \CC_v^r$ is a clique.
\end{enumerate}
We restate \Cref{lemma:alg partveri} for convenience:
\LemmaAlgClusterVeri*

The first step in the verification is to pick a \emph{leader} for every cluster $C\in \CC_v^r$, defined as the unique neighbor of $v$ in $H_r$ inside $C$ (such a vertex exists by \Cref{claim:reps}).
We then check that the leader is connected to every other vertex in its cluster,
which takes $O(\abs{C})$ time per cluster, summing to $O(n)$ for $v$ and to $O(n^2)$ in total. We proceed to verify the two properties of a valid clustering,
using the following claims.
\begin{claim}[Verify Inter-Cluster Edges]\label{claim:verify inter-cluster}
    There is a deterministic algorithm that takes a subset of vertices $\Vp$ and a partition $\CC_v$ of the neighborhood of every vertex $v\in \Vp$.
    The algorithm outputs for every $v\in \Vp$ whether the clustering contains inter-cluster edges or not, taking $O(\MM{n,n,\abs{\Vp}}\cdot \log n)$ time.
\end{claim}

\begin{claim}[Verify Clique]\label{claim:verify clique}
    There is a deterministic algorithm that takes a subset of vertices $\Vp$ and a partition $\CC_v$ of the neighborhood of every vertex $v\in \Vp$.
    The algorithm outputs for every $v\in \Vp$ whether each cluster in the clustering is a clique, taking $O(\MM{n,n,\abs{\Vp}})$ time.
\end{claim}
\begin{proof}[Proof of \Cref{lemma:alg partveri}]
    We run the two verification algorithms from \Cref{claim:verify inter-cluster,claim:verify clique} in sequence.
    For every vertex $v\in \Vp$, if either of the two algorithms reports a violation, then we report $v$ as a \degc vertex.
    Otherwise, we accept the clustering as valid.
    The correctness and running time follows from the two claims.
\end{proof}

We first provide the first step and notation used in both verification algorithms.
Recall that for every vertex $v\in \Vp$, we have a partition $\CC_v$ of its neighborhood $N(v)\setminus \Vp$ into clusters.
For every cluster $C\in \CC_v$, a leader is picked, followed by verification that the leader is connected to every other vertex in $C$.
Then, define an auxiliary directed graph $H=((\Vp,V\setminus \Vp),E_H)$, where for every $v\in \Vp$ we keep only the edges between $v$ and its cluster leaders, directed from $v$ to the leader. Let $\Inc$ be the incidence matrix of the directed edges between $\Vp$ and $V\setminus \Vp$ in $H$. We also keep the adjacency matrix $A$ of the graph $G[V\setminus \Vp]$ (each edge is directed in both directions).

We are now ready to prove \Cref{claim:verify inter-cluster}.

\begin{proof}[Proof of \Cref{claim:verify inter-cluster}]
    We explain how to verify that there are no edges between clusters in the partition $\CC_v$, i.e., inter-cluster edges.
    For every vertex $v\in \Vp$, write $\CC_v=\{C_1,\ldots,C_k\}$.
    For each $h\in[\log n]$, we say that a cluster $C_i\in\CC_v$ is \emph{active} in iteration $h$ if the binary representation of $i$ has a $1$ in the $h$-th bit, and \emph{inactive} in iteration $h$ otherwise.
    Let $B_h$ be the $\abs{\Vp}\times\abs{V\setminus\Vp}$ matrix where $B_h[v,u]=1$ if and only if $u$ belongs to a cluster of $\CC_v$ that is active in iteration $h$.
    We compute $M_h=B_h\cdot A$, and check if there exists a vertex $u\in N(v)\setminus\Vp$ such that $B_h[v,u]=0$ and $M_h[v,u]>0$.
    If so, the verification fails and we report that the clustering of $v$ contains an inter-cluster edge.
    The correctness follows because if there is an edge $(a,b)$ between two distinct clusters $C_i,C_j\in\CC_v$, then there exists a bit $h$ that separates $i$ and $j$.
    Assume without loss of generality that $C_i$ is active in iteration $h$ and $C_j$ is inactive in iteration $h$.
    Then $B_h[v,a]=1$, $A[a,b]=1$, and $B_h[v,b]=0$, hence $M_h[v,b]>0$, so the violation is detected.
    Conversely, if the test reports a violation for some $v,h,u$, then $B_h[v,u]=0$ and $M_h[v,u]>0$.
    Therefore, there exists a vertex $a$ in a cluster of $\CC_v$ that is active in iteration $h$ such that $(a,u)$ is an edge in $G[V\setminus\Vp]$.
    Since $u$ belongs to a cluster that is inactive in iteration $h$, this edge is between two distinct clusters of $\CC_v$.
    The running time is one matrix multiplication for every $h\in[\log n]$, each taking $O(\MM{\abs{\Vp},n,n})$ time.
    Thus the total running time is $O(\MM{n,n,\abs{\Vp}}\cdot \log n)$.
\end{proof}
We are now ready to prove \Cref{claim:verify clique}. We emphasize that this step assumes that there are no inter-cluster edges, as verified in the previous step.
\begin{proof}[Proof of \Cref{claim:verify clique}]
    Let $B$ be the $\abs{\Vp}\times\abs{V\setminus\Vp}$ matrix where $B[v,u]=1$ if and only if $u\in N(v)\setminus\Vp$.
    We compute $D=B\cdot A$, where $A$ is the adjacency matrix of $G[V\setminus\Vp]$.
    For every $v\in\Vp$, every cluster $C\in\CC_v$, and every vertex $u\in C$, we check that $D[v,u]=\abs{C}-1$.
    If one of these checks fails, then we report that the cluster $C$ is not a clique.
    The correctness follows because $D[v,u]$ is the number of neighbors of $u$ inside $N(v)\setminus\Vp$.
    Since there are no inter-cluster edges, all these neighbors belong to the same cluster as $u$.
    Therefore, for $u\in C$, the equality $D[v,u]=\abs{C}-1$ holds if and only if $u$ is adjacent to every other vertex in $C$.
    It follows that all checks pass if and only if every cluster is a clique.
    The running time is dominated by computing $D=B\cdot A$, which takes $O(\MM{\abs{\Vp},n,n})=O(\MM{n,n,\abs{\Vp}})$ time.
    The additional time needed to compute the cluster sizes and perform the checks is $O(n\abs{\Vp})$, so the total running time is $O(\MM{n,n,\abs{\Vp}})$.
\end{proof}

This completes the clustering and verification algorithms.

\bibliographystyle{alpha}
\bibliography{refs.bib}

\end{document}